\documentclass[11pt]{article}
\usepackage{bbm}
\newcommand{\lt}{\left}
\newcommand{\rt}{\right}

\makeatletter
\newcommand{\Biggg}{\bBigg@{4}}
\newcommand{\Bigggg}{\bBigg@{5}}
\makeatother

\def\bA{{\bf A}}

\def\bc{{\bf c}}
\def\bD{{\bf D}}
\def\bd{{\bf d}}

\def\be{{\bf e}}

\def\bI{{\bf I}}

\def\bL{{\bf L}}

\def\bM{{\bf M}}

\def\bO{{\bf O}}

\def\bu{{\bf u}}

\def\bv{{\bf v}}

\def\bx{{\bf x}}

\def\by{{\bf y}}

\def\bz{{\bf z}}
\def\bbeta{\boldsymbol{\beta}}

\def\blambda{\boldsymbol{\lambda}}

\def\bpi{\boldsymbol{\pi}}

\def\bphi{\boldsymbol{\phi}}

\def\bel{\boldsymbol{\ell}}
\def\blambN{\boldsymbol{\hat\pi}^\brLM}
\def\cB{{\cal B}}

\def\cI{{\cal I}}
\def\cJ{{\cal J}}

\def\cQ{{\cal Q}}

\def\cS{{\cal S}}

\def\cZ{{\cal Z}}

\def\bbR{\mathbb{R}}

\def\bzero{{\bf 0}}
\def\bone{{\bf 1}}

\def\xst{\bar{x}}				% x star, the single-period size for non-interconnected systems
\def\bxst{\bar{{\bf x}}}				% bold x star, the vector of x star
\def\xbar{\bar{x}^\brLM}				% x bar, the single-period size for interconnected systems
\def\bxbar{\bar{{\bf x}}^\brLM}				% bold x bar, the vector of x bar
\def\yf{\psi}				% psi function, function for second-period size w.r.t. first-period size
\def\byf{\boldsymbol{\psi}}
\def\yfpm{\psi^\gamma}      % function for second-period size w.r.t. first-period size under permanent impact, 
\def\xtld{x^*}

\def\xpm{x^\gamma}
\def\xfe{\tilde{x}}
\def\spm{s^\gamma}

\def\lf{\phi}
\def\lfpm{\Phi^\gamma}      % the robust objective with permanent price impact
\def\lpm{\phi^\gamma}               % the original objective with permanent price impact

\def\bgao{\be}
\def\btp{\beta^+}				% beta plus, the maximum financial connectivity
\def\ztot{z_{tot}}				% z total, the sum of all holding sizes of illiquid assets
\def\pst{\bar{p}}				% p bar, the single-period price for non-interconnected systems
\def\ppm{p^\gamma}              % p gamma, the second-period price under permanent price impact
\def\psA{p_2^\brLM}             % p_2^\brLM, the second-period price under non-zero relative liability matrix
\def\pbar{\bar{p}^\brLM}				% p bar N, the single-period price for interconnected systems
\def\lbar{\bar{\pi}^\brLM}				% l bar, the clearing payment for interconnected systems
\def\blbar{\bar{\bpi}^\brLM}
\def\ltN{\Phi^\brLM}

\def\lrLM{\lf^\brLM}

\def\cNf{c}
\def\bcNf{{\bf c}}				% c N vector, the operating cash after liquidation in the interconnected systems

\def\brpA{\bpi^\brLM}
\def\wNf{\pi}
\def\bwNf{\bpi}				% w N vector, the clearing vector in the interconnected systems
\def\rhoN{\hat p^\brLM}
\def\lambN{\hat\pi^\brLM}
\def\kapN{\hat x^\brLM}
\def\bkapN{\hat\bx^\brLM}

\def\sbar{\bar{s}^\brLM}

\def\lcQ{\hat{\ell}}
\def\yfQ{\hat{\psi}}
\def\rLM{A}
\def\brLM{\bA}
\def\argmin{\mathop{\arg\min}}
\def\argmax{\mathop{\arg\max}}
\def\Zset{\cZ}
\def\Bset{\cB}
\def\cvec{\bc}
\def\ctl{\tilde{c}}

\def\bctl{\tilde{\bc}}

\def\tcl{\zeta}
\def\gtld{\tilde{g}}
\def\clst{\zeta^*}
\def\tlb{\tilde{\zeta}}

\def\bcht{\hat{\bc}}
\def\cht{\hat{c}}

\def\clbar{\bar{\zeta}}

\def\hlf{\mu}

\def\hrt{\nu}

\def\clht{\hat{\tcl}}
\def\etdlt{\Delta}
\def\bvst{\bv^*}
\def\Qst{Q^*}
\def\phat{\hat{p}}

\def\vthst{\vartheta^*}
\def\vthtl{\tilde{\vartheta}}
\def\rcbrLM{\hat\bA}            % used as the reconstruction matrix

\def\ettl{\tilde{\eta}}
\def\etst{\eta^*}
\def\lid{l}
\def\qf{q}

\def\sst{\bar{s}}				% s star, the sum of x star of other banks
\def\Tset{\cI}				% used in the proof of Lemma 4.1
\def\Smin{\cS}				% the set of argmin of g function w.r.t. s

\usepackage[hidelinks]{hyperref}
 \usepackage{amsmath}
 \usepackage{amsfonts}
 \usepackage{amssymb}
 \usepackage{amsbsy}
\usepackage{mathtools}
\usepackage{amsthm}
\newtheoremstyle{sfstyle}
  {\topsep} % Space above
  {\topsep} % Space below
  {\itshape} % Body font (or \rmfamily)
  {} % Indent
  {\sffamily\bfseries} % Title font (sans-serif, bold)
  {.} % Separator after title
  {1em} % Space after title
  {} % Title spec (e.g., \thetheorem)
\theoremstyle{sfstyle}
\newtheorem{theorem}{Theorem}
\newtheorem{lemma}{Lemma}
\newtheorem{proposition}{Proposition}

\newtheoremstyle{sfstyl}
  {\topsep} % Space above
  {\topsep} % Space below
  {} % Body font (or \rmfamily)
  {} % Indent
  {\sffamily\bfseries} % Title font (sans-serif, bold)
  {.} % Separator after title
  {1em} % Space after title
  {} % Title spec (e.g., \thetheorem)
\theoremstyle{sfstyl}

\newtheorem{assumption}{Assumption}

\newtheorem{remark}{Remark}
\usepackage{pgfplots}
\usepackage{pgfplotstable}
\usepackage[draft]{changes}
\usepackage[normalem]{ulem}
\usepackage{enumitem}
\usepackage{multirow}
\usepackage{subcaption}
\usepackage{dsfont}
\usepackage{stackengine}
\usepackage{lipsum}
\usepackage{natbib}
\mathtoolsset{showonlyrefs}
\usepackage{booktabs}
\usepackage{titlesec}
\titleformat{\section}{\sffamily\large\bfseries}{\thesection.}{1em}{}
\titleformat{\subsection}{\sffamily\bfseries}{\thesubsection.}{1em}{}
\titlespacing*{\section}{0pt}{.7\baselineskip}{.3\baselineskip}
\titlespacing*{\subsection}{0pt}{.4\baselineskip}{.1\baselineskip}

\usepackage{enumitem}
\usepackage[margin=1in]{geometry}
\usepackage{setspace}
\allowdisplaybreaks

\title{\sf\textbf{Robust Optimal Strategies for Early Liquidation\\ in Financial Systems}}
\author{Dohyun Ahn\thanks{Department of Systems Engineering and Engineering Management, E-mail: \href{mailto:dohyun.ahn@cuhk.edu.hk}{\tt dohyun.ahn@cuhk.edu.hk}},~~~Hongyi Jiang\thanks{Department of Systems Engineering and Engineering Management, E-mail: \href{hyjiang@link.cuhk.edu.hk}{\tt hyjiang@link.cuhk.edu.hk}}\\
{\small \it The Chinese University of Hong Kong}
}
\date{March 2026}

\begin{document}

\maketitle
\begin{abstract}
We study the problem of asset liquidation in financial systems. During financial crises, asset liquidation is often inevitable but can lead to substantial losses if a significant amount of illiquid assets are sold simultaneously at depressed prices---a phenomenon known as \emph{price impact}. To tackle this challenge, we consider a two-period liquidation model that allows for early liquidation prior to clearing, thereby mitigating price impact at clearing, and we develop a worst-case approach to solve the decision-making problem on the optimal size of early liquidation. Specifically, we propose a \emph{robust optimal strategy}---a tractable liquidation approach that maximizes the worst-case value of liquid assets at clearing, taking into account the uncertainty of other banks’ early liquidation decisions. We derive a (semi-)closed-form representation of this strategy in a practical scenario involving permanent price impact and analyze its sensitivity to that impact's magnitude. We further identify its closed-form expression in another practical scenario featuring interbank exposures.
Our findings, although built upon a stylized model, offer valuable guidelines for developing robust liquidation strategies that mitigate losses resulting from asset liquidation.

\end{abstract}

\section{Introduction}\label{sec:intro}
	
    The global financial crisis in 2007-2008 provoked intense discussions on the factors that threatened the stability of the financial system. In the majority of relevant studies, two channels of risk contagion have been widely acknowledged as significant factors contributing to this instability. One such channel involves direct debt exposures among financial institutions. When an institution fails to meet debt obligations due to adverse economic shocks, its creditor institutions experience financial distress due to the resulting repayment shortfall \citep{EisenbergNoe2001}. Furthermore, during the crisis of asset erosion, financial institutions with inadequate liquid assets are forced to sell off substantial illiquid assets to satisfy their debt claims. Such liquidation overwhelms shallow market demand and triggers sharp declines in asset prices. This phenomenon, termed \emph{price impact}, imposes significant losses on institutions holding the same illiquid assets, even without direct debt exposures among them. These fire sales constitute another channel of risk transmission and are recognized as a salient ingredient of the crisis \citep{ClercEtAl2016}.
    Although the first channel has been studied more extensively in the literature, the second channel has gained increasing attention in recent years, particularly as financial institutions have actively reduced their debt exposures to other banks since the said crisis.

    Our study departs from a conventional assumption in the existing literature that liquidation is conducted only at clearing. This assumption stems from the observation that liquidity-constrained institutions tend to favor borrowing over asset sales to avoid the prohibitive costs of liquidation. However, since borrowing capacity is not always guaranteed, relying solely on this option exposes banks to future funding shortages, particularly during systemic crises when credit conditions deteriorate \citep{AcharyaEtAl2011}. %is analytically convenient but can be restrictive, as it excludes liquidation behaviors that may better align with institutions' fundraising incentives. For instance, when many institutions anticipate substantial future repayment pressure, simultaneous asset sales at clearing can easily exceed the market's absorption capacity. In the congested sale market, competition among institutions compels them to accept deeply depressed prices to convert their assets into cash, which significantly elevates their funding costs.
    In this case, preemptive liquidation prior to clearing becomes a rational solution. %\footnote{In practice, institutions may postpone the repayments to maturing debts by borrowing new debts instead of liquidation, which is termed debt rollover. However, the rollover is not always guaranteed because of the rollover risk \citep{AcharyaEtAl2011}. Here we focus on the early liquidation problem without elaborating on debt rollover. One may refer to \citep{BichuchFeinstein2019} for the study of decision-making with both borrowing and liquidation costs when debt rollover is always allowed.}
    Then, this naturally raises a critical question: \emph{what volume of illiquid assets should an institution sell early to ensure its solvency?} The answer is not immediately evident as liquidation outcomes are driven by the interplay of institutional liquidation decisions that jointly determine sale prices and are constrained by repayment requirements.

	This study is the first to propose a solution to address the aforementioned issue, particularly focusing on two distinct practical settings: permanent price impact induced by early liquidation and direct loss contagion through interbank exposures.
    % In this work, we study the early liquidation decision of an individual bank, particularly focusing on the funding sufficiency for debt repayment. Moving beyond the common assumption that liquidation occurs only at clearing, we develop a stylized model that endogenizes the interdependence of liquidation volume allocations over time. We take a worst-case approach to identify a bank's optimal liquidation strategy against the uncertainty of prices impacts, which stems from the lack of knowledge of other banks' strategies. In particular, our approach accommodates two distinct practical settings: one featuring permanent price impact induced by early liquidation, and the other capturing contagion effects through interbank liabilities.
	  Specifically, our main contributions are summarized as follows:
    \begin{enumerate}[label=\arabic*.]

        \item  \emph{Design of a novel decision-theoretic framework for early liquidation.} We develop a two-period liquidation model that offers flexibility in allocating liquidation volume over time and endogenizes the interdependence of institutional liquidation decisions. Based on this model, we adopt a worst-case approach to optimize an individual bank's early liquidation strategy by formulating a maximin problem that maximizes its liquidation proceeds under the worst-case risk contagion arising from counterparties' strategies. This formulation serves as the analytical foundation for quantitatively evaluating early liquidation decisions.
			
        \item \emph{Tractable relaxation for the maximin liquidation problem.} A critical challenge of the said maximin problem is its analytical intractability, which stems from a multi-layer optimization structure driven by implicit strategic interactions among banks. To overcome this, we propose a relaxed version of the problem, which replaces the liquidation size at clearing by its worst-case upper bound, and establish its validity as a robust surrogate for the original formulation.

        \item \emph{Analytical solution to the relaxed problem.} We demonstrate that the relaxed problem yields a unique solution---hereafter referred to as \emph{robust optimal strategy}---that can be expressed in (semi-)closed form, even in complex scenarios involving permanent price impacts or interbank exposures. This analytical form of the strategy not only reduces the computational burden relative to the original problem by circumventing the need for exhaustive search over high-dimensional spaces but also facilitates effective decision support by providing practitioners with a robust tool to secure liquidity positions during early liquidation, even without observable data on counterparty actions.

        \item \emph{Practical insights into strategic behaviors.} Our analysis demystifies the strategic drivers of fire sales. We show that while early liquidation is theoretically capped at half the total liquidation volume required in the absence of early action, the incentive to sell early depends on a trade-off between immediate competition and future price risks. In particular, we identify a critical tipping point governed by market conditions. Initially, a higher permanent price impact (i.e., slower price recovery) prompts counterparties to sell early in the worst-case scenario, leading an individual bank is to reduce its own early liquidation to avoid crowded trades. Beyond this threshold, however, the magnified risk of a future price collapse dominates, paradoxically driving an increase in the robust optimal early liquidation size. %Our framework offers tangible advantages for implementation and decision support. The robust optimal strategy enables practitioners to safeguard their liquidity positions in early liquidation without requiring detailed knowledge of other institutions’ actions. The tractable structure of the robust optimal strategy substantially reduces computational burden relative to the original problem, avoiding exhaustive searches over high-dimensional strategy spaces. These results could also offer a quantitative lens for regulators to anticipate the market consequences of preemptive liquidation behaviors.

    \end{enumerate}
	
    	The remainder of the paper is organized as follows. Section~\ref{sec:literature} provides a review of the related literature. In Section~\ref{sec:preliminaries}, we identify key issues and motivate our research questions through an illustrative example. Section~\ref{sec:perm_impact} derives the robust optimal strategy for the case with permanent price impact and analyzes its sensitivity to the magnitude of the impact. In Section~\ref{sec:intb_liab}, we perform a similar analysis for the case with interbank liabilities. % where we develop the robust optimal strategy of a similar form to that in the non-interconnected system and provide an alternative worst-case near-optimal strategy for the partially known but weakly interconnected system.
    Section~\ref{sec:conclusions} concludes the paper. All proofs of the theoretical results can be found in the appendix.

    \section{Literature Review}
    \label{sec:literature}
	
    %The indirect contagion effect due to fire sale also has been well recognized. \citep{ShleiferVishny2011} stress that systemic losses are exacerbated by the liquidation of overlapping assets. \citep{GreenwoodEtAl2015} establishes a framework for empirical estimation of the liquidation effects based on the data on European banks' portfolios disseminated by the European Banking Authority (EBA), which is extended by \citep{DuarteEisenbach2021} and \citep{PangVeraart2023}. \citep{GlassermanYoung2016} provides a comprehensive review on contagions in financial networks.
	
    {\sffamily\bfseries Impact of fire sales on systemic risk.} Our research builds upon the literature on risk contagion via fire sales of illiquid assets, initiated by \cite{CifuentesEtAl2005}. They investigate the equilibrium of liquidation prices in financial systems where banks are forced to liquidate assets to satisfy regulatory constraints. \cite{ChenEtAl2016} formulate this equilibrium as an optimization problem and introduce a ``liquidity amplifier'' via sensitivity analysis, which quantifies how illiquidity intensifies fire-sale effects, thereby offering policy implications for intervention. In addition, \cite{AminiEtAl2016c} establish sufficient conditions for the uniqueness of the equilibrium liquidation price, while \cite{WeberWeske2017} characterize the liquidation price in more complex settings where bankruptcy costs and cross-holdings are taken into account. These papers mainly focus on the assessment of systemic risk rather than the decision-making in liquidation and assume that liquidation takes place only at clearing. By contrast, we accommodate liquidation that occurs before clearing and offer novel insights into strategic liquidation decisions.		
    A distinct line of work, pioneered by \cite{GreenwoodEtAl2015}, leverages correlations in banks' asset exposures to measure fire-sale spillover effects. Subsequent studies have enriched this paradigm by modeling the dynamic process of fire sales~\citep{DuarteEisenbach2021} and incorporating partial information about asset holdings~\citep{PangVeraart2023}. Although their approaches open up new directions, our modeling focus remains on extending the framework of \cite{CifuentesEtAl2005}, leaving the other models for further research. %\footnote{A complementary strand of the literature, such as \cite{ContWagalath2013,ContWagalath2016}, studies fire sales through their effects on price dynamics. They analyze how distress selling distorts statistical properties of asset prices, including volatility and correlation. While this approach offers a parsimonious model of fire-sale effect, it abstracts from balance-sheet exposures and clearing mechanisms, which is significantly different from our research.}

    % In contrast to both streams, our work treats the early liquidation sizes as decision variables chosen by individual banks rather than as equilibrium outcomes or passive responses to regulatory requirements as in the existing studies. Furthermore, we shift our focus from the system-level assessment of risk contagion to the strategic decision-making in liquidations of individual institutions.

    {\sffamily\bfseries Strategic liquidation in financial systems.} Our work contributes to the growing literature of strategic liquidation under systemic risk. \cite{Feinstein2017} studies a financial system with multiple illiquid assets in which firms strategically choose asset sales to maximize their own valuations, and establishes the existence of equilibrium liquidation strategies. To analyze risk contagion arising from strategic liquidation, \cite{BraouezecWagalath2019} consider a static game in which banks maximize their expected profits subject to regulatory capital requirements and characterize the resulting Nash equilibrium. \cite{BichuchFeinstein2019} formulate an optimization problem in which borrowing and liquidation decisions are jointly determined to minimize total financing costs. While these studies analyze strategic liquidation through ex-post equilibrium outcomes of institutional interactions assuming rational behaviors of institutions, our work focuses on the ex-ante liquidation decision of an individual institution adopting a worst-case approach to the characterization of early liquidation strategies.\footnote{It is essential to distinguish the notion of strategic liquidation considered here from the classical literature on optimal execution \citep[e.g.,][]{BertsimasLo1998, GatheralSchied2013}, which focuses on minimizing trading costs when unwinding a fixed asset position over a given period under exogenous market illiquidity. In our systemic context, by contrast, illiquidity is not merely a frictional cost but an endogenous channel of risk contagion that amplifies price declines and threatens systemic stability.}

    {\sffamily\bfseries The role of permanent price impact in financial systems.} Our research is closely related to studies on liquidation-induced systemic risk in the presence of permanent price impact, where price declines caused by liquidation do not immediately reverse and losses from liquidation accumulate over time. \cite{ContSchaanning2019} incorporate such cumulative impacts to empirically measure indirect contagion from fire sales empirically. \cite{CapponiEtAl2020b} examine a dynamic setting where persistent price impacts can trigger a self-reinforcing cycle of mutual fund redemptions and further liquidations. In the context of collateralized systems, \cite{GhamamiEtAl2022} and \cite{PangVeraart2025} model clearing processes in which collateral is liquidated in multiple rounds, with price impacts from early sales carrying over entirely to subsequent rounds. While these works assume \emph{fully} permanent price impact, implying that depressed prices do not recover, our work allows for \emph{partial} price recovery. Furthermore, we examine robust liquidation decisions under permanent price impact from an individual decision-making perspective, whereas the aforementioned studies primarily focus on measuring systemic risk or elucidating contagion mechanisms.

    {\sffamily\bfseries Contagion through interbank exposures.}
    A large body of literature on systemic risk investigates the impact of interbank obligations. The seminal framework of \cite{EisenbergNoe2001} models interbank repayments using a fixed-point characterization. Subsequent research has extended this framework along several dimensions while retaining its core analytical structure. For instance, \cite{RogersVeraart2013} incorporate bankruptcy costs into the framework and explore incentives for consortium-based bailouts. \cite{AcemogluEtAl2015} and \cite{CapponiEtAl2016} investigate how network topology and concentration affect shock transmission, respectively. \cite{KusnetsovMariaVeraart2019} and~\cite{Veraart2020} extend the framework to accommodate multiple maturities and quantify distress versus default contagion, respectively. Furthermore, this approach has been applied to diverse areas such as systemic stress scenario selection \citep{AhnEtAl2023}, debt valuation \citep{BaruccaEtAl2020, BanerjeeFeinstein2022}, intervention policies \citep{Ahn2019, CapponiBernard2022}, and multilateral netting effects~\citep{Ahn2020, AminiEtAl2020}. We also incorporate interbank exposures as an additional risk contagion channel. However, in contrast to the aforementioned works that analyze systemic risk from a regulatory perspective, our work examines the impact of this contagion channel on early liquidation decisions from the standpoint of an individual bank.

    % focus on the evaluation of systemic risks while our work undertakes the construction of a liquidation strategy under systemic risks.

    % Several other works researching liquidation strategy in the financial system, are also related to our paper. \citep{CaccioliEtAl2014} conducts stability analysis on contagion due to liquidation effects. It proposes a multi-period setting and devises an active liquidation strategy for simulations. Whereas, the optimality of the liquidation strategy is yet to be discussed. It only assumes a heuristic proportional strategy without discussing the uncertainty of banks' strategies, which is significantly different from our optimization problem. \citep{BraouezecWagalath2018} exhibits a bank's optimal selling strategy on loans and marketable assets due to the regulatory requirements. However, financial contagion is not taken into account in this paper.

\section{Research Questions}\label{sec:preliminaries}
    % This section is devoted to illustrating our robust approach for a simplified example, while highlighting the issues emerging in more practical extensions. We defer the formal model setup and its justification to Section \ref{sec:perm_impact}, and focus only on the essential background for clarity in this section. Here, we assume the absence of permanent price impact and interbank liabilities.
We motivate our research questions via an illustrative example---a special case of our main problems.

{\sffamily\bfseries An illustrative example.} Consider a system of $m$ banks operating over two periods ($t = 1, 2$), where liabilities are cleared at the end of the time horizon. Each bank $i$ possesses liquid assets valued at $e_i$ and holds $z_i$ units of illiquid assets. If the liquid assets prove insufficient to meet liabilities $L_i$ (i.e., $e_i < L_i$), the bank must cover the shortfall by selling off its illiquid holdings at a depressed price---a process known as fire sales. In contrast to the standard literature that assumes no early fire sales, we allow each bank $i$ to \emph{strategically} liquidate its illiquid holdings $x_i$ at $t=1$, prior to clearing at $t=2$, in order to mitigate price impacts at clearing and sell a smaller amount of illiquid assets than in the scenario without early liquidation. Given a realization $\bx=(x_1,\ldots,x_m)^\top$ of all banks' early liquidation decisions, let $\yf_i(\bx)$ denote the amount of illiquid assets bank $i$ \emph{must} sell at clearing to fulfill its debt obligation for each $i=1,\ldots,m$. Then, by defining $\bx_{-i} \coloneqq (x_1, \dots, x_{i-1}, x_{i+1}, \dots, x_m)^\top$, bank $i$'s liquid asset value at clearing becomes:
\begin{equation}
    \lf_i(x_{i},  \bx_{-i}) \coloneqq e_i + x_i Q\left(\sum_{j = 1}^{m} x_j \right) + \yf_{i}(\bx) Q\left(\sum_{j = 1}^{m} \yf_{j}(\bx) \right),\label{l-bi-def}
\end{equation}
where $Q(\cdot)$ maps the total supply of illiquid assets to their price, and thus, the second and third terms represent proceeds from asset liquidation at $t=1$ and $t=2$, respectively.

		{\sffamily\bfseries Robust decision-making.} In this example, \emph{each bank $i$ would seek to determine its strategy $x_i$ that optimizes $\phi_i(x_i,\bx_{-i})$ while maintaining robustness against other banks' strategies $\bx_{-i}$}, as these are practically unobservable.
% In practice, other banks' strategies $\bx_{-i}$ are not disclosed to bank $i$. From a robust perspective, to secure sufficient liquid assets for fulfillment, bank $i$ seeks to maximize $\lf_i(x_{i},  \bx_{-i} ) $ under the worst-case scenarios of $\bx_{-i}$.
However, as well documented in the literature~\citep{AminiEtAl2016c} and as we shall demonstrate later, the mapping $\yf_{j}(\cdot)$ is characterized by a fixed-point equation involving nonconvex functions and is hence analytically intractable. A natural approach to this issue would be to replace $\yf_j(\bx)$ by its surrogate $\yf_j(\bzero)-x_j$ for all $j$, where $\bzero$ denotes the zero vector.\footnote{The early liquidation size $x_j$ is capped at $\yf_j(\bzero)$ for all $j=1,\ldots,m$. This is because, as noted earlier, the goal of early liquidation is to reduce the amount of illiquid assets sold and the losses from liquidation compared to the case without it. Similarly, $\yf_j(\bx)$ is capped at $\yf_j(\bzero)-x_j$.} This accounts for the worst-case scenario in which banks eventually sell off the same amount of illiquid assets as they would in the absence of early liquidation. %We elaborate on its detailed formulation in Section \ref{subsec:perm_prob}. To address this difficulty, we robustify $\lf_i(x_{i},  \bx_{-i} ) $ by substituting $\yf_j(\bx)$ with a robust counterpart $\xst_j - x_j$, which is the upper bound of $\yf_{j}(\bx)$ since the alleviation of price impact under split liquidations ensures that the bank $j$'s total liquidation sizes over two periods does not exceed a concentrated single-period sale. Consequently, we formulate the following maximin problem:
Then, the said decision-making problem of bank $i$ can be formulated as follows:
        \begin{equation}
            \begin{aligned}
                \max_{x_{i}}  \min_{\bx_{-i} } ~~   & e_i + x_i Q \left(\sum_{j=1}^m x_j \right) + \left(\xst_i - x_i \right) Q\left(\sum_{j=1}^m (\xst_j - x_j ) \right)
                \\  \textrm{s.t.} ~~~   &  0 \leq x_j \leq \xst_j~~\text{for}~j=0,1,\ldots,m,
            \end{aligned}
            \label{maximin-prob}
        \end{equation}
where $\bar x_j \coloneqq \psi_j(\bzero)$ for all $j=1,\ldots,m$. We refer to the optimal solution to the outer problem of \eqref{maximin-prob}, denoted by $\xtld_i$, as bank $i$'s \emph{robust optimal strategy}. %To solve \eqref{maximin-prob}, we aggregate the first-period liquidation sizes of other banks into a single quantity, denoted by

Solving \eqref{maximin-prob} is surprisingly simple. By letting $s_i = \sum_{j \neq i} x_j$ and $\sst_i = \sum_{j\neq i}\xst_j$, it can be recast as $\max_{x_{i} \in [0, \xst_i]} \min_{s_{i} \in [0, \sst_i]} g_i(x_i, s_i)$, where $g_i(x, s) \coloneqq e_i + x Q(x + s) + (\xst_i - x) Q(\xst_i + \sst_i - (x + s_i))$. % and rewrite the objective of \eqref{prob:maximin-perm} as
        % \begin{equation}
        %     g_i(x_i, s_i) \coloneqq e_i + x_i Q(x_i + s_i) + (\xst_i - x_i) Q\left(\xst_i + \sst_i - (x_i + s_i) \right).
        % \end{equation}
        Under the common assumptions on $Q(\cdot)$ in the literature (see Assumption~\ref{assump1}), $g_i(x, s)$ is concave in $x$ and convex in $s$. Therefore, the minimax theorem \citep[][Theorem 36.3]{Rockafellar2015} holds: $\max_{x_{i} \in [0, \xst_i]} \min_{s_{i} \in [0, \sst_i]} g_i(x_i, s_i)  = \min_{s_i \in [0, \sst_{i}]} \max_{x_{i} \in [0, \xst_i]} g_i(x_i, s_i) $. Furthermore, $h_i$ satisfies the following point symmetry: $g_i(x_i, s_i) = g_i(\xst_i-x_i, \sst_i - s_i)$ for any ${x_{i} \in [0, \xst_i]}$ and ${s_{i} \in [0, \sst_i]}$. Accordingly,  $(\xst_i/2, \sst_i/2)$ is the saddle point of $h_i$, implying that
        \begin{equation}
            \xtld_i = \frac{\xst_i}{2}.
        \end{equation}
        % Intuitively, for bank $i$, the worst-case scenarios given its overly aggressive or conservative strategies are that all other banks also decide to liquidate excessively or insufficiently at $t = 1$, which exposes bank $i$ to extreme price declines. Therefore, to achieve profit maximization under worst-case scenarios, bank~$i$ should liquidate assets moderately at $t = 1$.
		
		{\sffamily\bfseries Main issues and questions.}
		The simplicity of identifying the robust optimal strategy, however, is limited to the special case presented above. When practical considerations are incorporated, finding the strategy becomes significantly more complex. Specifically, this paper focuses on two distinct extensions prevalent in the literature: incorporating permanent price impact and including interbank exposures.
		% Despite the illustrative guideline provided, the result rests on a simplified framework, restricting its application to more practical settings. Specifically, when we further incorporate permanent price impact,
        In the former case, early liquidation  exerts downward pressure on the asset price at clearing, which causes the associated version of $h_i$ to no longer exhibit point symmetry. In the latter case, interbank repayments must be factored into the value of liquid assets at clearing. This makes the corresponding version of $h_i$ dependent on the intricate interplay of liquidation strategies and payment requirements, thereby destroying its concave-convexity and precluding the application of the minimax theorem. Consequently, our primary research question is:
        \begin{itemize}
            \item[Q1.] \emph{Can we characterize (semi-)closed-form expressions of robust optimal strategies in these two cases?}
        \end{itemize}
        Additionally, we investigate the following question:
        \begin{itemize}
            \item[Q2.] \emph{Is the idea of replacing $\yf_j(\bx)$ with a surrogate $\xst_j-x_j$ valid for identifying a good liquidation strategy?}
        \end{itemize}
        We examine these questions in the subsequent sections.

{\sffamily\bfseries Further notation.} All vectors are column vectors and denoted by bold symbols, e.g., $\bu = (u_1, \dots, u_d)^\top \in \bbR^d$.  We use $[m]$ and $\bone$ to represent the set $\{1, \dots, m\}$ and the vector of ones in a suitable dimension, respectively. For any two vectors $\bu,  \bv \in \bbR^d$, $\bu \leq \bv$ means an entry-wise inequality, $[\bu, \bv] = [u_1, v_1] \times \cdots \times [u_d,  v_d]$, and $\bu \wedge \bv = (\min\{u_1,  v_1\},  \dots,  \min\{u_d,  v_d\})^\top$.  For any matrix $\bM$, we denote its $i$-th column by $\bM_i$. Finally, for any $x\in\bbR$, $x^+\coloneqq\max\{x,0\}$.

    \section{Robust Optimal Strategy with Permanent Price Impact}	\label{sec:perm_impact}
	
	It is widely known in the literature that price impacts are often persistent; buyers observe the depressed post-liquidation price and anticipate a slow recovery, preventing immediate reversion. In this section, we analyze the impact of this permanent price impact on the robust optimal strategy by characterizing it under different degrees of permanent price impact.
	
	\subsection{Model Setup}\label{subsec:model_perm}
	We detail our full modeling framework, elaborating on technical aspects omitted from the illustrative example and demonstrating how permanent price impact is incorporated into that example.

    {\sffamily\bfseries Balance sheet information.} We consider a financial system comprised of $m$ banks, indexed by $1, \dots, m$, that operates over two periods ($t = 1, 2$). Each bank $i$'s balance sheet consists of the following four components:
    \begin{itemize}
        \item $e_i > 0$ stands for the ``value'' of its liquid assets;

        \item $z_i \geq 0$ represents the ``volume'' of its illiquid assets, with a nominal unit price of $P$;\footnote{We adopt the standard assumption in the literature that all banks hold the same type of illiquid assets, which serves as an approximation of multiple correlated assets~\citep{WeberWeske2017,BraouezecWagalath2019}.}

        \item $L_i> 0$ denotes its total liabilities to creditors;

        \item $w_i \coloneqq e_i + z_i P - L_i$ means its initial net worth before liquidation.
    \end{itemize}
    We assume that these components remain constant throughout the two periods and that liquid assets (e.g., cash, central bank reserves, and short-term government securities) can be sold and transferred at their nominal value without any cost. To exclude trivial cases where banks are not required to liquidate their illiquid assets, it is assumed that $e_i<L_i$ for all $i\in[m]$. The associated balance sheet is described in Table~\ref{tab:init-bst-noint}.

    \begin{table}[t]
        \centering\caption{Initial balance sheet of bank $i$ before liquidation\label{tab:init-bst-noint}}
        {\begin{tabular}{ccc}  % centered columns (4 columns)
                \toprule
                \textbf{Assets}  & &\textbf{Liabilities and Net Worth}   \\  % inserts table %heading
                \midrule
                Liquid assets $e_i$ & & Liabilities $L_i$   \\
                Illiquid assets $z_i$ with price $P$ & &Net worth $w_i$     \\
                \bottomrule
        \end{tabular}}{}
    \end{table}

    % Assuming no interbank obligations, the total liability of bank $j$ is given by $L_j = b_j$ for all $j \in [m]$. Then, the initial net worth of bank $j$ is given by $w_j = e_j + z_j P - L_j$. The external liquid assets, originating from financial institutions outside the system, can be converted to cash without incurring loss in value. Typical examples of external liquid assets include cash, short-term commercial paper, and Treasury bills. While
    {\sffamily\bfseries Illiquid asset price modeling.} Illiquid assets (e.g., loans, mortgages, and real estate holdings) are, however, difficult to sell without a significant loss of value. In the related literature \citep{CifuentesEtAl2005, ChenEtAl2016}, this loss is typically modeled by the so-called inverse demand function $Q: [0,  \sum_{j = 1}^m z_j] \rightarrow \bbR_+$, a decreasing function that maps the aggregate supply of illiquid assets across all banks to their price. %Given the total liquidation size $x$, the total \emph{liquidation income}, which is the total amount of cash raised in the liquidation, is given by $xQ(x)$.
    Throughout this paper, we impose several regularity conditions on the function $Q$, described below:
    \begin{assumption}The inverse demand function $Q$ is twice continuously differentiable and satisfies the following five conditions: (i) $Q(0) = P$; (ii) $Q'(x) < 0$; (iii) $(xQ(x))' > 0$; (iv) $Q^{\prime\prime}(x) \geq 0$; and (v)  $(xQ'(x))' < 0$.        \label{assump1}
    \end{assumption}
    The first condition states that illiquid assets are valued at their nominal price in the absence of sales. Conditions~(ii) and~(iii) reflect the observation that as the supply of illiquid assets increases, their price decreases while the total value of these assets in the market grows. The fourth condition imposes a diminishing marginal effect on the price of illiquid assets with respect to their size. Condition~(v) ensures that the total value of illiquid assets in the market is concave in their size. Note that the above conditions (or their variants) are commonly used in the literature~\citep{AminiEtAl2016c,BichuchFeinstein2019}. Examples of inverse demand functions satisfying Assumption~\ref{assump1} include $Q(x) = P - a x$ with $0 < a < P/(2\sum_{k = 1}^m z_k)$, $Q(x) = P\exp(-a x)$ with $0 < a < 1/\sum_{k = 1}^m z_k$, and $Q(x) = P(\beta x+1)^{-a}$ with $0 < a < 1/(\beta \sum_{k = 1}^m z_k)$ and $\beta>0$.
    % \begin{example}
    %     % Consider different inverse demand functions $Q \in \cQ$.
    %     \begin{enumerate}[label=(\roman*)]
    %         \item If $Q(x) = P - a x$, then Assumption \ref{assump1} implies $0 < a < P/(2\sum_{j = 1}^m z_j)$.% by Assumption \ref{assump1} (i)-(iii). By Assumption \ref{assump1} (iv) and (v), $a \geq 0$. In conclusion, $0 < a < P/(2y_{tot})$.

    %         \item If $Q(x) = P\exp(-a x)$, then Assumption \ref{assump1} implies $0 < a < 1/\sum_{j = 1}^m z_j$. % by Assumption \ref{assump1} (i)-(iii). By Assumption \ref{assump1} (iv) and (v), we have $a < 1/y_{tot}$. As a result, $0 < a < 1/y_{tot}$.

    %         \item If $Q(x) = P(\beta x+1)^{-a}$ for $\beta > 0$, then Assumption \ref{assump1} implies $0 < a < 1/(\beta \sum_{j = 1}^m z_j)$. % by Assumption \ref{assump1} (i)-(iii). By Assumption \ref{assump1} (iv) and (v), $a < 1/(\beta y_{tot})$. In combination, $\beta > 0$ and $0 < a < 1/(\beta y_{tot})$.

    %     \end{enumerate}
    % \end{example}
    % To concentrate our analysis of the liquidation effect, we impose that no exogenous assets and liabilities are introduced across two periods. Thus, apart from asset liquidation, the quantities or values of banks' liquid assets, illiquid assets, and liabilities remain unchanged.

    % \begin{remark}
    %     Compared to \cite{CifuentesEtAl2005}, our framework extend to the two-period setting, while assuming the absence of interbank liabilities. This allows us to focus our analysis on the effects of permanent price impact, whereas the role of interbank liabilities is investigated separately in Section \ref{sec:intb_liab}.
    % \end{remark}

    {\sffamily\bfseries Early liquidation and permanent price impact.} We assume that all banks clear their debts at the end of the time horizon (i.e., $t=2$), consistent with prior studies \citep[e.g.,][]{CifuentesEtAl2005, ChenEtAl2016, AminiEtAl2016c, GhamamiEtAl2022}. As alluded to earlier, a salient feature of our modeling framework lies in the option for \emph{strategic early liquidation} at $t=1$ to mitigate liquidation losses and to sell a smaller volume of illiquid assets, which is in stark contrast to most existing models that allow liquidation only at clearing. Suppose that each bank $i$ chooses to sell off $x_i$ units of illiquid assets at $t = 1$. Then, $Q(\sum_{k = 1}^{m}x_{k})$ becomes the market price of illiquid assets at $t=1$ and is strictly smaller than the nominal price $P$ by Assumption~\ref{assump1}. % When in the case of insufficient liquid assets, banks need to raise funds by liquidating illiquid assets, even incurring liquidation losses. Since we focus on the analysis of liquidation strategy, we also assume that $e_j < L_j$ for all $j \in [m]$, implying that each bank $j$ is short of liquid assets for repayment and required to liquidate their illiquid assets.

    To account for \emph{permanent price impact}, we consider a situation where the price impact generated at $t=1$ persists into $t=2$. Specifically, the asset price at clearing is governed by another inverse demand function of the asset supply at $t=2$, defined as $u\mapsto Q(\gamma\sum_{k=1}^mx_k+u)$, where the parameter $\gamma\in[0,1]$ represents the proportion of the early liquidation volume that continues to depress prices at $t = 2$. That is, a larger $\gamma$ implies slower price recovery. The limiting case $\gamma = 0$ corresponds to fully transient impact, where the initial price at $t=2$ recovers to the nominal price $P$, as in the example in Section~\ref{sec:preliminaries}. Conversely, $\gamma = 1$ signifies fully permanent price impact, where the opening price at $t=2$ remains at the post-liquidation price at $t=1$; this latter case coincides with the model in \cite{GhamamiEtAl2022}.
    While existing research commonly focuses on either of these two extreme cases, our approach generalizes the framework of \cite{GhamamiEtAl2022} by allowing for ``partially'' permanent price impact via the parameter $\gamma$.

    {\sffamily\bfseries Clearing process at $t=2$.} Following the early liquidation at $t=1$, each bank $i\in[m]$ retains $z_i - x_i$ units of illiquid assets, while its liquid asset value becomes $e_i + x_{i} Q(\sum_{k = 1}^{m}x_{k})$. At the clearing stage, given a price $p$, each bank $i$ attempts to cover its liabilities by selling  $p^{-1}(L_i - e_i-x_{i} Q(\sum_{k = 1}^{m}x_{k}))^+$ units of illiquid assets. If the bank cannot meet the obligations, it defaults and liquidates all remaining $z_i-x_i$ units of assets. Thus, the quantity sold by bank~$i$ is  $\{p^{-1}(L_i - e_i-x_{i} Q(\sum_{k = 1}^{m}x_{k}))^+\} \wedge (z_i-x_i)$ for all $i\in[m]$. Consequently, the market price $\ppm_2$ of illiquid assets at $t = 2$ is determined by the following fixed-point equation
        \begin{equation}\label{eq:p2-pm}
            \ppm_2 = Q\left(\gamma \sum_{j=1}^{m}x_j + \sum_{j = 1}^{m} \frac{\left(L_j - e_j - x_{j}Q(\sum_{k = 1}^{m}x_{k})\right)^+}{\ppm_2} \wedge \left(z_j - x_{j}\right) \right).
        \end{equation}
        Note that a straightforward extension of Lemma~3 in \cite{AminiEtAl2016c} ensures the uniqueness of $\ppm_2$.
        We denote each bank $i$'s liquidation size at $t = 2$ by
        \begin{equation}
            \yfpm_{i}(\bx) \coloneqq \frac{\left(L_i - e_i - x_{i}Q(\sum_{k = 1}^{m}x_{k})\right)^+}{\ppm_2} \wedge \left(z_i - x_{i}\right).\footnote{When $\gamma = 0$, $\yfpm_{i}(\cdot) $ reduces to $\yf_i(\cdot)$ in the example in Section~\ref{sec:preliminaries}. }
            \label{eq:xj2-pm}
    \end{equation}
    Finally, bank $i$'s liquid asset value at clearing is given by
    \begin{equation}
        \lpm_i(x_{i},  \bx_{-i}) \coloneqq e_i + x_i Q \big(\bone^\top \bx \big) + \yfpm_{i}(\bx) Q\left(\gamma \bone^\top \bx + \sum_{j = 1}^{m} \yfpm_{j}(\bx) \right),
        \label{def:l-perm}	
    \end{equation}
    where the second and last terms denote bank $i$'s income from fire sales at $t=1$ and $t=2$, respectively.

    \begin{remark}
        Permanent price impact could also be modeled using more sophisticated dynamics, such as, the Volume Weighted Average Price (VWAP) approach investigated in ~\cite{BanerjeeFeinstein2021}. Unlike static snapshots, this method captures the entire liquidation history by defining the average asset price through an integral of the inverse demand function.
        % Specifically, given the inverse demand function $Q$ and liquidation quantity from $x$ to $y$, the VWAP is $(y-x)^{-1} \int_{x}^{y}Q(t)dt $. Under permanent price impact, given the first-period strategies $\bx$, the first period VWAP is $(\bone^\top \bx)^{-1} \int_{0}^{\bone^\top \bx} Q(t) dt $, and the VWAP at $t = 2$ becomes $(\bone^\top \byf(\bx))^{-1} \int_{\gamma \bone^\top \bx}^{\gamma \bone^\top \bx + \bone^\top \byf(\bx)} Q(t) dt $.
        % Nonetheless, given that Assumption \ref{assump1} still applies to VWAP when $Q$ satisfies Assumption \ref{assump1}, the qualitative insights derived in the following Section~\ref{subsec:result_perm} are expected to persist under a VWAP specification,
        However, the technical analysis based on the VWAP framework introduces significant mathematical complexity due to the integration terms; therefore, we retain our current setup to preserve analytical tractability. We anticipate that the qualitative insights derived from our findings would persist under the VWAP framework.
        % This choice also ensures consistency with our introductory example in Section~\ref{sec:preliminaries}, which serves as a special case of the current model when $\gamma = 0$. We then defer these extensions on price impact dynamics to future research.
    \end{remark}

	\subsection{Main Problem}	\label{subsec:perm_prob}
	
	% In this subsection, we aim to identify a strategy that guarantees bank $i$'s liquid reserves for liability repayment regarding the uncertainty of other banks' strategies. To facilitate the problem formulation, we first observe that banks have no incentive to raise funds beyond their payment requirements at $t=1$ since front-loading liquidation would only intensify the price impact, and consequently, each bank's first-period liquidation size is constrained from above by its single-period liquidation size, which is formalized in the following lemma.
 %    \begin{lemma}
 %        If $L_j \geq e_j +  x_{j} Q(\sum_{k = 1}^{m} x_{k} )$ for all $j \in [m]$, then $\bx \leq \bxst$.
 %        \label{lemma-xj1-bound}
 %    \end{lemma}
 %    This result can be easily derived by contrapositive using \eqref{useful-ineq-x} and \eqref{useful-ineq-y}, which is presented in Section \ref{sec:proof_perm}.
%			Intuitively, implied by Assumption \ref{assump1}~(ii), excessive liquidations by some banks lead to excessive incomes earned beyond their needs.
    Based on the modeling framework described above, we now formulate the decision-making problem for strategic early liquidation. While certain elements of this subsection may overlap with the discussion on robust decision-making in Section~\ref{sec:preliminaries}, we provide here the full technical details omitted there, while extending the formulation to account for permanent price impact.

    {\sffamily\bfseries Robustification against other banks' actions.} As shown in~\eqref{def:l-perm}, each bank's liquid asset value at clearing depends on the volume of illiquid assets sold by other banks at $t=1$, which is often unknown. Hence, the solution to the following problem yields bank $i$'s optimal early liquidation decision under the worst-case strategies of other banks:
    \begin{equation}
        \begin{aligned}
            \max_{x_{i}}  \min_{\bx_{-i} } ~~  & \lpm_i(x_{i}, \bx_{-i})
            \\  \textrm{s.t.}~~~   &  \bzero \leq \bx \leq \bxst,
        \end{aligned}
        \label{prob:bilevel-perm}
    \end{equation}
    where $\xst_j$ is defined as in Section~\ref{sec:preliminaries} and is equal to $\yfpm_j(\bzero)$ for any $\gamma\in[0,1]$ and for all $j\in[m]$. The constraint in \eqref{prob:bilevel-perm} rules out irrational scenarios where banks liquidate more assets than they would without the early liquidation option, as this would contradict the purpose of the option.

    However, solving \eqref{prob:bilevel-perm} is extremely difficult because the values $\yfpm_{1}(\bx),\ldots,\yfpm_m(\bx)$ are endogenous to the fixed-point equation in \eqref{eq:p2-pm}, inducing a nested three-level optimization structure in \eqref{prob:bilevel-perm}. Indeed, one can numerically verify that the problem is highly nonconvex, making it challenging to identify the global optimum even via grid search in low-dimensional settings. To address this issue, we obtain robust upper bounds of $\yfpm_{1}(\bx),\ldots,\yfpm_m(\bx)$:
    \begin{lemma}
        For any $\bx \leq \bxst$, $\yfpm_{j}(\bx) \leq \xst_j-x_j$ for all $j \in [m]$ and $\gamma \in [0, 1]$.
        \label{lemma:yfpm-ub}
    \end{lemma}
    Several remarks are in order. First, the bound in the preceding lemma is tight and robust against as equality is attained when $\bx=\bzero$ or $\bx=\bxst$. Second, this bound is robust to the early liquidation decisions of other banks; as we shall see shortly, this significantly simplifies the subsequent analysis. Finally, this lemma formalizes the intuition that the early liquidation option does not increase a bank's total liquidation size.%Lemma \ref{lemma:yfpm-ub} suggests that without over-funding at $t=1$, each bank's total liquidation sizes in the two-period setting does not exceed its single-period liquidation size. Given that liquidation is dispersed in the two-period setting, this result is natural once we verify that the aggregate price impact at $t=2$, including the impact of persistent first-period liquidation, is less significant than in the single-period case, namely $\ppm_2 \geq \pst$. Since $\ppm_2$ and $\pst$ are both solutions to fixed-point systems, the inequality can be established via an iterative scheme, as detailed in Section \ref{sec:proof_perm}.

    {\sffamily\bfseries Relaxation of liquidation size at clearing.} Invoking Lemma \ref{lemma:yfpm-ub}, we circumvent the aforementioned intractability of \eqref{prob:bilevel-perm} by relaxing the constraint on each bank $j$'s liquidation size at clearing. Specifically, we allow it to take any value in the interval $[0,\xst_j-x_j]$ rather than restricting it to $\yf_j(\bx)$ for each $j\in[m]$. Under this relaxation, the maximin problem in \eqref{prob:bilevel-perm} can be converted into the following form:
    \begin{equation}
        \begin{aligned}
            \max_{x_i, y_{i}}  \min_{\bx_{-i},  \by_{-i} } ~~  & e_i + x_{i} Q\big(\bone^\top \bx\big) + y_{i} Q\big(\bone^\top (\gamma \bx + \by)\big)
            \\  \textrm{s.t.} ~~~~ &  \bx+\by\leq\bxst~~\textrm{and}~~\bx,\by\geq\bzero.
        \end{aligned}
        \label{prob:conv-perm}
    \end{equation}
    For any feasible solution $(\bx,\by)$ of \eqref{prob:conv-perm}, Assumption~\ref{assump1} indicates that
    $$
    \begin{aligned}
    (\xst_{i} -x_i)Q\big(\gamma x_i+\xst_{i}-x_i+\bone^\top (\gamma\bx_{-i}+\by_{-i})\big)  &\geq y_{i} Q\big(\gamma x_i+y_i+\bone^\top (\gamma\bx_{-i} + \by_{-i})\big)\\%= y_{i} Q\big(\bone^\top (\gamma\bx + \by)\big)\\
    &\geq
    y_{i} Q\big(\gamma x_i + y_i + \bone^\top (\gamma\bx_{-i} + \bxst_{-i}-\bx_{-i})\big).
    \end{aligned}
    $$
    This implies that, first, for fixed $\bx_{-i},\by_{-i}$, the objective function in \eqref{prob:conv-perm} is maximized when $y_i=\xst_i-x_i$, and second, for fixed $x_i,y_i$, the objective function is minimized when $\by_{-i}=\bxst_{-i}-\bx_{-i}$. In other words, the first constraint of \eqref{prob:conv-perm} is binding at its optimal solution. Consequently, solving \eqref{prob:conv-perm} is equivalent to solving
    \begin{equation}
        \begin{aligned}
            \max_{x_{i}}  \min_{\bx_{-i} } ~~   & \lfpm_i(x_i, \bx_{-i})
            \\  \textrm{s.t.} ~~~   &  \bzero \leq \bx \leq \bxst,
        \end{aligned}
        \label{prob:maximin-perm}
    \end{equation}
    where
    % \begin{equation}
        $\lfpm_i(x_{i}, \bx_{-i}) \coloneqq e_i + x_i Q \big(\bone^\top \bx \big) + (\xst_i - x_i) Q\big(\gamma \bone^\top \bx + \bone^\top (\bxst - \bx )\big)$ for all $\bx\in[\bzero,\bxst]$. This is the main problem we aim to investigate in this section, and we refer to the optimal solution to its outer problem as bank $i$'s \emph{robust optimal strategy} for early liquidation and denote it by $\xpm_i$.
        % \label{def:lf-perm}
    % \end{equation}

{\sffamily\bfseries Validity of the formulation.} Before presenting the analysis of our problem \eqref{prob:maximin-perm}, let us discuss its practical validity to address potential concerns about the implication of optimizing $\lfpm_i(x_{i}, \bx_{-i})$. Specifically, the following result provides insights into how the relaxation affects the liquidation strategy, or equivalently, the relationship between the solutions to \eqref{prob:bilevel-perm} and \eqref{prob:maximin-perm}
    \begin{proposition}[Validity of~\eqref{prob:maximin-perm}]
        For any fixed $x_i \in [0, \xst_i]$, the following statements hold:

        \begin{enumerate}[label=(\roman*)]

            \item If $\min_{\bx_{-i} \in [\bzero,  \bxst_{-i} ]} \lfpm_i(x_i, \bx_{-i}) \geq L_i$, then
            $
            \min_{\bx_{-i} \in [\bzero,  \bxst_{-i} ] } \lpm_i(x_i, \bx_{-i}) = L_i.
            $

            \item If $\min_{\bx_{-i} \in [\bzero,  \bxst_{-i}] } \lfpm_i(x_i, \bx_{-i}) < L_i$, then
            $
            \min_{ \bx_{-i} \in [\bzero,  \bxst_{-i} ]} \lpm_i\left(x_i, \bx_{-i}\right) \geq \min_{\bx_{-i} \in [\bzero,  \bxst_{-i} ]} \lfpm_i\left(x_i, \bx_{-i}\right).
            $
        \end{enumerate}
        \label{prop:eff-perm}
    \end{proposition}
    By the definitions of $\xst_i$ and $\yfpm_j(\bx)$, it is easy to check that $\lpm_i(x_i, \bxst_{-i}) \leq L_i$ for any $x_i \in [0, \xst_i]$; see the proof of Proposition~\ref{prop:eff-perm} for more details.\footnote{From a practical perspective, this inequality holds since each bank's liquid asset value at clearing is characterized as either exactly meeting the payment obligation or failing to do so.} Thus, the first statement of Proposition~\ref{prop:eff-perm} shows that if the optimal value of~\eqref{prob:maximin-perm} is greater than or equal to $L_i$, then bank $i$'s robust optimal strategy, $\xpm_i$, is also an optimal solution to the outer problem of \eqref{prob:bilevel-perm}. Furthermore, the second statement of the theorem suggests that, even if the optimal value of~\eqref{prob:maximin-perm} is smaller than $L_i$, bank $i$'s robust optimal strategy still maximizes a lower bound for the worst-case value of its liquid assets at clearing. Accordingly, both statements indicate that the robust optimal strategy obtained by solving~\eqref{prob:maximin-perm} can serve as a conservative approximation and a reliable surrogate for the true optimal decision that solves~\eqref{prob:bilevel-perm}. This answers our second research question (Q2), introduced in Section~\ref{sec:preliminaries}, for the case of permanent price impact.%The second statement is straightforward if $\lpm_i(x_i, \bu^\gamma) = L_i$ where $\bu^\gamma \coloneqq \argmin_{\bx_{-i} \in [\bzero, \bxst_{-i}]} \lpm_i(x_i, \bx_{-i})$. Even if $\lpm_i(x_i, \bu^\gamma) < L_i$ that bank $i$ fails to fully repay and sells out all remaining assets under the worst-case scenarios, the second statement still holds because the relaxations $\{\xst_j - x_j\}_{j \in [m]}$ in $\lfpm_i(x_i, \bx_{-i})$ provide the worst quantification for the second-period price impact determined by $\{\yfpm_j(\bx)\}_{j \in [m]} $ in $\lpm_i(x_i, \bx_{-i}) $. The details are entailed in the proof of Proposition \ref{prop:eff-perm}.
	
	In Figure~\ref{fig:bi-comp-gm}, we numerically validate the results in Proposition~\ref{prop:eff-perm} by comparing the inner optimal values $\min_{\bx_{-i} \in [\bzero,  \bxst_{-i} ] } \lpm_i(x_i, \bx_{-i})$ and $\min_{\bx_{-i} \in [\bzero,  \bxst_{-i} ] } \lfpm_i(x_i, \bx_{-i})$ across different values of $x_i$ in a two-bank system with $m=i=2$ for varying levels of permanent price impact: $\gamma = $ $0$, $0.5$, and $1$. The former optimal value is plotted using blue curves---solid for $\gamma=0$, dashed for $\gamma=0.5$, and dotted for $\gamma=1$---without any markers, whereas the latter is plotted using red curves---solid for $\gamma=0$, dashed for $\gamma=0.5$, and dotted for $\gamma=1$---with circle $(\bullet)$ markers. The associated parameter setting can be found in the note of the figure.
    The left panel of the figure illustrates cases where bank~2's robust optimal strategy guarantees its solvency regardless of $\gamma$. By contrast, the right panel depicts scenarios in which the proposed strategy leads to solvency when $\gamma = 0$ and $0.5$, but default is inevitable when $\gamma = 1$.

    In both panels, we observe that in the range of $x_2$ where the $\bullet$-marked red curves exceed $L_2=5$, the unmarked blue curves remain at $L_2=5$, which is consistent with Proposition~\ref{prop:eff-perm}(i). Additionally, whenever the $\bullet$-marked red curves are below $L_2=5$, they are positioned slightly beneath the unmarked blue curves, which confirms Proposition~\ref{prop:eff-perm}(ii), and the negligible gap between two curves demonstrates the reliability of our relaxed problem \eqref{prob:maximin-perm} at least in this example.
    \begin{figure}[t]
        {\includegraphics[width=0.99\textwidth]{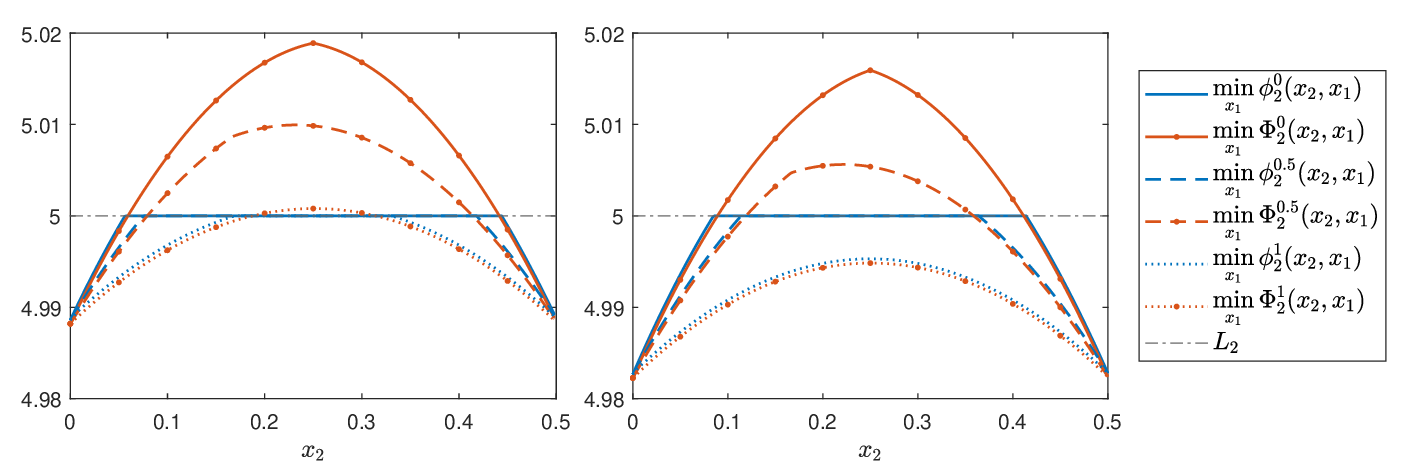}}
        \caption{A numerical comparison between the inner optimal values of \eqref{prob:bilevel-perm} and \eqref{prob:maximin-perm}.\label{fig:bi-comp-gm} We use a two-bank system with $Q(x) = 1 - 0.2x$, $L_1 = L_2 = 5$, $z_1 = 1$, $z_2 = 0.5$, and $e_2 = 4.55$, under $\gamma = 0$, $0.5$, and $1$. Also, we set $e_1 = 4.9$ and $e_1 = 4.85$ for the left and right panels, respectively. }
    \end{figure}
	
	\subsection{Analysis of the Robust Optimal Strategy with Permanent Price Impact}	\label{subsec:result_perm}
In this section, we address our primary research question (Q1) for the case of permanent price impact. To that end, we first derive an analytical expression for the robust optimal strategy $\xpm_i$ and then analyze how it changes with the level of permanent price impact (i.e., $\gamma$).

	{\sffamily\bfseries Derivation of the robust optimal strategy.} When $\gamma = 0$, \eqref{prob:maximin-perm} reduces to \eqref{maximin-prob}, leading to $\xpm_i = \xtld_i = \xst_i/2$. In contrast, when $\gamma > 0$, early liquidation not only incurs an immediate price impact at $t=1$ but also depresses the asset price at $t=2$. Consequently, from bank~$i$'s perspective, the worst-case scenario arises when other banks liquidate more aggressively at $t=1$. To counter this compounded impact at $t = 1$, bank~$i$ must employ a more prudent early liquidation strategy that depends on the permanency level $\gamma$, resulting in a downward deviation of the robust optimal strategy $\xpm_i$ from  $\xtld_i$ that solves \eqref{maximin-prob}. This is formalized in the following theorem. Furthermore, the theorem demonstrates that $\xpm_i$ is uniquely determined and admits a (semi-)closed-form expression for any $\gamma \in [0, 1]$.
    \begin{theorem}[Robust Optimal Strategy with Permanent Price Impact]\label{thm:opt-strat-perm}
        For each $i\in[m]$ and for each $\gamma\in[0,1]$, the robust optimal strategy $\xpm_i$ that solves \eqref{prob:maximin-perm} is uniquely determined and is bounded from above by $\xst_i/2$. More specifically, if  $\gamma \leq \gamma^* \coloneqq 1-\xst_i/(\sum_{j=1}^n\xst_j)$,  then $\xpm_i$ is given by
        %		$\xopt_i = [(1 - \gamma)\xst_i/(2 - \gamma)] \vee \xtld_i$.
        \begin{equation}\label{eq:closed sol}
            \xpm_i = \frac{1 - \gamma}{2 - \gamma}\xst_i,
        \end{equation}
        and if $\gamma>\gamma^*$, $\xpm_i$ is the unique solution to the following equation for $x$:
        \begin{equation}
            Q(x + \sst_i) + x Q'(x + \sst_i) = Q\big(\xst_i + \gamma \sst_i - (1 - \gamma) x \big) + (1 - \gamma)(\xst_i - x) Q'\big(\xst_i + \gamma \sst_i - (1 - \gamma) x \big),
            \label{eq:xpm-fst-ord}
        \end{equation}
        where $\sst_i = \sum_{j \neq i} \xst_j$. % which is equivalent to $\partial \lfpm_i(\xfe_i^\gamma, \bxst_{-i})/\partial x_i = 0$.
    \end{theorem}

    We briefly discuss the technical rationale for the characterization of $\xpm_i$ in the preceding theorem. Analogous to the argument in Section~\ref{sec:preliminaries}, we aggregate the early liquidation sizes of non-target banks into a single quantity, $s_i = \sum_{j \neq i} x_j$, and rewrite the main problem in \eqref{prob:maximin-perm} as
\begin{equation}\label{eq:ggamma}
    \max_{x_{i} \in [0, \xst_i]} \min_{s_{i} \in [0, \sst_i]} g_i^\gamma(x_i, s_i)  = \min_{s_{i} \in [0, \sst_i]} \max_{x_{i} \in [0, \xst_i]} g_i^\gamma(x_i, s_i),
\end{equation}
where $g_i^\gamma(x, s) \coloneqq e_i + x Q(x + s) + (\xst_i - x) Q\left(\xst_i + \sst_i - (1 - \gamma) (x + s) \right)$. The equality in \eqref{eq:ggamma} holds due to the concave-convexity of $g_i^\gamma$ and the minimax theorem \citep[][Theorem 36.3]{Rockafellar2015}. Then, the first-order conditions for \eqref{eq:ggamma}, $({\partial g_i^\gamma}/{\partial x})(x_i,s_i)=0$ and $({\partial g_i^\gamma}/{\partial s})(x_i,s_i)=0$, hold only at
\begin{equation}\label{eq:saddle pt}
(x_i,s_i) = \lt(\frac{1-\gamma}{2-\gamma}\xst_i,\frac{\gamma\xst_i+\sst_i}{2-\gamma}\rt)
\end{equation}
and this saddle point is a feasible solution to \eqref{eq:ggamma} if and only if $\gamma\leq \gamma^*$. This yields the characterization of $\xpm_i$ in~\eqref{eq:closed sol} for the case where $\gamma\leq \gamma^*$. Notably, we observe that \eqref{eq:saddle pt} satisfies $x_i+s_i = \xst_i+\sst_i-(1-\gamma)(x_i+s_i)$. This implies that \eqref{eq:saddle pt} effectively balances the resulting price impact at $t=1$ and $t=2$, as the left- and right-hand sides of this identity correspond to the arguments of the function $Q$ in the definition of $g_i^\gamma$.

Conversely, when $\gamma>\gamma^*$, one can easily intuit that the optimal solution to~\eqref{eq:ggamma} lies at the boundary $s_i=\sst_i$ because the unique saddle point~\eqref{eq:saddle pt} falls within $[0,\xst_i/2)\times(\sst_i,\xst_i+\sst_i]$. This implies that if the permanent price impact is sufficiently large, the worst-case scenario for bank~$i$ arises when other banks liquidate their entire required volume of illiquid assets only at $t=1$. In this case, the robust optimal strategy $\xpm_i$ is characterized by the following first-order condition $({\partial g_i^\gamma}/{\partial x})(x_i,\sst_i)=0$, which is equivalent to solving the equation in \eqref{eq:xpm-fst-ord}.
    While the reasoning outlined above provides a straightforward path to constructing a solution to~\eqref{prob:maximin-perm}, the essence of Theorem~\ref{thm:opt-strat-perm} lies in establishing its uniqueness, thereby ruling out the possibility of any alternative solutions; see the proof of the theorem in Section \ref{sec:proof_perm} for more details. It is also worth highlighting that solving \eqref{eq:xpm-fst-ord} is computationally tractable since $(\partial g_i^\gamma/\partial x)(\cdot, \sst_i)$ is strictly decreasing. Moreover, when the inverse demand function is linear---a common assumption in prior studies---the solution admits a closed form. Specifically, if $Q(x) = P - \alpha x$, then \eqref{eq:xpm-fst-ord} can be written as $P - \alpha (2x + \sst_i) = P - \alpha ((2-\gamma)\xst_i + \gamma \sst_i - 2(1 - \gamma) x)$,
        and thus, \begin{equation}\label{eq:linear}
        \xpm_i = \frac{\xst_i}{2} - \frac{(1-\gamma)\sst_i}{2(2-\gamma)}.
        \end{equation}
        % \begin{equation}
            %     \xpm_i = \frac{\xst_i}{2} - \frac{1 - \gamma}{2 - \gamma} \frac{\sst_i}{2}.
            % \end{equation}
        % Specifically, when $\gamma = 1$, $\xpm_i = \xst_i/2$, which coincides with the robust optimal strategy when $\gamma = 0$.
        Accordingly, in this linear setting, $\xpm_i$ exhibits a V-shaped dependence on $\gamma$: it decreases in $\gamma$ on $[0,\gamma^*)$, as shown in~\eqref{eq:closed sol}, and increases with $\gamma$ on $(\gamma^*,1]$ as indicated by \eqref{eq:linear}.
    % \end{example}

    {\sffamily\bfseries Structural dependence on the parameter $\gamma$.} The V-shaped dependence of the robust optimal strategy mentioned above is not limited to the specific choice of $Q(\cdot)$. Indeed, the following theorem demonstrates that, under a mild condition, this behavior holds for ``all'' inverse demand functions satisfying Assumption~\ref{assump1}.
    \begin{theorem}[Impact of $\gamma$: A Typical Scenario] For any $i\in[m]$, if $\gamma^*\geq1/2$, then the robust optimal strategy $\xpm_i$ increases in $\gamma$ on $(\gamma^*,1]$.
        \label{thm:xpm-non-mono}
    \end{theorem}
    % Theorem \ref{thm:xpm-non-mono} may appear counter-intuitive at first glance, as they suggest that intensified permanent price impact gives rise to a bifurcation in the strategic behaviors of bank $i$, yet this result is compatible with Theorem \ref{thm:opt-strat-perm}.  As implied by Theorem \ref{thm:opt-strat-perm}, in the regime where $\gamma \leq \sst_i/(\xst_i + \sst_i)$, bank $i$ strategically reduces its first-period sales compared to the case of no price impact. This strategic behavior continues as $\gamma$ rises within this regime.
    Firstly, we note that the condition $\gamma^*\geq1/2$ is equivalent to $\xst_i\leq \sst_i\coloneqq \sum_{j \neq i} \xst_j$, meaning that bank~$i$'s contribution to the price impact is smaller than the aggregate contribution of all other banks. Thus, it is typically satisfied in practice. Secondly, as discussed earlier, the worst-case scenario for bank~$i$  with $\gamma>\gamma^*$ corresponds to a situation in which all other banks collectively liquidate $\sst_i$ units of illiquid assets exclusively at $t = 1$. Therefore, Theorem~\ref{thm:xpm-non-mono} aligns with the intuition that, given $\xst_i \leq \sst_i$, a larger $\gamma$ leads to a sharper price decline at $t = 2$. This prompts bank $i$ to increase sales at $t = 1$ to avoid the anticipated future price distress. By combining this result with the analytical solution~\eqref{eq:closed sol} for $\xpm_i$ when $\gamma\leq\gamma^*$, the aforementioned V-shaped dependence holds generally.

    The left and right panels of Figure~\ref{fig:rbst-opt-strat1} depict a bank's robust optimal strategy and its worst-case scenario on the aggregate early liquidation size of other banks, respectively, as functions of $\gamma$ when the inverse demand function $Q$ is nonlinear; specifically $Q(x)=\exp(-0.2x)$. In this example, $\gamma^* \approx 0.7$, satisfying the condition in Theorem \ref{thm:xpm-non-mono}. Consistent with our theoretical results, the left panel displays the aforementioned non-monotonic behavior of $\xpm_i$ with the tipping point $\gamma^*$ and the upper bound $\xst_i/2$, and the right panel shows that under the worst-case scenario, higher $\gamma$ drives other banks to liquidate more at $t = 1$ when $\gamma \leq \gamma^*$, and to liquidate exclusively in the first period once $\gamma > \gamma^*$.

    \begin{figure}[t]
        {\includegraphics[width=0.5\textwidth]{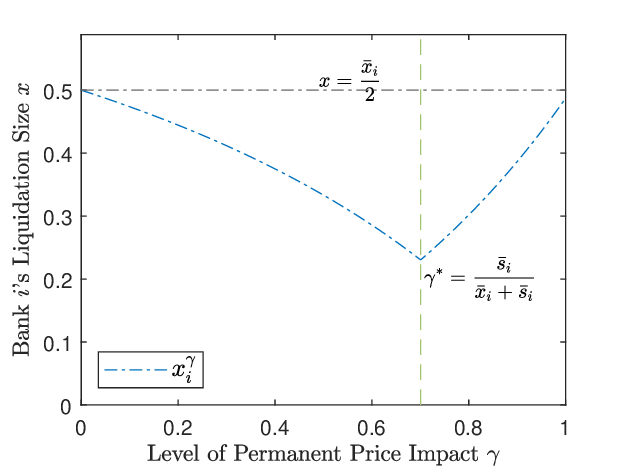}
            \includegraphics[width=0.5\textwidth]{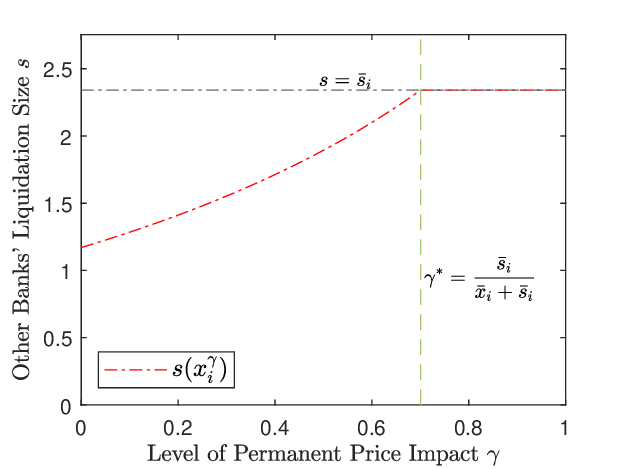}}
        \caption{A bank's robust optimal strategy and its worst-case scenario on the aggregate early liquidation size of other banks under various levels of permanent price impact when $\gamma^*$ is large.\label{fig:rbst-opt-strat1} We use a four-bank system with $i=1$, $L_1=\cdots=L_4 = 6$, $z_1=\cdots=z_4 = 1$, $w_1 = 0.3$, $w_2 = w_3 = 0.5$, and $w_4 = 0.8$. We set $Q(x) = \exp(-0.2x)$ and $s(x)=\argmin_{s \in [0, \sst_i]} g_i^\gamma(x, s)$, where $g_i^\gamma$ is defined as in \eqref{eq:ggamma}.}
    \end{figure}

    However, one can easily intuit that the said V-shaped dependence does not hold universally. Particularly, in an extreme case where bank $i$'s contribution to the price impact is significantly greater than the aggregate contribution of all other banks (i.e., $\xst_i\gg \sst_i$), the price impact would be predominantly driven by bank~$i$'s own liquidation. Hence, as $\gamma$ increases, the optimal response for bank $i$ would be to monotonically reduce its early liquidation size to alleviate the effects of persistent price impact. This intuition is rigorously verified in the following proposition.
    \begin{proposition}[Impact of $\xpm_i$: An Extreme Scenario]
        For any $i\in[m]$, if $Q''(\cdot)>0$ and $\gamma^*$ is sufficiently small, then
        % and define $\delta \coloneqq \min_{z \in [(\xst_i - \sst_i)/5, ~\xst_i + \sst_i]} \{-zQ''(z)/Q'(z) \}$. If $\sst_i/\xst_i \leq 2\delta/15$,
        the robust optimal strategy $\xpm_i$ decreases in $\gamma$ on $[0,1]$.
        \label{thm:rbst-dec}
    \end{proposition}
    In Figure~\ref{fig:rbst-opt-strat2}, we conduct an experiment similar to that in Figure~\ref{fig:rbst-opt-strat1}, but with $\gamma^* \approx 0.11$. In this case, the non-monotonicity vanishes; instead, $\xpm_i$ decreases monotonically with $\gamma$, validating Theorem~\ref{thm:rbst-dec}. The right panel of the figure mirrors the pattern in that of Figure \ref{fig:rbst-opt-strat1}: increasing with $\gamma$ until $\gamma \leq \gamma^*$ and plateauing at $\sst_i$ thereafter.

    % \begin{remark}
        % In practice, the parameter $\gamma$ is generally given. Then, the robust optimal strategy is mainly shaped by the relative magnitudes of $\xst_i$ and $\sst_i$. In extreme cases where $\sst_i \gg \xst_i$, implying that other banks liquidate extensively, it is likely that $\gamma \leq \sst_i/(\xst_i + \sst_i) \approx 1$. Thus, bank $i$ places more weight on mitigating the worst-case first-period price impacts caused by other banks than the induced permanent price impacts at $t=2$. Conversely, if $\sst_i \ll \xst_i$, suggesting that $\gamma \geq \sst_i/(\xst_i + \sst_i) \approx 0$, bank $i$ tailors its liquidation strategy by considering the price impacts caused by itself, particularly the persistent impact at $t=2$ induced by its own liquidation.
    % \end{remark}.
	
	% We provide the graphical illustrations of the behaviors characterized in Theorem \ref{thm:xpm-non-mono} and \ref{thm:rbst-dec}.  In Figure \ref{fig:rbst-opt-strat1}, . C In contrast, Figure \ref{fig:rbst-opt-strat2} considers the case in which
	
	% validate Theorem \ref{thm:xpm-non-mono} and \ref{thm:rbst-dec}, respectively.
	
	% The numerical results shown in Figure \ref{fig:rbst-opt-strat2} validate the monotonicity of the robust optimal strategy $\xpm_i$ with respect to the level of permanent impact $\gamma$ when $\sst_i$ is sufficiently small.
	\begin{figure}[t]
		{\includegraphics[width=0.5\textwidth]{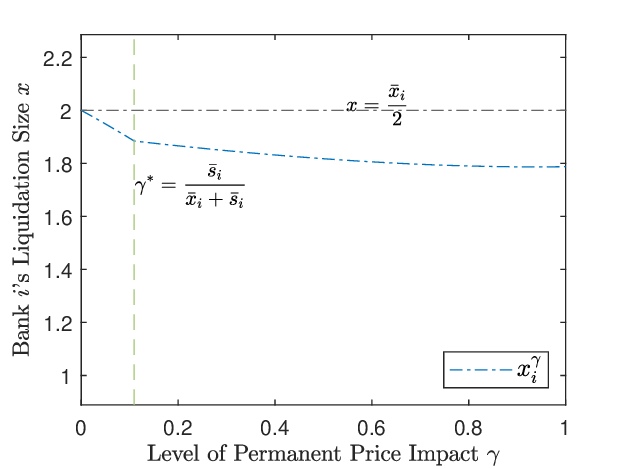}
			\includegraphics[width=0.5\textwidth]{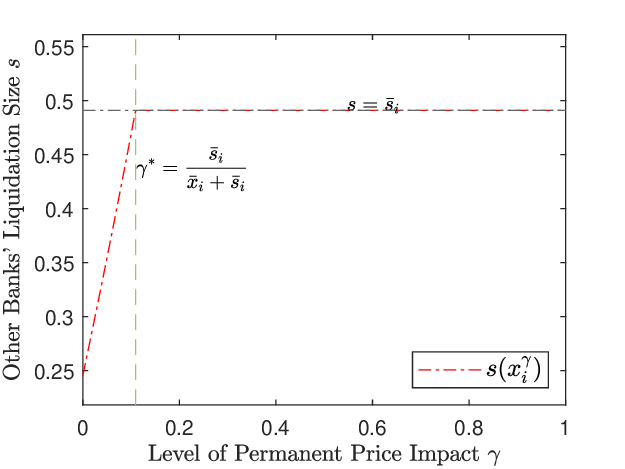}}
		\caption{A bank's robust optimal strategy and its worst-case scenario on the aggregate early liquidation size of other banks under various levels of permanent price impact when $\gamma^*$ is small. \label{fig:rbst-opt-strat2} We use a four-bank system with $i=1$, $L_1=\cdots=L_4 = 6$, $z_1 = 4$, $z_2=z_3=z_4=0.6$, $w_1 = 0.3$, $w_2 = w_3 = 0.5$, and $w_4 = 0.8$. We set $Q(x) = \exp(-0.2x)$ and $s(x)=\argmin_{s \in [0, \sst_i]} g_i^\gamma(x, s)$, where $g_i^\gamma$ is defined as in \eqref{eq:ggamma}.}
	\end{figure}

    \section{Robust Optimal Strategy with Interbank Exposures}	\label{sec:intb_liab}
		
    We now take into account debt exposures between banks and investigate how they affect the robust optimal strategy. Below, we present the associated model, which extends the framework introduced in \cite{CifuentesEtAl2005} and used in many follow-up studies~\citep{ChenEtAl2016,AminiEtAl2016c,WeberWeske2017,CapponiBernard2022}. Our extension accounts for (i) the possibility of early liquidation and (ii) each bank's liabilities to external entities beyond
the financial system. %into the framework while reverting to the assumption of no permanent impact, and extend the explicit characterization of the robust optimal strategy.
			
    Table~\ref{tab:init-bst} describes the initial balance sheet of each bank $i$ in an $m$-bank system operating over two periods. The difference from Table~\ref{tab:init-bst-noint} lies in that interbank assets $\{L_{ji}\}_{j\neq i}$ are added to the asset side, while total liabilities $L_i$ are now composed of liabilities to external entities (denoted by $b_i>0$) and interbank liabilities $\{L_{ij}\}_{j\neq i}$, i.e., $L_i = b_i+\sum_{j\neq i}L_{ij}$. Here, $L_{ij}$ denotes the nominal liabilities of bank $i$ owing to bank $j$, with $L_{ii}=0$, for all $i,j \in [m]$. We again assume that all balance sheet components are assumed to remain unchanged over the two periods ($t=1,2$) and all liabilities are cleared at $t=2$. %We still assume that $e_i < L_i$ for all $i \in [m]$, and banks must liquidate illiquid assets to repay their liabilities as far as possible.

%			 The rest of the financial system considered in this section remains as defined in Section \ref{sec:perm_impact}.
    \begin{table}[t]
        \centering\caption
        {Initial balance sheet of bank $i$ with interbank exposures before liquidation\label{tab:init-bst}}
        {\begin{tabular}{ccc}  % centered columns (4 columns)
                \toprule
                \textbf{Assets}  & &\textbf{Liabilities and Net Worth}   \\  % inserts table %heading
                \midrule
                Liquid assets $e_i$ & &External liabilities $b_i$   \\
                Interbank assets $L_{ji},  j \neq i$ & &Interbank liabilities $L_{ij},  j \neq i$    \\
                Illiquid assets $z_i$ with price $P$ & &Net worth $e_i+\sum_{j\neq i}L_{ji}+z_iP-L_i$     \\
                \bottomrule
        \end{tabular}\\
        {\footnotesize \sf The component $L_i$ represents bank $i$'s total liabilities, given by $L_i = b_i + \sum_{j\neq i}L_{ij}$.}}
    \end{table}

    % We assume that interbank assets are liquid as repayments are considered to be settled in cash. Then, the total liabilities and total liquid assets of bank $j$ are denoted by $L_j = b_j + \sum_{k \neq j}L_{jk}$ and $o_j = e_j + \sum_{k = 1}^m L_{kj}$, respectively. The initial net worth of bank $j$ is given by $w_j = e_j + \sum_{k=1}^{m}L_{kj} + z_j P - L_j$.
    It is also assumed that all liabilities are of equal seniority; if a bank's asset value at clearing is insufficient to cover total liabilities, its remaining assets are distributed to creditors on a \emph{pro rata} basis. To illustrate this, we introduce the relative liability matrix $\brLM = (\rLM_{ij})$, where $\rLM_{ij} = L_{ij}/L_i$ represents the proportion of bank $i$'s liabilities to bank $j$. %, satisfying that $\sum_{k=1}^m \rLM_{jk} = (L_j - b_j)/L_j < 1$. The assumption $_j < L_j$ for all $j \in [m]$ is maintained, and banks are required to liquidate illiquid assets to repay their liabilities to the greatest extent possible. As is commonly adopted in the literature\citep[see][]{EisenbergNoe2001, CifuentesEtAl2005, ChenEtAl2016}, we assume that banks' repayments are proportional to the notional values of their liabilities, which is implied by that all debt claims of each bank are of the same seniority.
    Then, if bank~$j$'s total debt repayment is $\pi_j$ for all $j\in[m]$, the total payment received by bank $i$ (i.e., the value of its interbank assets) becomes $\sum_{j=1}^m\rLM_{ji}\pi_j$ (or equivalently, $\brLM_i^\top\bpi$). Consequently, given that each bank~$i$ sells off $x_i$ units of illiquid assets at $t=1$, the market price $\psA$ of illiquid assets at $t=2$ and each bank~$i$'s clearing payment $\pi_j^\brLM$ are defined as the solution to the fixed-point equations, using the inverse demand function $Q$ satisfying Assumption~\ref{assump1} over two periods:\footnote{In this section, we set aside permanent price impact in order to focus exclusively on the impact of interbank exposures.}
    \begin{equation}\label{eq:p2-l-rLM}
        \left\{~\begin{aligned}
            & \psA = Q\left(\sum_{j = 1}^{m} \frac{\left(L_j - e_j - x_{j}Q\big(\bone^\top\bx\big) - \brLM_j^\top \bpi^\brLM\right)^+}{p_2^\brLM } \wedge \left(z_j - x_{j}\right) \right), 		\\
            & \pi_i^\brLM = L_i \wedge \left(e_i + x_i Q\big(\bone^\top\bx\big) + (z_i - x_i) \psA + \brLM_i^\top \bpi^\brLM \right), ~~ i = 1,\ldots,m.
        \end{aligned}\right.
    \end{equation}
    Similar to Section~\ref{sec:perm_impact}, the existence and uniqueness of the above solution is guaranteed by \cite{AminiEtAl2016c}. Each bank $i$'s liquidation size at $t=2$ is then given by
    \begin{equation}
        \yf_i^\brLM(\bx) \coloneqq \frac{\left(L_i - e_i - x_iQ(\bone^\top \bx) - \brLM_i^\top \bpi^\brLM \right)^+}{p_2^\brLM} \wedge \left(z_i - x_i\right).
        \label{eq:psij-rLM}
    \end{equation}
     %In particular, $\yf_j^\bO(\bx) = \yf_j(\bx)$, where $\yf_j(\bx) $ is the second-period liquidation size of bank $j$ given $\bx$ in the example of Section \ref{sec:preliminaries}. Compared to \eqref{p_l_bar-def} and \eqref{x_bar-def}, \eqref{eq:p2-l-rLM} and \eqref{eq:psij-rLM} reflect the evolution in asset positions after the first-period liquidation, where bank $j$'s liquid assets accumulate to $e_j + x_j Q(\sum_{j = 1}^m x_j)$, while leaving a residual illiquid asset position of $z_j - x_j$ units.
%			Let $\bgao = (e_1, \dots, e_m)^\top$, $\bL = (L_1, \dots, L_m)^\top$, $\bpi^\brLM = (\pi_1^\brLM, \dots, \pi_m^\brLM)$, $\byf^\brLM(\bx) = (\yf_1^\brLM(\bx), \dots, \yf_m^\brLM(\bx))^\top$ denote the vectors of external liquid assets, total liabilities, liability repayments, and the second-period liquidation sizes, respectively.
Finally, bank $i$'s liquid asset value at clearing is  defined as
    \begin{equation}\label{eq:lA}
        \lrLM_i(x_i, \bx_{-i}) \coloneqq e_i + x_i Q\big(\bone^\top\bx\big) + \yf_i^\brLM(\bx) Q\big(\bone^\top\byf^\brLM(\bx)\big) + \brLM_i^\top \bpi^\brLM.
    \end{equation}

The subsequent analyses proceed in parallel with Section~\ref{sec:perm_impact}; we (i) relax the liquidation size at clearing, (ii) formulate the robust optimization problem (i.e., the counterpart of \eqref{prob:maximin-perm}), (iii) verify the validity of this formulation, and (iv) derive the robust optimal strategy, which is the optimal solution to this problem. Although the theoretical results here mirror those in Section~\ref{sec:perm_impact}, the proofs are nontrivial extensions and differ significantly due to the presence of $\bpi^\brLM$, which is characterized by a fixed-point equation.

Let us denote $\xbar_i\coloneqq\yf_i^\brLM(\bzero)$ for each $i\in[m]$, which represents bank $i$'s liquidation size at clearing in the absence of early liquidation, and we assume $\bxbar > \bzero$ since otherwise liquidation is not required for some banks. We then establish a variant of Lemma~\ref{lemma:yfpm-ub} that gives a robust bound of the liquidation size at $t=2$:

    \begin{lemma}
        For any $\bx \leq \bxbar$, $\yf_{j}^\brLM(\bx) \leq \xbar_j - x_j$ for all $j\in[m]$.
        \label{lemma:yf-ub-intb}
    \end{lemma}
%			{\bf \itshape We prove Lemma \ref{lemma:yf-ub-intb} using a similar argument in the proof of Lemma \ref{lemma:yfpm-ub}. The proof of Lemma \ref{lemma:bx-ub-intb} is more complicated than that of Lemma \ref{lemma-xj1-bound} due to the nuanced dependence of repayments on liquidation size. We utilize the property of interbank repayments that they are $\ell_1$-nonexpansive in}

Based on the above lemma, given a realization $\bx \leq \bxbar$ of the early liquidation size, we can identify each bank $i$'s maximum possible income from fire sales at $t=1$ and $t=2$ as follows:
\begin{equation}
    \cNf_i(\bx) \coloneqq x_i Q\big(\bone^\top\bx\big) + \big(\xbar_i - x_i\big) Q\Big(\bone^\top\big(\bxbar - \bx\big)\Big).
    \label{op-cash-inter}
\end{equation}
This provides an upper bound on the sum of the middle two terms on the right-hand side of \eqref{eq:lA}, reflecting the worst-case scenario in which banks ultimately liquidate the same volume of illiquid assets as they would without early liquidation. Then, analogous to the construction of the second equation in~\eqref{eq:p2-l-rLM}, we can characterize the vector of clearing payments associated with this worst-case scenario, denoted by $\bwNf(\bcNf(\bx))$, as the solution to the following fixed-point equation:
    \begin{equation}
        \bwNf = \bL \wedge \left(\bgao + \bcNf(\bx) + \brLM^\top \bwNf\right),
        \label{wNf-fixed-point}
    \end{equation}
    where $\bcNf(\cdot) \coloneqq \big(\cNf_1(\cdot),  \dots,  \cNf_m(\cdot)\big)^\top$ and $\bwNf(\bcNf(\bx))$ is unique for any $\bx\in[\bzero,\bxbar]$ by Theorem~2 of~\cite{EisenbergNoe2001}.

    Consequently, the counterpart of \eqref{prob:maximin-perm} for the case with interbank exposures can be formulated as follows:
    \begin{equation}
        \begin{aligned}
            \max_{x_{i}}  \min_{\bx_{-i} } ~~& \ltN_i(x_i, \bx_{-i})
            \\  \textrm{s.t.}~~~&  \bzero \leq \bx \leq \bxbar,
        \end{aligned}
        \label{maximin-network}
    \end{equation}
    where $\ltN_i(x_i, \bx_{-i}) \coloneqq e_i + \cNf_i(\bx) + \brLM_i^\top\bwNf\big(\bcNf(\bx)\big)$ consists of the initial external liquid assets $e_i$, the maximum possible liquidation income $\cNf_i(\bx)$, and the worst-case interbank assets $\brLM_i^\top\bwNf(\bcNf(\bx))$. \iffalse Obviously, if $\brLM=\bO$, then $\bxbar=\bxst$, and $\ltN_i\left(x_i, \bx_{-i} \right)$ coincides with the objective of \eqref{maximin-prob} for all $i\in[m]$ and $\bx\in[\bzero,\bxst]$.

    Finally, to obtain bank $i$'s robust optimal strategy for the case with interbank liabilities, we formulate a variant of \eqref{maximin-prob} as follows:
    \begin{equation}
        \begin{aligned}
            \max_{x_{i}}  \min_{\bx_{-i} } ~~& \ltN_i\left(x_i, \bx_{-i} \right)
            \\  \textrm{s.t.}~~~&  \bzero \leq \bx \leq \bxbar.
        \end{aligned}
        \label{maximin-network}
    \end{equation}\fi
    In line with Proposition \ref{prop:eff-perm}, we demonstrate the validity of the above formulation in the following proposition, which addresses (Q2) in Section~\ref{sec:preliminaries} for the case of interbank exposures:
    \begin{proposition}[Validity of \eqref{maximin-network}]
        For any fixed $x_i \in [0, \xst_i]$, the following statements hold:
        \begin{enumerate}[label=(\roman*)]

            \item If $\min_{\bx_{-i} \in [\bzero,  \bxbar_{-i} ]} \ltN_i(x_i, \bx_{-i}) \geq L_i$, then
            $
            \min_{\bx_{-i} \in [\bzero,  \bxbar_{-i} ] } \lrLM_i(x_i, \bx_{-i}) = L_i.
            $

            \item If $\min_{\bx_{-i} \in [\bzero,  \bxbar_{-i}] } \ltN_i(x_i, \bx_{-i}) < L_i$, then
            $
            \min_{ \bx_{-i} \in [\bzero,  \bxbar_{-i} ]} \lrLM_i\left(x_i, \bx_{-i}\right) \geq \min_{\bx_{-i} \in [\bzero,  \bxbar_{-i} ]} \ltN_i\left(x_i, \bx_{-i}\right).
            $

        \end{enumerate}
        \label{prop:eff-intb}
    \end{proposition}
    % The same as Proposition \ref{prop:eff-perm}, Proposition \ref{prop:eff-intb} also states that the optimal solution of \eqref{maximin-network} achieves the optimality of \eqref{maximin-intb-org} if $\min_{\bx_{-i} \in [\bzero,  \bxst_{-i} ]} \ltN_i(x_i, \bx_{-i}) \geq L_i $, and otherwise it still maximizes a lower bound of bank $i$'s worst-case terminal liquid asset value. The proof of Proposition \ref{prop:eff-intb} proceeds along the same lines as that of Proposition \ref{prop:eff-perm}, while the involvement of the interbank repayments slightly complicates the analysis. The details are presented in Section \ref{sec:proof_intb}.

    A salient feature of the problem involving interbank exposures---distinguishing it from the previous case---is that each bank must not only avoid extreme price impacts on its own assets over the two periods but also seek to mitigate price impacts on other banks to increase the payments it receives from them. Furthermore, these payments are implicitly characterized by the fixed-point equation~\eqref{wNf-fixed-point}, making the objective function in \eqref{maximin-network} neither convex nor concave in $x_i$ or $\bx_{-i}$. This precludes the use of the minimax theorem and the approach of aggregating other banks' early liquidation sizes, significantly complicating the technical analysis required to identify the optimal solution to~\eqref{maximin-network}. Nevertheless, the following theorem successfully characterizes the robust optimal strategy for this problem, which turns out to take a surprisingly simple form, structurally analogous to that of the illustrative example in Section~\ref{sec:preliminaries}.  %Different from \eqref{prob:maximin-perm}, to achieve the optimality of~\eqref{maximin-network}, bank $i$ should not only reduce its exposure to extreme price impacts in both periods, but also should seek to temper the price impacts experienced by other banks to increase their interbank repayments. Interestingly, however, in a consistent form with the robust optimal strategy $\xtld_i$ of \eqref{maximin-prob}, the optimal solution to the outer problem of \eqref{maximin-network} remains to liquidate half of the associated single-period size. This result is formally described in the following theorem.
    \begin{theorem}[Robust Optimal Strategy with Interbank Exposures]
        For each $i \in [m]$, the robust optimal strategy $x_i^\brLM$ that solves \eqref{maximin-network} is uniquely determined as $x_i^\brLM = \xbar_i/2$.
        \label{thm-maximin-network}
    \end{theorem}

    %Due to the intertwined mapping structure of $\bwNf(\cdot)$, the objective $\ltN_i(x_i, \bx_{-i})$ becomes highly non-convex, and can no longer be expressed as a function of the aggregate liquidation size of other banks because the repayments depend on the full profile of other banks' strategies $\bx_{-i}$. Therefore, unlike in \eqref{maximin-prob} and \eqref{prob:maximin-perm}, the minimax theorem is not applicable. Despite this complication, $\ltN_i(x_i, \bx_{-i})$ remains symmetric around $(\xbar_i/2, \bxbar_{-i}/2)$. We therefore reformulate the inner problem of \eqref{maximin-network} in terms of the liquidation incomes $\bcNf(\bx)$, or simply $\bc$, as decision variables, which preserves the symmetry. Using the $\ell_1$-nonexpansiveness of $\bwNf(\cdot)$ \citep[see][Lemma 5]{EisenbergNoe2001}, we show that under the worst-case scenarios, the reduction in repayments received from other banks' is bounded by the aggregate losses in their liquidation incomes $\bc_{-i}$, which is further dominated by the losses in bank $i$'s own liquidation income $c_i$. By exploiting the symmetry, we demonstrate that the one-half liquidation strategy substantially reduces bank $i$'s worst-case losses from liquidation, and thus attains the optimality for the reformulated problem, even though it may incur minor losses in received repayments under the worst-case scenarios.
    Despite the aforementioned technical complications, the simple characterization of the robust optimal strategy is attributed to our finding that the sensitivity of $\bwNf(\bcNf(\bx))$ to $\bx$ is dominated by that of $\cNf_i(\bx)$ under bank $i$'s worst-case scenario. Consequently, the robust optimal strategy is primarily driven by a solution to the problem $\max_{x_i\in[0,\xbar_i]}\min_{\bx_{-i} \in [\bzero,  \bxbar_{-i}] } c_i(\bx)$. By applying symmetry and the first-order conditions, its solution is found to be $(\xbar_i/2, \bxbar_{-i}/2)$, leading to the above result.

    Practically speaking, if bank $i$ were to adopt an overly aggressive or conservative strategy for its early liquidation, the associated price impact under the worst-case scenario would generate negative externalities across the entire financial system, leading to significant losses in counterparties' liquidation income. Although a minority of these institutions may derive limited gains through asset fire sales, the resulting repayments from this group would be insufficient to counterbalance the aggregate shortfall. Hence, as bank $i$ would fail to hedge its own liquidations costs through incoming interbank payments, the robust optimal strategy converges toward a moderate and balanced early liquidation solution: selling half of the assets that would otherwise be liquidated under the no-early-liquidation benchmark. This finally addresses (Q1) in Section~\ref{sec:preliminaries} for the case involving interbank exposures.

    \section{Conclusion}
    \label{sec:conclusions}

    In this paper, we provide new insights into an individual bank's early liquidation strategy in the context of systemic risk. Using a stylized two-period model, we formulate a maximin problem that maximizes the worst-case value of the bank’s liquid assets at clearing, accounting for the uncertain liquidation decisions of counterparties. To tackle the analytical intractability of the said problem, we propose a relaxed version that is both practically valid and analytically tractable. Crucially, we show that this relaxation yields a unique, robust optimal strategy that admits a (semi-)closed-form expression. We further highlight the distinct effects of two key factors widely investigated in the systemic risk literature: permanent price impact and interbank exposures. Specifically, permanent price impact leads to structural changes in the robust optimal strategy, which typically exhibits a non-monotonic (V-shaped) dependence on the magnitude of the impact. In contrast, interbank exposures play a less significant role and do not drive such structural shifts.

    % We also apply the framework to the setting where banks are connected via interbank liabilities. We demonstrate that the relaxed formulation remains valid in this setting. Despite the intricate coupling between liquidation decisions and interbank repayments, the robust optimal strategy in the presence of interbank liabilities can be explicitly characterized, preserving the structural consistency with the no-interlinkage case.

    Several directions remain open for future research. Firstly, while our robust optimal strategy is explicitly characterized and offers insights into early liquidation based on the relaxed maximin problem, it remains inconclusive whether a tighter yet analytically tractable relaxation exists. Addressing this question may require more precise and explicit approximations of the dependence among liquidation sizes across two periods. Secondly, although our framework accommodates permanent price impact and interbank liabilities individually, integrating both features into a unified framework poses substantial analytical challenges such as nonconvexity and asymmetry. These issues prevent a direct extension of our current analysis and hence represent a promising path for subsequent studies. Lastly, we develop the two-period liquidation strategy without delving into the quantification of other sources of financing risks. In practice, institutions generally evaluate various risk factors comprehensively. For example, apart from liquidation, illiquid assets can also be used in collateralized borrowing. Therefore, further investigation of a versatile financing strategy would be an interesting research direction.

    \appendix

    \section{Proofs of Theoretical Results in Section \ref{sec:perm_impact}}    \label{sec:proof_perm}
    We first introduce several inequalities that will be useful in the proofs of our theoretical results. Specifically, for all $x > 0,  y \geq 0$ with $x + y < \ztot$, let $f(x, y) = xQ(x + y)$ where $Q$ is the inverse demand function. Then, under Assumption~\ref{assump1}, it is straightforward to check that
    \begin{subequations}
        \begin{align}
            & f_x(x, y) \coloneqq \frac{\partial f(x, y)}{\partial x} = Q(x + y) + xQ'(x + y) > 0,
            \label{useful-ineq-x}
            \\ & f_y(x, y) \coloneqq \frac{\partial f(x, y)}{\partial y} = xQ'(x + y) < 0,
            \label{useful-ineq-y}
            \\ & f_{xx}(x, y) \coloneqq \frac{\partial^2 f(x, y)}{\partial x^2} = 2Q'(x + y) + xQ''(x + y) < 0,
            \label{useful-ineq-xx}
            % \\ & f_{yy}(x, y) \coloneqq \frac{\partial^2 f(x, y)}{\partial y^2} = xQ''(x + y) \geq 0,
            % \label{useful-ineq-yy}
            \\ & f_{xy}(x, y) \coloneqq \frac{\partial^2 f(x, y)}{\partial x \partial y} = Q'(x + y) + xQ''(x + y) < 0.
            \label{useful-ineq-xy}
        \end{align}
    \end{subequations}
\iffalse
    \proof{Proof of Lemma \ref{lemma-xj1-bound}.} We prove the statement by contrapositive. Let $\Tset = \{j \in [m]:  \xst_j < x_{j} \}$. Assume that $\Tset \neq \emptyset$. Then, for all $j \in \Tset$, $\xst_j < z_j$ since $x_{j} \leq z_j$. Hence, by \eqref{single_x}, we obtain
        $$
            \xst_j Q\left(\sum_{k = 1}^{m} \xst_k \right) = \xst_j \pst = L_j - e_j,
        $$
        which implies the following relationship:
        \begin{align}
            \sum_{j \in \Tset} \left(L_j - e_j \right) & = \sum_{j \in \Tset} \xst_j Q\left(\sum_{k = 1}^{m} \xst_k \right)		\nonumber
            \\ & < \sum_{j \in \Tset} x_{j} Q\left(\sum_{j \in \Tset} x_{j} + \sum_{k \in \Tset^\mathrm{C}} \xst_k \right)
            \label{lemma-xj1-bound-ineq1}
            \\ & \leq \sum_{j \in \Tset} x_{j} Q\left(\sum_{k = 1}^{m} x_{k} \right)
            \label{lemma-xj1-bound-ineq2}
        \end{align}
        where \eqref{lemma-xj1-bound-ineq1} is obtained by \eqref{useful-ineq-x}, and \eqref{lemma-xj1-bound-ineq2} holds by \eqref{useful-ineq-y}. Therefore, there exists some $i \in \Tset$ such that
        $$
            x_{i} Q\left(\sum_{k = 1}^{m}x_{k} \right) > L_i - e_i,
        $$
        and the result follows.
    \endproof
\fi
		\proof{Proof of Lemma \ref{lemma:yfpm-ub}.} Let $\pst \coloneqq Q(\bone^\top \bxst)$. By \eqref{eq:p2-pm} and \eqref{eq:xj2-pm}, it is easy to check that
        \begin{equation}
            \xst_j \pst = \yfpm_j(\bzero) Q\left(\bone^\top \yfpm_j(\bzero)\right) = (L_j - e_j) \wedge (z_j \pst) \leq L_j - e_j.
            \label{ineq:xst-pst-ub}
        \end{equation}
        Assume by contradiction that $\pst > \ppm_2$, where $\ppm_2$ is defined as in  \eqref{eq:p2-pm}. Then, we first observe that $p_1\coloneqq Q(\bone^\top \bx)\geq Q(\bone^\top \bxst)=\pst>\ppm_2$ since $Q(\cdot)$ is strictly decreasing by Assumption \ref{assump1}(ii). This implies that for any $\gamma\in[0,1]$,
		\begin{equation}
			\left(L_j - e_j - x_j p_1 \right)^+ + \gamma x_j \ppm_2 = \left[L_j - e_j - x_j\left(p_1 - \gamma \ppm_2 \right) \right] \vee \left(\gamma x_j \ppm_2 \right) \leq (L_j - e_j) \vee (x_j \pst) = L_j - e_j,
			\label{ineq:ppm-pst-sf}
		\end{equation}
        where the last equality holds by \eqref{ineq:xst-pst-ub} and
    $u \vee v \coloneqq \max\{u, v\}$ for any $u,v\in\bbR$.

		Let
        \begin{equation}
		\begin{aligned}
			\rho_0 = Q\left(\sum_{j = 1}^m \frac{L_j - e_j}{\ppm_2} \wedge z_j \right)~~\text{and}~~ \rho_{k+1} = Q\left(\sum_{j = 1}^m \frac{L_j - e_j}{\rho_k} \wedge z_j \right)~\text{for}~ k = 0,1,2,\ldots
		\end{aligned}
        \end{equation}
		Then, by \eqref{ineq:ppm-pst-sf} and Assumption \ref{assump1}(ii), we have
		\begin{align*}
			\ppm_2 & = Q\left(\gamma \sum_{j = 1}^m x_j + \sum_{j = 1}^m \frac{\left(L_j - e_j - x_j p_1 \right)^+}{\ppm_2} \wedge (z_j - x_j) \right)
			\\ & = Q\left(\sum_{j = 1}^m \frac{\left(L_j - e_j - x_j p_1 \right)^+ + \gamma x_j \ppm_2}{\ppm_2} \wedge \big(z_j - (1 - \gamma)x_j\big) \right)
			\\ & \geq Q\left(\sum_{j = 1}^m \frac{L_j - e_j}{\ppm_2} \wedge z_j \right)
			\\ & = \rho_0.
		\end{align*}
		Furthermore, if $\rho_k \leq \ppm_2$ for some $k \geq 0$, then
        \begin{equation}
            \rho_{k + 1} = Q\left(\sum_{j = 1}^m \frac{L_j - e_j}{\rho_k} \wedge z_j \right) \leq Q\left(\sum_{j = 1}^m \frac{L_j - e_j}{\ppm_2} \wedge z_j \right) = \ppm_2.
        \end{equation}
        Hence, by induction, we obtain that $\rho_{k} \leq \ppm_2$ for all $k$, which suggests that $\pst \leq \ppm_2$ since $\rho_k \rightarrow \pst$ as $k$ increases. This contradicts the assumption and indicates that $\pst \leq \ppm_2$.

        Consequently, %since $\pst \leq \ppm_2$, $x_j \leq \xst_j$, and $\xst_j \pst \leq L_j - e_j$ implied by the definition \eqref{single_x},
        the following holds for any $\gamma \in [0, 1]$:
		\begin{align*}
			x_j + \yfpm_j(\bx) & = x_j + \frac{\left(L_j - e_j - x_{j}p_1 \right)^+}{\ppm_2} \wedge \left(z_j - x_{j}\right)
			\\ & \leq x_j + \frac{\left(L_j - e_j - x_{j}p_1 \right)^+}{\pst} \wedge \left(z_j - x_{j}\right)
			\\ & = \frac{\big(L_j - e_j - x_{j}(p_1 - \pst)\big) \vee \left(x_j \pst \right)}{\pst} \wedge z_j
			\\ & \leq \frac{\left(L_j - e_j \right) \vee \left(\xst_j \pst \right) }{\pst} \wedge z_j
			\\ &  = \frac{L_j - e_j}{\pst} \wedge z_j
			\\ & = \xst_j,
		\end{align*}
		which completes the proof.
		\endproof

        % {\color{red}
        \proof{Proof of Proposition \ref{prop:eff-perm}.} Fix $x_i \leq \xst_i$ and let $\bu^\gamma \coloneqq \argmin_{\bx_{-i} \in [\bzero, \bxst_{-i}]} \lpm_i(x_i, \bx_{-i})$. By the definition of $\bxst$ and \eqref{useful-ineq-x}, we observe that $e_i + x_i Q(x_i + \bone^\top \bxst_{-i}) \leq e_i + \xst_i Q(\bone^\top \bxst) \leq L_i$. Thus, we have
		\begin{equation}
        \begin{aligned}
			\lpm_i(x_i, \bu^\gamma) & \leq \lpm_i(x_i, \bxst_{-i})
			\\ & = e_i + x_i Q\big(x_i + \bone^\top \bxst_{-i}\big) + \left(L_i - e_i - x_i Q\big(x_i + \bone^\top \bxst_{-i} \big)\right) \wedge \left((z_i - x_i)\ppm_2 \right)
			\\ & \leq e_i + x_i Q\big(x_i + \bone^\top \bxst_{-i}\big) + L_i - e_i - x_i Q\big(x_i + \bone^\top \bxst_{-i} \big)
			\\ & = L_i,
		\end{aligned}\label{eq:lLineq}
		\end{equation}
		where the first equality holds by \eqref{eq:xj2-pm} to \eqref{def:l-perm}.
		
%		By the definition of $\bxst$, it is easy to check that $$\min_{ \bx_{-i} \in [\bzero,  \bxst_{-i} ]}\lfpm_i(x_i, \bx_{-i}) \leq \lfpm_i(x_i, \bxst_{-i}) \leq L_i.$$
		% For the first statement, by contrapositive, it suffices to show that if $\lpm_i(x_i, \bu^\gamma) < L_i$, then $\min_{ \bx_{-i} \in [\bzero,  \bxst_{-i} ]}\lfpm_i(x_i, \bx_{-i}) < L_i$.
        We now assume $\lpm_i(x_i, \bu^\gamma) < L_i$ and define $\bv^\gamma \in \bbR^m$ with $v_i^\gamma = x_i$ and $\bv_{-i}^\gamma = \bu^\gamma$. Then, by \eqref{def:l-perm}, it follows that
		\begin{equation}
			\yfpm_{i}(\bv^\gamma) = \frac{\lpm_i(x_i, \bu^\gamma) - e_i - x_i Q\left(\bone^\top \bv^\gamma\right)}{Q\big(\gamma \bone^\top \bv^\gamma + \sum_{j = 1}^m \yfpm_j(\bv^\gamma) \big)} < \frac{\big(L_i - e_i - x_i Q(\bone^\top \bv^\gamma)\big)^+}{Q\big(\gamma \bone^\top \bv^\gamma + \sum_{j = 1}^m \yfpm_j(\bv^\gamma) \big)}.
		\end{equation}
		Hence, $\yfpm_i(\bv^\gamma) = z_i - x_i$ by \eqref{eq:xj2-pm}. Since $\xst_i \leq z_i$ by the definition of $\xst_i$ and $\yfpm_j(\bv^\gamma) \leq \xst_j - x_j$ for all $j$ by Lemma \ref{lemma:yfpm-ub}, we obtain
        \begin{equation}\label{eq:ltildel}
		\begin{aligned}
			\min_{ \bx_{-i} \in [\bzero,  \bxst_{-i} ]}\lfpm_i(x_i, \bx_{-i})  & \leq \lfpm_i\left(x_{i},  \bu^\gamma \right)
			\\ & = e_i + x_i Q \big(\bone^\top \bv^\gamma \big) + \left(\xst_i - x_i \right) Q\Big(\gamma \bone^\top \bv^\gamma + \bone^\top \big(\bxst - \bv^\gamma \big) \Big)
			\\ & \leq e_i + x_i Q \big(\bone^\top \bv^\gamma \big) + (z_i - x_i) Q\Bigg(\gamma \bone^\top \bv^\gamma + \sum_{j = 1}^{m} \yfpm_{j}(\bv^\gamma) \Bigg)
			\\ & = e_i + x_i Q \big(\bone^\top \bv^\gamma \big) + \yfpm_i(\bv^\gamma) Q\Bigg(\gamma \bone^\top \bv^\gamma + \sum_{j = 1}^{m} \yfpm_{j}(\bv^\gamma) \Bigg)
			\\ & = \lpm_i(x_i, \bu^\gamma).
		\end{aligned}
        \end{equation}
		Consequently, $\min_{ \bx_{-i} \in [\bzero,  \bxst_{-i} ]}\lfpm_i(x_i, \bx_{-i})<L_i$, which demonstrates the first statement by contrapositive.
		
		The second statement is straightforward if $\lpm_i(x_i, \bu^\gamma) = L_i $. Otherwise, it holds by \eqref{eq:ltildel}.
		\endproof
\proof{Proof of Theorem~\ref{thm:opt-strat-perm}.} Fix $i\in[m]$. For any $x \in [0, \xst_i]$ and $s \in [0, \sst_i]$, we define $$g_i^\gamma(x, s) \coloneqq  e_i + x Q(x + s) + (\xst_i - x) Q(\xst_i + \sst_i - (1 - \gamma)(x + s)).$$ Then, \eqref{prob:maximin-perm} can be rewritten as
	\begin{equation}
		\max_{x \in [0, \, \xst_i]} \, \min_{s \in \left[0, \, \sst_i \right]} \; g_i^\gamma(x, \, s).
		\label{equiv-maxmin}
	\end{equation}
    Using this reformulation, we prove the theorem in three steps.
		% \begin{equation}
		% 	\frac{(1 - \gamma)\xst_i}{2 - \gamma} \leq \xfe_i^\gamma \leq \frac{\xst_i - (1 - \gamma)\sst_i}{2 - \gamma}.
		% 	\label{ineq:xpm-lb-ub-biggm}
		% \end{equation}
		% Otherwise,
		% \begin{equation}
		% 	\xfe_i^\gamma \leq \frac{(1 - \gamma)\xst_i}{2 - \gamma}.
		% 	\label{ineq:xpm-lb-ub-smlgm}
		% \end{equation}
		% Therefore, for any $\gamma \in [0, 1]$,
		% \begin{equation}
		% 	\frac{(1 - \gamma)\xst_i}{2 - \gamma} \vee \xfe_i^\gamma + \frac{\sst_i + \gamma \xst_i}{2 - \gamma} \wedge \sst_i \leq \frac{\xst_i + \sst_i}{2 - \gamma}.
		% \end{equation}
		% \label{lemma:xpm-lb-ub}
	% \end{lemma}
    %
    % \proof{Proof of Lemma \ref{lemma:xpm-lb-ub}}
	
    \textbf{Step 1.} Let $\xfe_i^\gamma$ be the value satisfying \eqref{eq:xpm-fst-ord} with respect to $x$. We establish its existence and uniqueness given $\gamma \geq \sst_i/(\xst_i + \sst_i)$. By Assumption \ref{assump1}(iv) and \eqref{useful-ineq-xx}, we observe that $g_i^\gamma(\cdot,\sst_i)$ is strictly concave on $[0, \xst_i]$. Also, it is easy to see that \eqref{eq:xpm-fst-ord} holds if and only if
	\begin{align}\label{eq:gxprime}
		\frac{\partial g_i^\gamma}{\partial x}\left(\xfe_i^\gamma, \sst_i \right) =0. %& = Q\left(\xfe_i^\gamma + \sst_i\right) + \xfe_i^\gamma Q'\left(\xfe_i^\gamma + \sst_i \right) - Q\left(\xst_i + \gamma \sst_i - (1 - \gamma) \xfe_i^\gamma \right)
		% \\ & ~~~~ - \left(1 - \gamma \right) \left(\xst_i - \xfe_i^\gamma \right) Q'\left(\xst_i + \gamma \sst_i - (1 - \gamma) \xfe_i^\gamma \right)
		% \\ & = 0.
	\end{align}
	% Note that for any $x \in [0, \xst_i]$,
	% \begin{align*}
	% 	\frac{\partial^2 g_i}{\partial x^2}\left(x, \sst_i \right) & = 2Q'(x + \sst_i) + xQ''(x + \sst_i) + 2(1 - \gamma)Q'\left(\xst_i + \gamma \sst_i - (1 - \gamma) x \right)
	% 	\\ & ~~~~ + (1 - \gamma)^2 (\xst_i - x) Q''\left(\xst_i + \gamma \sst_i - (1 - \gamma) x \right)
	% 	\\ & \leq 2Q'(x + \sst_i) + xQ''(x + \sst_i)
	% 	\\ & ~~~~ + (1 - \gamma) \left[2Q'\left(\xst_i + \gamma \sst_i - (1 - \gamma) x \right) + (\xst_i - x) Q''\left(\xst_i + \gamma \sst_i - (1 - \gamma) x \right) \right]
	% 	\\ & < 0
	% \end{align*}
	% by Assumption \ref{assump1}(iv) and \eqref{useful-ineq-xx}. Therefore, it suffices to show that if $\gamma > \sst_i/(\xst_i + \sst_i)$, we have
	% \begin{equation}
	% 	\frac{\partial g_i}{\partial x}\left(\frac{(1 - \gamma)\xst_i}{2 - \gamma}, \sst_i \right) \geq 0, ~~~~ ~~~~ \frac{\partial g_i}{\partial x}\left(\frac{\xst_i - (1 - \gamma)\sst_i}{2 - \gamma}, \sst_i \right) \leq 0,
	% 	\label{ineq:gi-x-fd-2p}
	% \end{equation}
	% and if $\gamma \leq \sst_i/(\xst_i + \sst_i)$, we have
	% \begin{equation}
	% 	\frac{\partial g_i}{\partial x}\left(0, \sst_i \right) \geq 0, ~~~~ ~~~~ \frac{\partial g_i}{\partial x}\left(\frac{(1 - \gamma)\xst_i}{2 - \gamma}, \sst_i \right) \leq 0.
	% \end{equation}
	Furthermore, a straightforward calculation yields that
    \begin{equation}
	\begin{aligned}
		\frac{\partial g_i^\gamma}{\partial x}\left(\frac{(1 - \gamma)\xst_i}{2 - \gamma}, \sst_i \right)
         & = Q\left(\frac{(1 - \gamma)\xst_i}{2 - \gamma} + \sst_i \right) + \frac{(1 - \gamma)\xst_i}{2 - \gamma} Q'\left(\frac{(1 - \gamma)\xst_i}{2 - \gamma} + \sst_i \right)
		\\ & ~~~~ - Q\left(\frac{(1 - \gamma)\xst_i}{2 - \gamma} + \gamma (\xst_i + \sst_i) \right) - \frac{(1 - \gamma)\xst_i}{2 - \gamma}Q'\left(\frac{(1 - \gamma)\xst_i}{2 - \gamma} + \gamma (\xst_i + \sst_i) \right),
	\end{aligned}
    \label{eq:gi-dx-xgm}
    \end{equation}
    and thus, by \eqref{useful-ineq-xy}, this value is non-negative if and only if $\gamma \geq \sst_i/(\xst_i + \sst_i)$. Similarly, if $\gamma > \sst_i/(\xst_i + \sst_i) $, by Assumption \ref{assump1}~(iv), we obtain
		\begin{equation}
			\begin{aligned}
				& \frac{\partial g_i^\gamma}{\partial x}\left(\frac{\xst_i}{2}-\frac{(1-\gamma)\sst_i}{2(2-\gamma)}, \sst_i\right)
				\\ & = Q\left(\frac{\xst_i}{2} + \frac{(3 - \gamma)\sst_i}{2(2-\gamma)} \right) + \left(\frac{\xst_i}{2} - \frac{(1 - \gamma)\sst_i}{2(2-\gamma)} \right)  Q'\left(\frac{\xst_i}{2} + \frac{(3 - \gamma)\sst_i}{2(2-\gamma)} \right) - Q\left(\frac{(1 + \gamma)\xst_i}{2} + \frac{(1+2\gamma-\gamma^2)\sst_i}{2(2 - \gamma)} \right)
				\\ & ~~~~ - \frac{(1 - \gamma)(2-\gamma)\xst_i + (1 - \gamma)^2\sst_i}{2(2 - \gamma)} Q'\left(\frac{(1 + \gamma)\xst_i}{2} + \frac{(1+2\gamma-\gamma^2)\sst_i}{2(2 - \gamma)}\right)
				\\ & \leq \left(\frac{\xst_i}{2} - \frac{(1 - \gamma)\sst_i}{2(2-\gamma)} \right)  Q'\left(\frac{\xst_i}{2} + \frac{(3 - \gamma)\sst_i}{2(2-\gamma)} \right) - \frac{\gamma\xst_i - (1-\gamma)\sst_i}{2}  Q'\left(\frac{\xst_i}{2} + \frac{(3 - \gamma)\sst_i}{2(2-\gamma)} \right)
				\\ & ~~~~ - \frac{(1 - \gamma)(2-\gamma)\xst_i + (1 - \gamma)^2\sst_i}{2(2 - \gamma)} Q'\left(\frac{(1 + \gamma)\xst_i}{2} + \frac{(1+2\gamma-\gamma^2)\sst_i}{2(2 - \gamma)} \right)
				\\ & \leq \left(\frac{\xst_i}{2} - \frac{(1 - \gamma)\sst_i}{2(2-\gamma)} \right)  Q'\left(\frac{\xst_i}{2} + \frac{(3 - \gamma)\sst_i}{2(2-\gamma)} \right) - \frac{\gamma\xst_i - (1-\gamma)\sst_i}{2}  Q'\left(\frac{\xst_i}{2} + \frac{(3 - \gamma)\sst_i}{2(2-\gamma)} \right)
				\\ & ~~~~ - \frac{(1 - \gamma)(2-\gamma)\xst_i + (1 - \gamma)^2\sst_i}{2(2 - \gamma)} Q'\left(\frac{\xst_i}{2} + \frac{(3 - \gamma)\sst_i}{2(2-\gamma)} \right)
				\\ & = 0,
			\end{aligned}
			\label{ineq:gi-p1-ub2}
		\end{equation}
		where the second inequality holds since $$
		\frac{\xst_i}{2} + \frac{(3 - \gamma)\sst_i}{2(2-\gamma)} < \frac{(1 + \gamma)\xst_i}{2} + \frac{(1+2\gamma-\gamma^2)\sst_i}{2(2 - \gamma)}~~\text{for all}~\gamma\geq\sst_i/(\xst_i+\sst_i).
		$$
		
    \iffalse
    we obtain
	\begin{align*}
		\frac{\partial g_i^\gamma}{\partial x}\left(\frac{\xst_i - (1 - \gamma)\sst_i}{2 - \gamma}, \sst_i \right) %& = Q\left(\frac{\xst_i  + \sst_i}{2 - \gamma} \right) + \frac{\xst_i - (1 - \gamma)\sst_i}{2 - \gamma} Q'\left(\frac{\xst_i  + \sst_i}{2 - \gamma} \right) - Q\left(\frac{\xst_i  + \sst_i}{2 - \gamma} \right)
		%\\ & ~~~~ - \frac{(1 - \gamma)^2(\xst_i + \sst_i)}{2 - \gamma}Q'\left(\frac{\xst_i  + \sst_i}{2 - \gamma} \right)
		& = \big(\gamma \xst_i + (1 - \gamma) \sst_i \big)Q'\left(\frac{\xst_i  + \sst_i}{2 - \gamma} \right)
		 < 0~~\text{for all}~\gamma\in[0,1].
	\end{align*}
    \fi
     By combining all the above results and the continuity of $\partial g_i^\gamma/\partial x(\cdot, \sst_i)$, $\xfe_i^\gamma$ exists and is unique if $\gamma \geq \sst_i/(\xst_i + \sst_i)$, and we have
    \begin{equation}
        \frac{(1 - \gamma)\xst_i}{2 - \gamma} \leq \xfe_i^\gamma \leq \frac{\xst_i}{2}-\frac{(1-\gamma)\sst_i}{2(2-\gamma)}~~\text{if}~~\gamma \geq \frac{\sst_i}{\xst_i + \sst_i}.
        \label{ineq:xpm-lb-ub-smlgm}
    \end{equation}
    \iffalse
    we have
		\begin{equation}
		\left\{\begin{aligned}	&\frac{(1 - \gamma)\xst_i}{2 - \gamma} \leq \xfe_i^\gamma \leq \frac{\xst_i}{2}-\frac{(1-\gamma)\sst_i}{2(2-\gamma)}&&~\text{if}~ \gamma > \frac{\sst_i}{\xst_i + \sst_i};\\
			&\xfe_i^\gamma \leq \frac{(1 - \gamma)\xst_i}{2 - \gamma}&&~\text{otherwise.}
            \end{aligned}\right.
			\label{ineq:xpm-lb-ub-smlgm}
		\end{equation}
    \fi
    %     This implies that for any $\gamma \in [0, 1]$,
    % \begin{equation}
    %     \frac{(1 - \gamma)\xst_i}{2 - \gamma} \vee \xfe_i^\gamma + \frac{\sst_i + \gamma \xst_i}{2 - \gamma} \wedge \sst_i \leq \frac{\xst_i + \sst_i}{2 - \gamma}.
    % \end{equation}
    % Then the proof is complete.
	% \endproof

    % \proof{Proof of Theorem \ref{thm:opt-strat-perm}.}
	% For any $x \in [0, \xst_i]$, $s \in [0, \sst_i]$, define $$g_i^\gamma(x, s) \coloneqq  e_i + x Q(x + s) + (\xst_i - x) Q(\xst_i + \sst_i - (1 - \gamma)(x + s)).$$ Then, \eqref{maximin-prob} can be rewritten as
	% \begin{equation}
	% 	\max_{x \in [0, \, \xst_i]} \, \min_{s \in \left[0, \, \sst_i \right]} \; g_i^\gamma(x, \, s).
	% 	\label{equiv-maxmin}
	% \end{equation}
	% Hence, we aim to show that $\xpm_i = (1 - \gamma)\xst_i/(2 - \gamma) \vee \xfe_i^\gamma$ is the optimal solution to \eqref{equiv-maxmin}.
	
	\textbf{Step 2.} Let
    $$
    z^\gamma\coloneqq \left\{\begin{aligned}	&\frac{(1 - \gamma)\xst_i}{2 - \gamma} &&~\text{if}~ \gamma \leq \frac{\sst_i}{\xst_i + \sst_i};\\
			&\xfe_i^\gamma &&~\text{otherwise.}
            \end{aligned}\right.
    $$
    \iffalse
    $$
    z^\gamma\coloneqq\frac{(1 - \gamma)\xst_i}{2 - \gamma} \vee \xfe_i^\gamma.
    $$
    \fi
    For each $x\in[0,\xst_i]$, define $\Smin(x) = \argmin\limits_{s \in [0, \sst_i]} g_i^\gamma(x, s)$ and
	\begin{equation}
		s(x) = \left\{\begin{aligned}
			&\inf \Smin(x), && \text{if } x < z^\gamma; \\
			&\frac{\sst_i + \gamma\xst_i}{2 - \gamma} \wedge \sst_i, && \text{if } x = z^\gamma; \\
			&\sup \Smin(x), && \text{if } x > z^\gamma.
		\end{aligned}\right.
		\label{def:s-fun}
	\end{equation}
	% Note that for any $\gamma \in [0, 1]$,
	% \begin{equation}
	% 	\gamma \leq \frac{\sst_i}{\xst_i + \sst_i} ~~~~ \Longleftrightarrow ~~~~ \frac{\sst_i + \gamma \xst_i}{2 - \gamma} \leq \sst_i.
	% 	\label{ineq:wrst-s}
	% \end{equation}
	If $\gamma > \sst_i/(\xst_i + \sst_i)$, Assumption \ref{assump1}(ii) and \eqref{ineq:xpm-lb-ub-smlgm} lead to the following relationship:
    \begin{equation}
    \begin{aligned}
		\frac{\partial}{\partial s} g_i^\gamma\big(z^\gamma, s(z^\gamma)\big) &= \frac{\partial}{\partial s} g_i^\gamma( \xfe_i^\gamma, \sst_i ) \\
        & =  \xfe_i^\gamma Q'\left( \xfe_i^\gamma + \sst_i \right) - \left(1 - \gamma\right)\left(\xst_i -  \xfe_i^\gamma \right) Q'\left(\xst_i + \sst_i - (1 - \gamma)  (\xfe_i^\gamma+\sst_i) \right)		\notag
		\\ & = Q\left(\xst_i + \sst_i - (1 - \gamma)\left( \xfe_i^\gamma + \sst_i\right) \right) - Q\left( \xfe_i^\gamma + \sst_i \right)		\notag
%		\\ & < \frac{\xst_i}{\xst_i + \sst_i} Q\left(\xst_i + \sst_i - \frac{\xst_i}{\xst_i + \sst_i}\left(\xtld_i + \sst_i\right) \right) - Q\left(\xtld_i + \sst_i \right)		\notag
		\\ & \leq 0,
%		\label{ineq:gi-derv-s-min1}
	\end{aligned}
    \end{equation}
    where the third equality stems from \eqref{eq:xpm-fst-ord} and the inequality holds by observing that $\xfe_i^\gamma \leq \xst_i/2 - (1-\gamma)\sst_i/[2(2-\gamma)] < ({\xst_i - (1 - \gamma)\sst_i})/({2 - \gamma})$ if $\gamma > \sst_i/(\xst_i + \sst_i) $.
    % replacing $\xfe_i^\gamma$ with $({\xst_i - (1 - \gamma)\sst_i})/({2 - \gamma})$.
    On the other hand, if $\gamma \leq \sst_i/(\xst_i + \sst_i)$, then \eqref{ineq:xpm-lb-ub-smlgm} suggests that
	% \begin{equation}
	% 	s(z^\gamma) = \begin{cases}
	% 		 \dfrac{\sst_i + \gamma\xst_i}{2 - \gamma}, ~~ & ~~ \text{if } \gamma \leq \dfrac{\sst_i}{\xst_i + \sst_i}, \\[10pt]
	% 		\sst_i, ~~ & ~~ \text{otherwise}.
	% 	\end{cases}
	% 	\label{eq:wrst-s}
	% \end{equation}
 %    Then, $s([(1 - \gamma)\xst_i/(2 - \gamma)] \vee \xfe_i^\gamma) \in \Smin([(1 - \gamma)\xst_i/(2 - \gamma)] \vee \xfe_i^\gamma)$ and thus $s(x) \in \Smin(x)$ for all $x$ since $g_i^\gamma(x, s)$ is continuous in $s$ for each $x$ and it can be obtained that
    \begin{equation}
	\begin{aligned}
		% & \frac{\partial^2}{\partial s^2} g_i\left(x, s \right) \geq 0, ~\text{ for all } (x, s) \in [0, \xst_i] \times [0, \sst_i],
		\frac{\partial}{\partial s} g_i^\gamma\big(z^\gamma, s(z^\gamma) \big) &= \frac{\partial}{\partial s} g_i^\gamma\left(\frac{(1 - \gamma)\xst_i}{2 - \gamma}, \frac{\sst_i + \gamma \xst_i}{2 - \gamma} \right)\\
        &=\frac{(1 - \gamma)\xst_i}{2 - \gamma}Q'\left(\frac{\xst_i+\sst_i}{2 - \gamma} \right) - \frac{(1 - \gamma)\xst_i}{2 - \gamma}Q'\left(\xst_i + \sst_i - (1 - \gamma) \frac{\xst_i+\sst_i}{2 - \gamma} \right)\\
        &=0.
	\end{aligned}
    \end{equation}
    These results imply that $s(z^\gamma) \in \Smin(z^\gamma)$ for all $\gamma\in[0,1]$ since $g_i^\gamma(x,\cdot)$ is convex for each $x\in[0,\xst_i]$, and consequently, $s(x)\in\Smin(x)$ for all $x\in[0,\xst_i]$ by the continuity of $g_i^\gamma(x,\cdot)$ on $[0,\sst_i]$ for each $x\in[0,\xst_i]$.
% 	and
% 	\begin{equation}
% 		\frac{\partial}{\partial s} g_i\left(\xpm_i, \sst_i \right) \leq 0, ~~ \text{if } \gamma > \dfrac{\sst_i}{\xst_i + \sst_i}
% 		\label{ineq:gi-s-min}
% 	\end{equation}
% 	by observing that
% 	\begin{align*}
% 		\frac{\partial}{\partial s} g_i\left(\xpm_i, \sst_i \right) & = \xpm_i Q'\left(\xpm_i + \sst_i \right) - \left(1 - \gamma\right)\left(\xst_i - \xpm_i \right) Q'\left(\xst_i + \gamma \sst_i - (1 - \gamma) \xpm_i \right)		\notag
% 		\\ & = Q\left(\xst_i + \sst_i - (1 - \gamma)\left(\xpm_i + \sst_i\right) \right) - Q\left(\xpm_i + \sst_i \right)		\notag
% %		\\ & < \frac{\xst_i}{\xst_i + \sst_i} Q\left(\xst_i + \sst_i - \frac{\xst_i}{\xst_i + \sst_i}\left(\xtld_i + \sst_i\right) \right) - Q\left(\xtld_i + \sst_i \right)		\notag
% 		\\ & \leq Q\left(\xst_i + \sst_i - \frac{1 - \gamma}{2 - \gamma}(\xst_i + \sst_i) \right) - Q\left(\frac{1 - \gamma}{2 - \gamma}\xst_i + \sst_i \right)
% 		\\ & = Q\left(\frac{\xst_i + \sst_i}{2 - \gamma} \right) - Q\left(\frac{\xst_i + \sst_i}{2 - \gamma} + \frac{(1 - \gamma)\sst_i - \gamma \xst_i}{2 - \gamma} \right)
% 		\\ & \leq 0,
% %		\label{ineq:gi-derv-s-min1}
% 	\end{align*}
%     where the second equality holds by \eqref{eq:xpm-fst-ord}, the first inequality is obtained by Assumption \ref{assump1}(ii) and Lemma \ref{lemma:xpm-lb-ub}, and the second inequality by the fact that $\gamma > \sst_i/(\xst_i + \sst_i)$.

    Furthermore, by Assumption \ref{assump1}(ii) and \eqref{useful-ineq-xy}, the following holds for all $(x, s) \in [0, \xst_i] \times [0, \sst_i]$:
    \begin{equation}
	\begin{aligned}\label{ineq:supm-gi}
		\frac{\partial^2 g_i^\gamma(x, s)}{\partial x \partial s} & = Q'(x + s) + x Q''(x + s)  + (1 - \gamma)Q'\left(\xst_i - x + \sst_i - s - \gamma(x + s)\right)
		\\ & ~~~~ + (1 - \gamma)^2(\xst_i - x) Q''\left(\xst_i - x + \sst_i - s - \gamma(x + s)\right)
		\\ & \leq Q'(x + s) + x Q''(x + s) + (1 - \gamma) Q'\left(\xst_i - x + \sst_i - s - \gamma(x + s)\right)
		\\ & ~~~~ + (1 - \gamma)(\xst_i - x) Q''\left(\xst_i - x + \sst_i - s - \gamma(x + s)\right)
		\\ & < 0.
	\end{aligned}
    \end{equation}
	Then, by Lemma 1 of~\cite{Amir2005}, $g_i^\gamma(x, s)$ is submodular in $[0, \xst_i] \times [0, \sst_i]$, i.e.,
    $$
    g_i^\gamma(x, s) + g_i^\gamma(x', s') \geq g_i^\gamma(x \wedge x', s \wedge s') + g_i^\gamma(x \vee x', s \vee s')~~\text{for all}~(x', s') \in [0, \xst_i] \times [0, \sst_i].
    $$
    Thus, Topkis's Monotonicity Theorem~\citep[see, e.g.,][Theorem 1]{Amir2005} implies that $\inf \Smin(x)$ and $\sup \Smin(x)$ are non-decreasing in $x$. This leads to the following relationship:
    \begin{equation}
		\begin{cases}
        s(x)=\inf\Smin(x) \leq \inf\Smin(z^\gamma) \leq s(z^\gamma), ~~ & \text{if } x <z^\gamma; \\s(x)=\sup\Smin(x) \geq \sup\Smin(z^\gamma) \geq s(z^\gamma) ~~ & \text{if } x > z^\gamma,
		\end{cases}
		\label{ineq:min-sx-1}
	\end{equation}
	which demonstrates that $s(x)$ is non-decreasing in $x$.

    \textbf{Step 3.} We finally prove $\xpm_i=z^\gamma$ by showing that $g_i^\gamma(\cdot, s(\cdot))$ is strictly increasing on $[0, z^\gamma)$ and is strictly decreasing on $(z^\gamma, \xst_i]$, which immediately implies that $\xpm_i = z^\gamma \leq \xst_i/2$ for any $\gamma \in [0, 1]$ by \eqref{ineq:xpm-lb-ub-smlgm} and the definition of $z^\gamma$.
	We observe that for any $x, \tilde{x}$ with $x < \tilde{x}$,
	\begin{equation}
		\begin{aligned}
        & g_i^\gamma\big(\tilde{x}, s(\tilde{x})\big) - g_i^\gamma\big(x, s(x)\big)
        \\ & = \tilde{x} Q\big(\tilde{x} + s(\tilde{x}) \big) + (\xst_i - \tilde{x}) Q\Big(\xst_i + \sst_i - (1 - \gamma)\big(\tilde{x} + s(\tilde{x})\big) \Big) 			
        \\ & ~~~~ - x Q\big(x + s(x) \big) - (\xst_i - x) Q\Big(\xst_i + \sst_i - (1 - \gamma)\big(x + s(x)\big) \Big)  	
        \\ & = (\tilde{x} - x)\bigg(Q\big(\tilde{x} + s(\tilde{x}) \big) - Q\Big(\xst_i + \sst_i - (1 - \gamma)\big(\tilde{x} + s(\tilde{x})\big) \Big) \bigg)		
        \\ & ~~~~  + x \Big(Q\big(\tilde{x} + s(\tilde{x})\big) - Q\big(x + s(x) \big) \Big) 			
        \\ & ~~~~  + (\xst_i - x) \bigg(Q\Big(\xst_i + \sst_i - (1 - \gamma)\big(\tilde{x} + s(\tilde{x})\big) \Big) - Q\Big(\xst_i + \sst_i - (1 - \gamma)\big(x + s(x)\big) \Big) \bigg).		
		\end{aligned}\label{eq:gdiff}
	\end{equation}
	
    We first fix $y, \tilde{y} \in [0, z^\gamma)$ with $y < \tilde{y}$.
	We assume that $s(y) < \sst_i$. Then, by \eqref{eq:gdiff}, there exist $\theta_1 \in (y + s(y), \tilde{y} + s(\tilde{y}) )$ and $\theta_2 \in (\xst_i + \sst_i - (1-\gamma)(\tilde{y} + s(\tilde{y})), \xst_i + \sst_i - (1-\gamma)(y + s(y)) )$ such that %, $\eta_1 \in (y + s(y), \theta_1)$ and $\eta_2 \in (\theta_2, \xst_i + \sst_i - (1-\gamma)(y + s(y)) )$, such that
	\begin{align}
		& g_i^\gamma(\tilde{y}, s(\tilde{y}) ) - g_i^\gamma(y, s(y) )
		\\ & = (\tilde{y} - y)\bigg(Q\big(\tilde{y} + s(\tilde{y}) \big) - Q\Big(\xst_i + \sst_i - (1 - \gamma)\big(\tilde{y} + s(\tilde{y})\big) \Big) \bigg)
		\\ & ~~~~ + \Delta_0 \big(y Q'(\theta_1) - (1 - \gamma)(\xst_i - y) Q'(\theta_2) \big)      		\label{mean-val-1}
		\\ & > \Delta_0 \big(y Q'(\theta_1) - (1 - \gamma)(\xst_i - y) Q'(\theta_2) \big)		
		\label{first-ineq}
		\\ & \geq \Delta_0 \big(y Q'(\theta_1) - (1 - \gamma)(\xst_i - y) Q'(\theta_2) \big)- \Delta_0 \frac{\partial}{\partial s} g_i^\gamma(y, s(y))					\label{sec-ineq}	
		\\ & = \Delta_0 y \Big(Q'(\theta_1) - Q'\big(y + s(y) \big) \Big)			\nonumber
        \\ & ~~~~ + \Delta_0 (1 - \gamma)(\xst_i - y) \bigg(Q'(u_0) - Q'(\theta_2)\bigg)					\nonumber
		% \\ & ~~~~ + \Delta_0 (1 - \gamma)(\xst_i - y) \bigg(Q'\Big(\xst_i + \sst_i - (1 - \gamma)\big(y + s(y)\big)\Big) - Q'(\theta_2)\bigg)					\nonumber
		% \\ & = \Delta_0 y\left(\theta_1 - y - s(y) \right) Q''(\eta_1)
		% \\ & ~~~~ + \Delta_0 (1 - \gamma)(\xst_i - y) \left(\xst_i + \sst_i - (1 - \gamma)(y + s(y)) - \theta_2\right) Q''(\eta_2)					\label{mean-val-2}
		\\ & \geq 0,
		\label{pos-diff}
	\end{align}
	where $\Delta_0 \coloneqq \tilde{y} + s(\tilde{y}) - (y + s(y)) > 0$ and $u_0 \coloneqq \xst_i + \sst_i - (1 - \gamma)(y + s(y))$, \eqref{mean-val-1} follows from the mean value theorem, \eqref{first-ineq} is satisfied since $Q$ is strictly decreasing by Assumption \ref{assump1}(ii) and $\tilde{y} + s(\tilde{y}) < z^\gamma+s(z^\gamma)\leq (\xst_i+\sst_i)/(2-\gamma)$ by \eqref{ineq:xpm-lb-ub-smlgm}, \eqref{sec-ineq} holds because
	% \begin{equation}
	% 	\frac{\partial}{\partial s} g_i^\gamma(y, s(y)) = y Q'(y + s(y) ) - (1 - \gamma)(\xst_i - y) Q'\left(\xst_i + \sst_i - (1 - \gamma)(y + s(y)) \right) \geq 0,
	% \end{equation}
	% which stems from the fact that
    $g_i^\gamma(y, \cdot)$ is convex and $s(y)$ is its minimizer with $s(y) < \sst_i$, and \eqref{pos-diff} stems from the convexity of $Q$ in Assumption \ref{assump1}(iv). On the other hand, if $s(y) = \sst_i$, $s(y)=s(\tilde y) = \sst_i\leq s(z^\gamma)$ since $s(\cdot)$ is non-decreasing and $y < z^\gamma$. This implies that %, we have $\sst_i = s(y) \leq s\{[(1 - \gamma)\xst_i/(2 - \gamma)] \vee \xfe_i^\gamma\}$.
     $\gamma \geq \sst_i/(\xst_i + \sst_i)$ and $(1 - \gamma)\xst_i/(2 - \gamma) \leq \xfe_i^\gamma$ by the definition of $s(\cdot)$ and \eqref{ineq:xpm-lb-ub-smlgm}, respectively. Hence, we have $y<\tilde y\leq z^\gamma = \xfe_i^\gamma$. %, which gives
	% \begin{equation*}
	% 	y < \tilde{y} < [(1 - \gamma)\xst_i/(2 - \gamma)] \vee \xfe_i^\gamma = \xfe_i^\gamma.
	% \end{equation*}
	Furthermore, by \eqref{eq:gxprime} and the strict concavity of  $g_i^\gamma(\cdot,\sst_i)$,  we have %for any $x < \xfe_i^\gamma$,
	% \begin{align*}
	% 	\frac{\partial^2 g_i}{\partial x^2}\left(x, \sst_i \right) & = 2Q'(x + \sst_i) + xQ''(x + \sst_i) + 2(1 - \gamma)Q'\left(\xst_i + \gamma \sst_i - (1 - \gamma) x \right)
	% 	\\ & ~~~~ + (1 - \gamma)^2 (\xst_i - x) Q''\left(\xst_i + \gamma \sst_i - (1 - \gamma) x \right)
	% 	\\ & \leq 2Q'(x + \sst_i) + xQ''(x + \sst_i)
	% 	\\ & ~~~~ + (1 - \gamma) \left[2Q'\left(\xst_i + \gamma \sst_i - (1 - \gamma) x \right) + (\xst_i - x) Q''\left(\xst_i + \gamma \sst_i - (1 - \gamma) x \right) \right]
	% 	\\ & < 0
	% \end{align*}
	% by \eqref{useful-ineq-xx}, we have
	\begin{equation}
		\frac{\partial g_i^\gamma}{\partial x}(x, s(x)) = \frac{\partial g_i^\gamma}{\partial x}(x, \sst_i) > \frac{\partial g_i^\gamma}{\partial x}(\xfe_i^\gamma, \sst_i) = 0~~\text{for any }x\in [y,z^\gamma).
	\end{equation}
This indicates that $g_i^\gamma(y, s(y)) = g_i^\gamma(y, \sst_i) < g_i^\gamma(\tilde{y}, \sst_i) = g_i^\gamma(\tilde{y}, s(\tilde{y}))$.

    We next fix $t, \tilde{t} \in (z^\gamma, \xst_i]$ with $t < \tilde{t}$. If $s(t) < \sst_i$, we have $s(z^\gamma) = {(\sst_i + \gamma\xst_i)}/{(2 - \gamma)}$ by \eqref{def:s-fun}
and $z^\gamma=(1-\gamma)\xst_i/(2-\gamma)$ by \eqref{ineq:xpm-lb-ub-smlgm}. Then, by letting $\Delta_1 \coloneqq \tilde{t} + s(\tilde{t}) - (t + s(t)) > 0$, there exist $\zeta_1 \in (t + s(t), \tilde{t} + s(\tilde{t}) )$ and $\zeta_2 \in (\xst_i + \sst_i - (1-\gamma)(\tilde{t} + s(\tilde{t})), \xst_i + \sst_i - (1-\gamma)(t + s(t)) )$ such that %, $\eta_1 \in (y + s(y), \theta_1)$ and $\eta_2 \in (\theta_2, \xst_i + \sst_i - (1-\gamma)(y + s(y)) )$, such that
	\begin{align}
		& g_i^\gamma(\tilde{t}, s(\tilde{t}) ) - g_i^\gamma(t, s(t) )
		\\ & = (\tilde{t} - t)\bigg(Q\big(\tilde{t} + s(\tilde{t}) \big) - Q\Big(\xst_i + \sst_i - (1 - \gamma)\big(\tilde{t} + s(\tilde{t})\big) \Big) \bigg)
		\\ & ~~~~ + \Delta_1 \big(t Q'(\zeta_1) - (1 - \gamma)(\xst_i - t) Q'(\zeta_2) \big)      		\label{mean-val-2}
		\\ & < \Delta_1 \big(tQ'(\zeta_1) - (1 - \gamma)(\xst_i - t)  Q'(\zeta_2) \big)		
		\label{first-ineq-2}
        \\ & < \Delta_1 \big((1 - \gamma)(\xst_i - t) Q'(\zeta_1) - t Q'(\zeta_2) \big)		
		\label{ineq:gfo-ub}
		\\ & \leq \Delta_1 \big((1 - \gamma)(\xst_i - t) Q'(\zeta_1) - t Q'(\zeta_2) \big) + \Delta_1 \frac{\partial}{\partial s} g_i^\gamma(t, s(t))					\label{sec-ineq-2}	
		\\ & = \Delta_1 t \Big(Q'(t + s(t)) - Q'\big(\zeta_2 \big) \Big)			\nonumber
        \\ & ~~~~ + \Delta_1 (1 - \gamma)(\xst_i - t) \bigg(Q'(\zeta_1) - Q'(u_1)\bigg)					\nonumber
		% \\ & ~~~~ + \Delta_1 (1 - \gamma)(\xst_i - t) \bigg(Q'(\zeta_1) - Q'\Big(\xst_i + \sst_i - (1 - \gamma)\big(t + s(t)\big)\Big)\bigg)					\nonumber
		% \\ & = \Delta_0 y\left(\theta_1 - y - s(y) \right) Q''(\eta_1)
		% \\ & ~~~~ + \Delta_0 (1 - \gamma)(\xst_i - y) \left(\xst_i + \sst_i - (1 - \gamma)(y + s(y)) - \theta_2\right) Q''(\eta_2)					\label{mean-val-2}
		\\ & \leq 0,
		\label{pos-diff-2}
	\end{align}
	where $u_1 \coloneqq \xst_i + \sst_i - (1 - \gamma)(t + s(t))$,  \eqref{mean-val-2} is due to \eqref{eq:gdiff} and the mean value theorem, \eqref{first-ineq-2} is satisfied since $\tilde{t} + s(\tilde{t}) > z^\gamma+s(z^\gamma) = (\xst_i+\sst_i)/(2-\gamma)$, \eqref{ineq:gfo-ub} follows by the fact that $t > z^\gamma = (1 - \gamma)\xst_i/(2 - \gamma)$ and $Q'(\cdot)<0$, \eqref{sec-ineq-2} holds because
	% \begin{equation}
	% 	\frac{\partial}{\partial s} g_i^\gamma(y, s(y)) = y Q'(y + s(y) ) - (1 - \gamma)(\xst_i - y) Q'\left(\xst_i + \sst_i - (1 - \gamma)(y + s(y)) \right) \geq 0,
	% \end{equation}
	% which stems from the fact that
    $g_i^\gamma(t, \cdot)$ is convex and $s(t)$ is its minimizer with $s(t) < \sst_i$, and \eqref{pos-diff-2} stems from Assumption \ref{assump1}(iv) and the inequalities $\zeta_2 < \xst_i +\sst_i - (1 - \gamma)(t + s(t)) < t + s(t) < \zeta_1$. On the other hand, if $s(t) = \sst_i$,  %, which gives
	% \begin{equation*}
	% 	y < \tilde{y} < [(1 - \gamma)\xst_i/(2 - \gamma)] \vee \xfe_i^\gamma = \xfe_i^\gamma.
	% \end{equation*}
	then for any $x\in(t,\xst_i]$, we have $x > z^\gamma$. By \eqref{eq:gi-dx-xgm} and \eqref{useful-ineq-xy}, if $\gamma \leq \sst_i/(\xst_i + \sst_i)$, $$\frac{\partial g_i^\gamma}{\partial x} \left(\frac{(1-\gamma)\xst_i}{2-\gamma}, \sst_i \right) \leq 0, $$ and thus,
     %for any $x < \xfe_i^\gamma$,
	% \begin{align*}
	% 	\frac{\partial^2 g_i}{\partial x^2}\left(x, \sst_i \right) & = 2Q'(x + \sst_i) + xQ''(x + \sst_i) + 2(1 - \gamma)Q'\left(\xst_i + \gamma \sst_i - (1 - \gamma) x \right)
	% 	\\ & ~~~~ + (1 - \gamma)^2 (\xst_i - x) Q''\left(\xst_i + \gamma \sst_i - (1 - \gamma) x \right)
	% 	\\ & \leq 2Q'(x + \sst_i) + xQ''(x + \sst_i)
	% 	\\ & ~~~~ + (1 - \gamma) \left[2Q'\left(\xst_i + \gamma \sst_i - (1 - \gamma) x \right) + (\xst_i - x) Q''\left(\xst_i + \gamma \sst_i - (1 - \gamma) x \right) \right]
	% 	\\ & < 0
	% \end{align*}
	% by \eqref{useful-ineq-xx}, we have
	\begin{equation*}
		\frac{\partial g_i^\gamma}{\partial x}(x, s(x)) = \frac{\partial g_i^\gamma}{\partial x}(x, \sst_i) < \frac{\partial g_i^\gamma}{\partial x}(z^\gamma, \sst_i) \leq 0
	\end{equation*}
by \eqref{eq:gxprime}, the definition of $z^\gamma$,  and the strict concavity of  $g_i^\gamma(\cdot,\sst_i)$. This suggests that $g_i^\gamma(t, s(t)) = g_i^\gamma(t, \sst_i) > g_i^\gamma(\tilde{t}, \sst_i) = g_i^\gamma(\tilde{t}, s(\tilde{t}))$.
    Consequently, the result follows.
    % {\color{blue}
    %  Fix $t, \tilde{t} \in (z^\gamma, \xst_i]$ with $t < \tilde{t}$. Thus, . Then, $s(t) = s(\tilde{t}) \geq s(z^\gamma) \geq s(\xfe_i)$ since $s(\cdot)$ is non-decreasing and $t < \tilde{t}$. Furthermore, we have $\tilde t > t > z^\gamma \geq \xfe_i^\gamma$.
    % }
    \iffalse
	Using a similar argument, it can be shown that $g_i^\gamma(\cdot, s(\cdot))$ is strictly decreasing on $(z^\gamma, \xst_i]$.
    \fi
    %
	\label{proofthm1}
	\endproof

    \proof{Proof of Theorem \ref{thm:xpm-non-mono}.}
    Fix $i\in[m]$ and define $g_i^\gamma(x, s) \coloneqq  e_i + x Q(x + s) + (\xst_i - x) Q(\xst_i + \sst_i - (1 - \gamma)(x + s))$ for each $x \in [0, \xst_i]$ and $s \in [0, \sst_i]$. Then, \eqref{prob:maximin-perm} can be rewritten as
	\begin{equation}
		\max_{x \in [0, \, \xst_i]} \, \min_{s \in \left[0, \, \sst_i \right]} \; g_i^\gamma(x, \, s).
	\end{equation}
    According to the proof of Theorem~\ref{thm:opt-strat-perm}, we observe that if $\gamma > \sst_i/(\xst_i + \sst_i)$, then $\xpm_i = \xfe_i^\gamma$ and $s(\xpm_i) \coloneqq \argmin_{s\in[0,\sst_i]}g_i^\gamma(\xpm_i,s) = \bar s_i$.
    Furthermore, since $(\partial/\partial x)g_i^\gamma(\cdot,\sst_i)$ is strictly decreasing on $[0, \xst_i]$ by Assumption~\ref{assump1}(iv) and \eqref{useful-ineq-xx},
    % The proof of the first statement is straightforward and omitted here. For the second statement, since by \eqref{useful-ineq-xx},
    % \begin{align}
    %     \frac{\partial^2 g_i}{\partial x^2}(x, \sst_i) & = 2Q'(x + \sst_i) + xQ''(x + \sst_i)
    %     \\ & ~~~~ + 2(1 - \gamma)Q'\left(\xst_i + \gamma\sst_i - (1 - \gamma)x\right) + (1 - \gamma)^2Q''\left(\xst_i + \gamma\sst_i - (1 - \gamma)x\right)
    %     \\ & \leq 2Q'(x + \sst_i) + xQ''(x + \sst_i)
    %     \\ & ~~~~ + 2(1 - \gamma)\left[Q'\left(\xst_i + \gamma\sst_i - (1 - \gamma)x\right) + Q''\left(\xst_i + \gamma\sst_i - (1 - \gamma)x\right)\right]
    %     \\ & \leq 0,
    % \end{align}
    % then by \eqref{eq:xpm-fst-ord},
    it suffices to show that $\gamma > \sst_i/(\xst_i + \sst_i)$ and $\xst_i\leq \sst_i$ imply
    \begin{equation}
        \frac{\partial^2 g_i}{\partial x \partial \gamma}(\xfe_i^\gamma, \sst_i) \geq 0.
    \end{equation}

    Assume that $\gamma>\sst_i/(\xst_i + \sst_i)$ and $\xst_i\leq \sst_i$. Then, we have
    $$
    \frac{(1 - \gamma^2)\xst_i}{2 - \gamma} + \gamma \sst_i \geq \xst_i
    $$
    since it holds when $\gamma = \sst_i/(\xst_i + \sst_i)$ and the left-hand side increases in $\gamma$. Thus, by \eqref{ineq:xpm-lb-ub-smlgm}, we obtain $(1 + \gamma)\xfe_i^\gamma + \gamma \sst_i \geq \xst_i$. This suggests that $\gamma\left(\xfe_i^\gamma + \sst_i \right)^2 \geq \left(\xfe_i^\gamma + \sst_i \right)\left(\xst_i - \xfe_i^\gamma\right)\geq\left(\xst_i - \xfe_i^\gamma\right)^2$, and therefore, we observe that
%
    % \begin{align}
    %     (\xst_i + \sst_i)\gamma^2 - (2\sst_i + \xst_i)\gamma + \xst_i & = (\xst_i + \sst_i) \left[\gamma^2 - \left(1 + \frac{\sst_i}{\xst_i + \sst_i} \right)\gamma + \frac{\xst_i}{\xst_i + \sst_i} \right]
    %     \\ & < (\xst_i + \sst_i) \left[\gamma^2 - \left(1 + \frac{\xst_i}{\xst_i + \sst_i} \right)\gamma + \frac{\xst_i}{\xst_i + \sst_i} \right]
    %     \\ & = (\xst_i + \sst_i) \left[\gamma(\gamma - 1) - \frac{\xst_i}{\xst_i + \sst_i} (\gamma - 1) \right]
    %     \\ & = (\xst_i + \sst_i) (\gamma - 1) \left(\gamma - \frac{\xst_i}{\xst_i + \sst_i} \right)
    %     \\ & \leq 0,
    % \end{align}
    % implying that
    % \begin{equation}
    %     (1 + \gamma)\xfe_i^\gamma + \gamma \sst_i \geq \frac{(1 - \gamma^2)\xst_i}{2 - \gamma} + \gamma \sst_i \geq \xst_i ~~
    %     \Longrightarrow ~~ \gamma\left(\xfe_i^\gamma + \sst_i \right) \geq \xst_i - \xfe_i^\gamma
    %     \label{ineq:xfe-gm-lb}
    % \end{equation}
    % by Lemma \ref{lemma:xpm-lb-ub}. Therefore,
    \begin{equation}
    \begin{aligned}
        & \left(2\xfe_i^\gamma + \sst_i - \xst_i \right) \left(\xst_i + \gamma \sst_i - (1 - \gamma)\xfe_i^\gamma \right)
        - (1 - \gamma) \left(\xfe_i^\gamma + \sst_i\right) \left(\xst_i - \xfe_i^\gamma \right) \\&= \gamma \left(\xfe_i^\gamma + \sst_i \right)^2 - \left(\xst_i - \xfe_i^\gamma \right)^2\\&\geq0,
        % \\ & \geq (1 - \gamma) \left(\xfe_i^\gamma + \sst_i\right) \left(\xst_i - \xfe_i^\gamma \right) + \gamma^2 \left(\xfe_i^\gamma + \sst_i \right)^2 - \left(\xst_i - \xfe_i^\gamma \right)
        % \\ & = (1 - \gamma) \left(\xfe_i^\gamma + \sst_i\right) \left(\xst_i - \xfe_i^\gamma \right) + \left[\gamma \left(\xfe_i^\gamma + \sst_i \right) - \left(\xst_i - \xfe_i^\gamma \right) \right]^2
        % \\ & \geq (1 - \gamma) \left(\xfe_i^\gamma + \sst_i\right) \left(\xst_i - \xfe_i^\gamma \right),
        \label{ineq:xfe-abgm}
    \end{aligned}
    \end{equation}
    which also gives $2\xfe_i^\gamma + \sst_i - \xst_i \geq 0$.
    % \begin{equation}
    %     2\xfe_i^\gamma + \sst_i - \xst_i \geq \frac{(1 - \gamma) \left(\xfe_i^\gamma + \sst_i\right) \left(\xst_i - \xfe_i^\gamma \right)}{\xst_i + \gamma \sst_i - (1 - \gamma)\xfe_i^\gamma} = \frac{(1 - \gamma) \left(\xfe_i^\gamma + \sst_i\right) \left(\xst_i - \xfe_i^\gamma \right)}{\xst_i - \xfe_i^\gamma + \gamma \left(\xfe_i^\gamma + \sst_i \right)} \geq 0.
    %     \label{ineq:2xfe+sbar-xbar}
    % \end{equation}
    Consequently, we obtain the following relationship:
    \begin{equation}
    \begin{aligned}
        &\frac{\partial^2 g_i}{\partial x \partial \gamma}(\xfe_i^\gamma, \sst_i)\\ & = (\xst_i - \sst_i - 2\xfe_i^\gamma) Q'\big(\xst_i + \gamma\sst_i - (1 - \gamma )\xfe_i^\gamma\big)
         - (1 - \gamma)(\xst_i - \xfe_i^\gamma)(\xfe_i^\gamma + \sst_i) Q''\big(\xst_i + \gamma\sst_i - (1 - \gamma )\xfe_i^\gamma\big)
        \\ & \geq -(2\xfe_i^\gamma + \sst_i - \xst_i)\Big(Q'\big(\xst_i + \gamma\sst_i - (1 - \gamma )\xfe_i^\gamma\big)
        +  \big(\xst_i + \gamma\sst_i - (1 - \gamma)\xfe_i^\gamma\big) Q''\big(\xst_i + \gamma\sst_i - (1 - \gamma )\xfe_i^\gamma\big)\Big)
        % \\ & = -(2\xfe_i^\gamma + \sst_i - \xst_i) \big[\big(\xst_i + \gamma\sst_i - (1 - \gamma)\xfe_i^\gamma\big) Q'\big(\xst_i + \gamma\sst_i - (1 - \gamma ) \xfe_i^\gamma \big) \big]'
        \\ & \geq 0,
    \end{aligned}
    \end{equation}
    where the first and second inequalities hold by \eqref{ineq:xfe-abgm} and Assumption~\ref{assump1}(v), respectively.
	\endproof

    % {\color{red}
    \proof{Proof of Proposition \ref{thm:rbst-dec}.}
    Fix $i\in[m]$ and $\xst_i>0$. By Theorem \ref{thm:opt-strat-perm}, it suffices to prove that $\xpm_i$ decreases in $\gamma$ if $\gamma > \sst_i/(\xst_i + \sst_i)$ and $\sst_i$ is sufficiently small. Define $g_i^\gamma(x, s) \coloneqq  e_i + x Q(x + s) + (\xst_i - x) Q(\xst_i + \sst_i - (1 - \gamma)(x + s))$ for each $x \in [0, \xst_i]$ and $s \in [0, \sst_i]$.  Fix $\gamma > \sst_i/(\xst_i + \sst_i)$.
    We observe that if $\gamma>\sqrt{3}-1$, then for some $\theta \in [(\xst_i + \sst_i)/2, (1+\gamma)(\xst_i + \sst_i)/2]$,
    \begin{equation}
        \begin{aligned}
            & \frac{\partial g_i}{\partial x}\bigg(\frac{\xst_i - \sst_i}{2}, \sst_i\bigg)
            \\ & = Q\bigg(\frac{\xst_i + \sst_i}{2}\bigg) + \frac{\xst_i - \sst_i}{2} Q'\bigg(\frac{\xst_i + \sst_i}{2}\bigg) - Q\bigg(\frac{(1+\gamma)(\xst_i + \sst_i)}{2}\bigg) - \frac{(1 - \gamma)(\xst_i + \sst_i)}{2}Q'\bigg(\frac{(1+\gamma)(\xst_i + \sst_i)}{2} \bigg)
            \\ & = \frac{\xst_i - \sst_i}{2} Q'\bigg(\frac{\xst_i + \sst_i}{2}\bigg) - \frac{\gamma(\xst_i + \sst_i)}{2} Q'\bigg(\frac{\xst_i + \sst_i}{2}\bigg) - \frac{\gamma^2(\xst_i + \sst_i)^2}{8}Q''(\theta) - \frac{(1 - \gamma)(\xst_i + \sst_i)}{2}Q'\left(\frac{(1+\gamma)(\xst_i + \sst_i)}{2} \right)
            \\ & \leq \frac{\xst_i - \sst_i}{2} Q'\bigg(\frac{\xst_i + \sst_i}{2}\bigg) - \frac{\gamma(\xst_i + \sst_i)}{2} Q'\bigg(\frac{\xst_i + \sst_i}{2}\bigg) - \frac{\gamma^2(\xst_i + \sst_i)^2}{8}Q''(\theta) - \frac{(1 - \gamma)(\xst_i + \sst_i)}{2}Q'\bigg(\frac{\xst_i + \sst_i}{2} \bigg)
            \\ & = - \sst_i Q'\bigg(\frac{\xst_i + \sst_i}{2} \bigg) - \frac{\gamma^2(\xst_i + \sst_i)^2}{8}Q''(\theta)
            \\ & \leq -Q'\bigg(\frac{\xst_i + \sst_i}{2} \bigg) \bigg(\sst_i + \frac{(\sqrt{3}-1)^2\xst_i^2 \delta}{8 Q'(0)} \bigg)
            \\ & =-Q'\bigg(\frac{\xst_i + \sst_i}{2} \bigg) \bigg(\sst_i + \frac{(2-\sqrt{3})\xst_i^2 \delta}{4 Q'(0)} \bigg),
        \end{aligned}
        \label{ineq:gi-p1-ub}
    \end{equation}
    where $\delta = \min_{x \in [0, \xst_i + \sst_i]}Q''(x) > 0$, the second equality follows from the first-order Taylor expansion, and the two inequalities hold since $ Q'(0) \leq Q'((\xst_i + \sst_i)/2) \leq Q'((1+\gamma)(\xst_i + \sst_i)/2) \leq 0$ by Assumption \ref{assump1}.

    As in the proof of Theorem \ref{thm:opt-strat-perm}, define $\xfe_i^\gamma$ as the solution of \eqref{eq:xpm-fst-ord} with respect to $x$ given $\gamma > \sst_i/(\xst_i + \sst_i)$. In the remainder of the proof, we assume that $\sst_i <-{(2-\sqrt{3})\xst_i^2 \delta}/({4 Q'(0)})$ and claim that
    \begin{equation}
        \frac{\partial^2 g_i}{\partial x \partial \gamma}(\xfe_i^\gamma, \sst_i) = (\xst_i - \sst_i - 2\xfe_i^\gamma) Q'\big(\xst_i + \gamma\sst_i - (1 - \gamma )\xfe_i^\gamma\big) - (1 - \gamma)(\xst_i - \xfe_i^\gamma)(\xfe_i^\gamma + \sst_i) Q''\big(\xst_i + \gamma\sst_i - (1 - \gamma )\xfe_i^\gamma\big)\leq 0.
    \end{equation}
    This is trivial if  $\xfe_i^\gamma < (\xst_i - \sst_i)/2$. Thus, suppose that $\xfe_i^\gamma \geq (\xst_i - \sst_i)/2$. Then, since $(\partial/\partial x)g_i^\gamma(\cdot,\sst_i)$ is decreasing on $[0, \xst_i]$ and $(\partial/\partial x)g_i^\gamma(\xfe_i^\gamma,\sst_i) = 0$ by the proof of Theorem~\ref{thm:opt-strat-perm},
    % by analogy with \eqref{ineq:gi-p1-ub2} and \eqref{ineq:xfe-xst-ub},
    \eqref{ineq:xpm-lb-ub-smlgm} and \eqref{ineq:gi-p1-ub} suggests $\gamma\leq \sqrt{3}-1$ and
    \begin{equation}\label{eq:xleqgeq}
        \frac{\xst_i-\sst_i}{2} \leq \xfe_i^\gamma \leq  \frac{\xst_i}{2} - \frac{(1 - \gamma)\sst_i}{2(2 - \gamma)},
    \end{equation}
    % respectively.
    We therefore obtain the following result:
    \begin{equation}
        \begin{aligned}
            \frac{\partial^2 g_i}{\partial x \partial \gamma}(\xfe_i^\gamma, \sst_i)
            & \leq - \frac{\sst_i}{2 - \gamma} Q'\big(\xst_i + \gamma\sst_i - (1 - \gamma )\xfe_i^\gamma\big) - (1 - \gamma)\bigg(\frac{\xst_i}{2} + \frac{(1 - \gamma)\sst_i}{2(2-\gamma)} \bigg) \frac{(\xst_i + \sst_i)}{2} Q''\big(\xst_i + \gamma\sst_i - (1 - \gamma )\xfe_i^\gamma\big)
            \\ & \leq - \sst_i Q'\big(\xst_i + \gamma\sst_i - (1 - \gamma )\xfe_i^\gamma\big) - (2-\sqrt{3})\frac{\xst_i^2}{4} Q''\big(\xst_i + \gamma\sst_i - (1 - \gamma )\xfe_i^\gamma\big)
            \\ & = -Q'\big(\xst_i + \gamma\sst_i - (1 - \gamma )\xfe_i^\gamma\big)  \bigg( \sst_i + (2-\sqrt{3})\frac{\xst_i^2}{4} \frac{ Q''\big(\xst_i + \gamma\sst_i - (1 - \gamma )\xfe_i^\gamma\big)}{Q'\big(\xst_i + \gamma\sst_i - (1 - \gamma )\xfe_i^\gamma\big)} \bigg)
            \\ & \leq -Q'\big(\xst_i + \gamma\sst_i - (1 - \gamma )\xfe_i^\gamma\big) \bigg(\sst_i + \frac{(2-\sqrt{3}) \xst_i^2 \delta}{4Q'(0)} \bigg)
            \\ & < 0,
        \end{aligned}
        \label{ineq:gip2-ub}
    \end{equation}
    where the first two inequalities comes from \eqref{eq:xleqgeq} and the third inequality holds by Assumption \ref{assump1}. Consequently, the above claim is demonstrated. This leads to the desired result because the proof of Theorem~\ref{thm:opt-strat-perm} shows that $(\partial/\partial x)g_i^\gamma(\cdot,\sst_i)$ is strictly decreasing on $[0, \xst_i]$ and that $\xpm_i = \xfe_i^\gamma$ and $s(\xpm_i) \coloneqq \argmin_{s\in[0,\sst_i]}g_i^\gamma(\xpm_i,s) = \bar s_i$.
    \endproof

    \section{Proofs of Theoretical Results in Section \ref{sec:intb_liab}}  \label{sec:proof_intb}
    % {\color{red}
    Throughout this section, we denote the 1-norm of a vector $\bu\in\bbR^d$ by $\|\bu\|_1 \coloneqq \sum_{k = 1}^{d} |u_k|$, and we define $\pbar\coloneqq Q(\bone^\top\bxbar)$.

    \proof{Proof of Lemma \ref{lemma:yf-ub-intb}.}
    To facilitate the proof, we use the following lemma, whose proof is provided in Appendix~\ref{app:auxiliary}.
    \begin{lemma}\label{lemma:psA-pbar-comp}
        $\pbar \leq \psA$.
    \end{lemma}
    Since $\bwNf(\cdot)$ is increasing and $p_1 \coloneqq Q(\bone^\top \bx) \geq Q(\bone^\top \bxbar) = \pbar$, Lemma \ref{lemma:psA-pbar-comp} implies that $$\blbar = \bwNf(\bz \pbar) \leq \bwNf(\bx p_1 + (\bz - \bx)p_2^\brLM) = \brpA.$$ Consequently, the following  holds for all $j$:
    \begin{equation}
        \begin{aligned}
            \yf_j^\brLM(\bx) & = \frac{\left(L_j - e_j - x_{j}p_1 - \brLM_j^\top \brpA \right)^+}{p_2^\brLM} \wedge \left(z_j - x_{j}\right)
            \\ & \leq \frac{\left(L_j - e_j - x_{j}p_1 - \brLM_j^\top \blbar \right)^+}{\pbar} \wedge \left(z_j - x_{j}\right)
            \\ &\leq \frac{\left(L_j - e_j - x_{j}\pbar - \brLM_j^\top \blbar \right)^+}{\pbar} \wedge \left(z_j - x_{j}\right)
            \\ & =\left( \frac{L_j - e_j - \brLM_j^\top \blbar}{\pbar}-x_j\right) \wedge \left(z_j - x_{j}\right)
            \\ & = \xbar_j-x_j,
        \end{aligned}
    \end{equation}
    where the second last equality holds since $x_j\leq\xbar_j\leq(L_j - e_j - \brLM_j^\top \blbar)/\pbar$ for all $j$. This completes the proof. 			\endproof

    \proof{Proof of Proposition \ref{prop:eff-intb}.}
        To explicitly show the dependence of $\bwNf^\brLM$ on the first-period liability repayment $\bx$, we write $\bwNf^\brLM=\bwNf^\brLM_\bx$ in this proof. We first show that
        \begin{equation}
            \bwNf^\brLM_\bx \geq \bwNf(\bcNf(\bx))~~\text{for all}~\bx \leq \bxbar.
            \label{ineq:bwNf-rl-org-comp}
        \end{equation}
         Define ${\bf \yf}^\brLM(\bx) \coloneqq (\yf_1^\brLM(\bx), \dots, \yf_m^\brLM(\bx))^\top$. Let
		$
			\bphi_0 = \bL \wedge \left(\bgao + \bx Q(\bone^\top \bx) + (\bxbar - \bx) Q\big(\bone^\top (\bxbar - \bx) \big) + \brLM^\top \bwNf^\brLM_\bx \right)
		$
		and
		$
			\bphi_{k+1} = \bL \wedge \left(\bgao + \bx Q(\bone^\top \bx) + (\bxbar - \bx) Q\big(\bone^\top (\bxbar - \bx) \big) + \brLM^\top \bphi_{k}\right) \text{ for } k = 0, 1, 2, \dots.
		$
		Then, $\bwNf^\brLM_\bx \geq \bphi_0$ since for any $j \in [m]$,%, if $\wNf_j^\brLM = L_j$, then $\wNf_j^\brLM \geq \phi_{0, j}$ by the definition of $\bphi_0$. Otherwise, if $\wNf_j^\brLM < L_j$, $\yf_j^\brLM(\bx) = z_j - x_j$ by the definition of $\yf_j^\brLM(\bx)$, and thus
        $$
		\begin{aligned}
			\wNf_{\bx, j}^\brLM & = L_j \wedge \left(e_j + x_j Q(\bone^\top \bx) + (z_j - x_j) Q\big(\bone^\top  {\bf \yf}^\brLM(\bx)\big) + \brLM_j^\top \bwNf^\brLM_\bx \right)
			\\ & \geq L_j \wedge \left(e_j + x_j Q(\bone^\top \bx) + (\xbar_j - x_j) Q\big(\bone^\top  (\bxbar - \bx)\big) + \brLM_j^\top \bwNf^\brLM_\bx \right)
			\\ & = \phi_{0, j},
		\end{aligned}
        $$
        where the inequality follows from Lemma \ref{lemma:yf-ub-intb} and the fact that $\bxbar \leq \bz$. Therefore, $\bwNf^\brLM_\bx \geq \bphi_0$. Furthermore, using a similar argument, it can be shown that if $\bwNf^\brLM_\bx \geq \bphi_k$ for some $k \geq 0$, then $\bwNf^\brLM_\bx \geq \bphi_{k+1}$. %either $\wNf_j^\brLM = L_j \geq \phi_{k+1, j}$ or
		% $$
  %       \begin{aligned}
		% 	\wNf_j^\brLM & = L_j \wedge \left[e_j + x_j Q(\bone^\top \bx) + (z_j - x_j) Q\left(\bone^\top  {\bf \yf}^\brLM(\bx)\right) + \brLM_j^\top \bwNf^\brLM \right]
		% 	\\ & \geq L_j \wedge \left[e_j + x_j Q(\bone^\top \bx) + (\xbar_j - x_j) Q\left(\bone^\top  (\bxbar - \bx)\right) + \brLM_j^\top \bphi_k \right]
		% 	\\ & = \phi_{k+1, j}.
		% \end{aligned}
  %       $$
		Hence, $\bwNf_\bx^\brLM \geq \bphi_k$ for all $k$ by induction, which suggests \eqref{ineq:bwNf-rl-org-comp} since $\bphi_k \rightarrow \bwNf(\bcNf(\bx))$ as $k$ increases.

        Next, we prove that $\lrLM_i(x_i, \bxbar_{-i}) \leq L_i$ for any $x_i \leq \xbar_i$. Define a vector $\bu \in \bbR^m$ with $u_i = x_i$ and $\bu_{-i} = \bxbar_{-i} $. Note that $\yf_j^\brLM(\bu) = 0$ for all $j\neq i$ by Lemma~\ref{lemma:yf-ub-intb}. If $\yf_i^\brLM(\bu) > 0$, then $L_i - e_i - x_i Q(\bone^\top \bu) - \brLM^\top \bwNf^\brLM_\bu > 0$ by \eqref{eq:psij-rLM}, %by the definition of $\yf_i^\brLM(\cdot)$, we have $(L_i - e_i - x_i Q(\bone^\top \bu) - \brLM^\top \bwNf(\bc^\brLM(\bu)))^+ > 0$, implying that $L_i - e_i - x_i Q(\bone^\top \bu) - \brLM^\top \bwNf(\bc^\brLM(\bu)) > 0$,
        and hence,
        \begin{equation}
            \begin{aligned}
                \lrLM_i(x_i, \bxbar_{-i}) & = e_i + x_i Q(\bone^\top \bu) + \yf_i^\brLM(\bu)Q(\yf_i^\brLM(\bu)) + \brLM_i^\top \bwNf^\brLM_\bu
                \\ & = \Big(L_i \vee \big(e_i + x_i Q(\bone^\top \bu) + \brLM^\top \bwNf^\brLM_\bu \big) \Big) \wedge \Big( e_i + x_i Q(\bone^\top \bu) + (z_i - x_i) Q(\yf_i^\brLM(\bu)) + \brLM_i^\top \bwNf^\brLM_\bu \Big)
                \\ & = L_i \wedge \Big( e_i + x_i Q(\bone^\top \bu) + (z_i - x_i) Q(\yf_i^\brLM(\bu)) + \brLM_i^\top \bwNf^\brLM_\bu \Big)
                \\ & \leq L_i,
            \end{aligned}
        \end{equation}
        where the second equality holds by \eqref{eq:psij-rLM}. If $\yf_i^\brLM(\bu) = 0$, then $\bwNf^\brLM_\bu = \bwNf(\bu Q(\bone^\top \bu))$ by \eqref{eq:p2-l-rLM} and \eqref{eq:psij-rLM}, and therefore,
        $$
        \begin{aligned}
            \lrLM_i(x_i, \bxbar_{-i}) & = e_i + x_i Q(\bone^\top \bu) + \brLM^\top \bwNf(\bu Q(\bone^\top \bu))
            \\ & \leq e_i + x_i Q(\bone^\top \bu) + \brLM_i^\top \bwNf(\bu Q(\bone^\top \bu)) - \left(e_i + \xbar_i \pbar + \brLM_i^\top \blbar \right) + L_i
            \\ & = x_i Q(\bone^\top \bu) - \xbar_i Q(\bone^\top \bxbar) + \brLM_i^\top \big(\bwNf(\bu Q(\bone^\top \bu)) - \bwNf(\bxbar Q(\bone^\top \bxbar)) \big) + L_i
            \\ & \leq x_i Q(\bone^\top \bu) - \xbar_i Q(\bone^\top \bxbar) + \brLM_i^\top \Big(\bwNf\big((\bu Q(\bone^\top \bu)) \vee (\bxbar Q(\bone^\top \bxbar)) \big) - \bwNf(\bxbar Q(\bone^\top \bxbar)) \Big) + L_i
            \\ & \leq x_i Q(\bone^\top \bu) - \xbar_i Q(\bone^\top \bxbar) + \big\|(\bu Q(\bone^\top \bu)) \vee (\bxbar Q(\bone^\top \bxbar)) -  \bxbar Q(\bone^\top \bxbar) \big\|_1
            \\ & = x_i Q(\bone^\top \bu) - \xbar_i Q(\bone^\top \bxbar) + (\bone^\top \bxbar_{-i})\big(Q(\bone^\top \bu) - Q(\bone^\top \bxbar)\big) + L_i
            \\ & = (\bone^\top \bu) Q(\bone^\top \bu) - (\bone^\top \bxbar) Q(\bone^\top \bxbar) + L_i
            \\ & \leq L_i,
        \end{aligned}
        $$
        where the first inequality comes from the definition of $\xbar_i$ and the assumption that $\bxbar>\bzero$, the second inequality holds because $\bwNf(\cdot)$ is increasing, the third inequality follows since $\bwNf(\cdot)$ is non-expansive and $\brLM_{i} < \bone$, the third equality is satisfied as $\bu_{-i} Q(\bone^\top \bu) = \bxbar_{-i} Q(x_i + \bone^\top \bxbar_{-i}) \geq \bxbar_{-i} Q(\bone^\top \bxbar)$ by Assumption~\ref{assump1}(ii) and $u_i Q(\bone^\top \bu) = x_i Q(x_i + \bone^\top \bxbar_{-i}) \leq \xbar_i Q(\bone^\top \bxbar)$ by \eqref{useful-ineq-x}, and the last inequality stems from Assumption \ref{assump1}(iii).

		Finally, we demonstrate the two statements. Define $\bv \in \bbR^m$ where $v_i = x_i$ and $\bv_{-i} = \argmin_{\bx_{-i} \in [\bzero, \bxst_{-i}]} \lrLM_i(x_i, \bx_{-i})$. Note that $\lrLM_i(x_i, \bv_{-i}) \leq \lrLM_i(x_i, \bxbar_{-i}) \leq L_i$. %It suffices to show that either $\ell_i^\brLM(x_i, \bv_{-i}^\brLM) = L_i$ or $\cNf_i(\bv) \leq c_i^\brLM(\bv)$ if $\ell_i^\brLM(x_i, \bv_{-i}^\brLM) < L_i$, which implies that $\ltN_i(x_i, \bv_{-i}^\brLM) \leq \ell_i^\brLM(x_i, \bv_{-i}^\brLM)$ by \eqref{ineq:bwNf-rl-org-comp}.
        Assume that $\lrLM_i(x_i, \bv_{-i}) < L_i$. Then, by \eqref{eq:lA}, we get
		\begin{equation}
			\yf_i^\brLM(\bv) = \frac{\lrLM_i(x_i, \bv_{-i}^\brLM) - e_i -x_i Q\big(\bone^\top \bv \big)- \brLM_i^\top \bwNf^\brLM_\bv}{Q\left(\sum_{j=1}^m \yf_j^\brLM(\bv) \right)} < \frac{\left(L_i - e_i-x_i Q\big(\bone^\top \bv \big) - \brLM_i^\top \bwNf^\brLM_\bv \right)^+}{Q\left(\sum_{j=1}^m \yf_j^\brLM(\bv) \right)},
		\end{equation}
		which implies that $\yf_i^\brLM(\bv) = z_i - x_i $ by \eqref{eq:psij-rLM}. Therefore, we obtain
        \begin{equation}\label{eq:lAldag}
		\begin{aligned}
			\lrLM_i(x_i, \bv_{-i}) %& = e_i + x_i Q\big(\bone^\top \bv \big) + \yf_i^\brLM(\bv) Q\Big(\sum_{j = 1}^m \yf_j^\brLM(\bv) \Big) + \brLM_i^\top \bwNf\big(\bc^\brLM(\bv) \big)
			& = e_i + x_i Q\big(\bone^\top \bv \big) + (z_i - x_i) Q\Big(\sum_{j = 1}^m \yf_j^\brLM(\bv) \Big) + \brLM_i^\top \bwNf^\brLM_\bv
			\\ & \geq e_i + x_i Q\big(\bone^\top \bv \big) + (\xbar_i - x_i) Q\Big(\sum_{j = 1}^m (\xbar_j - v_j) \Big) + \brLM_i^\top \bwNf\big(\bcNf(\bv) \big)
			\\ & = \ltN_i(x_i, \bv_{-i}),
		\end{aligned}
        \end{equation}
		where the inequality holds by \eqref{ineq:bwNf-rl-org-comp}, Lemma \ref{lemma:yf-ub-intb}, and the fact that $\xbar_i \leq z_i$.
        Consequently, $\min_{ \bx_{-i} \in [\bzero,  \bxst_{-i} ]}\ltN_i(x_i, \bx_{-i})<L_i$, which demonstrates the first statement by contrapositive.
		
		The second statement is straightforward if $\lrLM_i(x_i, \bv_{-i}^\brLM) = L_i $. Otherwise, it holds by \eqref{eq:lAldag}.
    \endproof
    % }

    \proof{Proof of Theorem \ref{thm-maximin-network}.}
        Fix $i \in [m]$. If $\xbar_i = 0$, the statement is straightforward. We assume that $\xbar_i > 0$. Note that $\ltN_i(x,  \bx_{-i}) = \ltN_i(\xbar_i - x,  \bxbar_{-i} - \bx_{-i})$ for all $x \in [0,  \xbar_i]$. Thus, it suffices to show that for all $x \in [0,  \xbar_i/2)$,
        \begin{equation}
            \min_{\bx_{-i} \in [\bzero,  \bxbar_{-i} ]} \ltN_i\left(\frac{\xbar_i}{2},  \bx_{-i} \right) > \min_{\bx_{-i} \in [\bzero,  \bxbar_{-i} ]} \ltN_i\left(x,  \bx_{-i} \right)
            \label{ltN-min-comp}
        \end{equation}
        We prove \eqref{ltN-min-comp} in the following four steps.

        \textbf{Step 1.} Let $\clbar = \sum_{j = 1}^{m}\xbar_j$ and $\sbar_{i, j} = \clbar - \xbar_i - \xbar_j$ for $j \in [m]$. For all $y \in [0,  \xbar_i/2]$, $j \neq i$, and $ \tcl \in \Zset(y) \coloneqq [y,  \clbar - \xbar_i + y]$, define
        \begin{equation}
            \begin{aligned}
                \hlf_j\left(y, \tcl \right) & \coloneqq \left(\tcl - y - \sbar_{i, j} \right)^+ \left(Q(\tcl) - Q(\clbar - \tcl) \right)
                 + \xbar_j Q(\clbar - \tcl),
            \end{aligned}
            \label{hlf-def}
        \end{equation}
        \begin{equation}
            \begin{aligned}
                \hrt_j\left(y, \tcl \right) & \coloneqq \left(\xbar_j \wedge \left(\tcl - y \right) \right) \left(Q(\tcl) - Q(\clbar - \tcl) \right)
                 + \xbar_j Q(\clbar - \tcl),
            \end{aligned}
            \label{hrt-def}
        \end{equation}
        and
        $$\Bset(y, \tcl) \coloneqq \left\{ \bc \in \bbR^{m} \Bigg|
        \renewcommand*{\arraystretch}{1.4}\begin{array}{l}
              \bone^\top \bc = \tcl \left(Q(\tcl) - Q(\clbar - \tcl)\right) + \clbar Q(\clbar - \tcl),
              c_i = y\left(Q(\tcl) - Q(\clbar - \tcl)\right) + \xbar_i Q(\clbar - \tcl),   			\nonumber
            \\  c_j \in \left[\hlf_j(y, \tcl) \wedge \hrt_j(y, \tcl),   \hlf_j(y, \tcl) \vee \hrt_j(y, \tcl) \right], ~\forall j \neq i
        \end{array}
        \right\}.$$
    %		\begin{equation}
    %			
    %		\end{equation}
        For each $y \in [0,  \xbar_i/2]$, $j \neq i$, and $\tcl \in \Zset(y)$, $(\tcl - y - \sbar_{i, j})^+ \leq \xbar_j \wedge (\tcl - y)$, and thus, we have
        \begin{equation}
            \begin{cases}
                \hlf_j(y, \tcl) \leq \hrt_j(y, \tcl),  & \text{if } \tcl \leq \clbar/2; \\
                \hlf_j(y, \tcl) > \hrt_j(y, \tcl), & \text{otherwise}.
            \end{cases}
            \label{hlf-hrt-comp}
        \end{equation}
        Observe that
        \begin{equation}
            \sum_{j \neq i} \left(\tcl - y - \sbar_{i, j} \right)^+ \leq \tcl - y \leq \sum_{j \neq i} \left(\xbar_j \wedge \left(\tcl - y \right) \right),
            \label{tcl-y-ineq}
        \end{equation}
        where the first inequality holds because if $\tcl - y > \sbar_{i, k^*}$, where $k^* \coloneqq \argmin_{k \neq i}\sbar_{i, k}$,
        $$
            \sum_{j \neq i} (\tcl - y - \sbar_{i, j})^+ \leq \tcl - y - \sbar_{i, k^*} + \sum_{j \neq i, k^*}\xbar_j = \tcl - y,
        $$
        and $\sum_{j \neq i} (\tcl - y - \sbar_{i, j})^+ = 0 \leq \tcl - y$ otherwise. Accordingly, if $\tcl > \clbar/2$, we obtain
        $$
            \sum_{j \neq i} \hlf_j(y, \tcl) \leq  (\tcl - y) \left(Q(\tcl) - Q(\clbar - \tcl) \right)
                + (\clbar - \xbar_i) Q(\clbar - \tcl)
             \leq \sum_{j \neq i} \hrt_j(y, \tcl),
        $$
        and otherwise
        $$
            \sum_{j \neq i} \hlf_j(y, \tcl) \geq   (\tcl - y) \left(Q(\tcl) - Q(\clbar - \tcl) \right)
             + (\clbar - \xbar_i) Q(\clbar - \tcl)
             \geq \sum_{j \neq i} \hrt_j(y, \tcl).
        $$
        This implies that $\Bset(y, \tcl)$ is non-empty for any $y \in [0,  \xbar_i/2]$ and $\tcl \in \Zset(y)$.

        \textbf{Step 2.} For $y \in [0,  \xbar_i/2]$, we claim that
        \begin{equation}
                \gtld_i(y)  \coloneqq \min_{\tcl \in \Zset(y)}\min_{\bc \in \Bset(y, \tcl)} \left\{e_i + c_i + \brLM_i^\top \bwNf(\bc) \right\}
                 = \min_{\bx_{-i} \in [\bzero,  \bxbar_{-i} ]} \ltN_i(y, \bx_{-i}).
            \label{gi-ltNi-eq}
        \end{equation}
        We first prove that $\gtld_i(y) \leq \min_{\bx_{-i} \in [\bzero,  \bxbar_{-i} ]} \ltN_i(y, \bx_{-i})$. For any $\bx^0$ with $x_i^0 = y$ and $\bx_{-i}^0 \in [\bzero,  \bxbar_{-i} ]$, let $\tcl^0 = \bone^\top \bx^0$ and $\bc^0 = \bcNf(\bx^0)$. Then, by \eqref{op-cash-inter}, for all $j \in [m]$,
        \begin{equation}
            c_j^0 = x_j^0 \left(Q(\tcl^0) - Q(\clbar - \tcl^0) \right) + \xbar_j Q(\clbar - \tcl^0).
            \label{cj0-eq}
        \end{equation}
        Obviously, $\tcl^0 \in \Zset(y)$. Moreover,  $(\tcl^0 - \sbar_{i, j} - y)^+ \leq x_j^0 \leq \tcl^0 - y$, which implies that for all $j \neq i$,
        $$
            \hlf_j(y, \tcl) \wedge \hrt_j(y, \tcl) \leq c_j^0 \leq \hlf_j(y, \tcl) \vee \hrt_j(y, \tcl).
        $$
        by \eqref{hlf-def}, \eqref{hrt-def} and \eqref{cj0-eq}. Therefore, we have $\bc^0 \in \Bset(y, \tcl^0)$. Consequently, since $\ltN_i(y, \bx_{-i}^0) = e_i + c_i^0 + \brLM_i^\top \bwNf(\bc^0)$ and $\gtld_i(y) \leq e_i + c_i^0 + \brLM_i^\top \bwNf(\bc^0)$, $\gtld_i(y) \leq \min_{\bx_{-i} \in [\bzero,  \bxbar_{-i} ]} \ltN_i(y, \bx_{-i})$ holds.

        We now show that $\gtld_i(y) \geq \min_{\bx_{-i} \in [\bzero,  \bxbar_{-i} ]} \ltN_i(y, \bx_{-i})$. Fix $\tcl^1 \in \Zset(y)$ and $\bc^1 \in \Bset(y, \tcl^1)$. Assume that $\tcl^1 \neq \clbar/2$. Let $\bx^1 \in \bbR^m$ with $x_i^1 = y$ and
        $$
            x_j^1 = \frac{c_j^1 - \xbar_j Q(\clbar - \tcl^1)}{Q(\tcl^1) - Q(\clbar - \tcl^1)}, \quad  j \neq i.
            \label{xj-equiv-def}
        $$
        It is easy to see that $x_j^1 \in [(\tcl^1 - \sbar_{i, j} - y)^+,   \xbar_j \wedge (\tcl^1 - y)] \subset [0, \xbar_j]$ for $j \neq i$, which implies that $\bx_{-i}^1 \in [\bzero,  \bxbar_{-i}]$. Further, it can be verified that $\bc^1 = \bcNf(\bx^1)$ by \eqref{op-cash-inter}. Otherwise, if $\tcl^1 = \clbar/2$, $\clbar/2 - y \leq \clbar - \xbar_i$ since $\tcl^1 \in \Zset(y)$. Then there exists some $\bx^1 \in [\bzero, \bxbar]$ such that $\bone^\top \bx_{-i}^1 = \clbar/2 - y$, $x_i^1 = y$, and $\cNf_j(\bx^1) = \xbar_j Q(\clbar/2) = c_j^1$ for all $j \in [m]$, implying that $\bc^1 = \bcNf(\bx^1)$.

        Hence,  $ \min_{\bx_{-i} \in [\bzero,  \bxbar_{-i} ]} \ltN_i(y, \bx_{-i}) \leq e_i + c_i^1 + \brLM_i^\top \bwNf(\bc^1)$ for any feasible $\tcl^1$ and $\bc^1$, which implies that $\min_{\bx_{-i} \in [\bzero,  \bxbar_{-i} ]} \ltN_i(y, \bx_{-i}) \leq \gtld_i(y)$. As a consequence, \eqref{gi-ltNi-eq} follows.

        \textbf{Step 3.} We first fix
        \begin{equation}
            \tlb \in \argmin_{\tcl \in \Zset\left(\xbar_i/2\right)} \left\{\min_{\cvec \in \Bset\left(\xbar_i/2,   \tcl\right)} \left\{e_i + c_i + \brLM_i^\top \bwNf(\cvec) \right\} \right\},
            \label{tlb-def}
        \end{equation}
        and
        \begin{equation}
            \bctl \in \argmin_{\cvec \in \Bset\left(\xbar_i/2,  \tlb\right)}\left\{e_i + c_i + \brLM_i^\top \bwNf(\cvec) \right\}\subset\Bset\left(\xbar_i/2,  \tlb\right).
            \label{bctl-def}
        \end{equation}
        For any $\bc^1 \in \Bset(\xbar_i/2,  \tcl)$ and $\bc^2 \in \Bset(\xbar_i/2,  \clbar - \tcl)$, we have $\bone^\top \bc^1 = \bone^\top \bc^2$ and $c_i^1 = c_i^2$. Also, note that
        $$\begin{aligned}
            \hlf_j\left(\frac{\xbar_i}{2},  \tcl\right)
            & = \left(\tcl - \sbar_{i, j} - \frac{\xbar_i}{2}\right)^+ \left(Q(\tcl) - Q(\clbar - \tcl) \right)
            + \xbar_j Q(\clbar - \tcl)
            \\ & = \left(\xbar_j - (\clbar - \tcl) + \frac{\xbar_i}{2} \right)^+ \left(Q(\tcl) - Q(\clbar - \tcl) \right)
            + \xbar_j Q(\clbar - \tcl)
            \\ & = \left(\xbar_j \wedge \left(\clbar - \tcl - \frac{\xbar_i}{2}\right) \right) \left(Q(\clbar - \tcl) - Q(\tcl) \right)
             + \xbar_j Q(\tcl)
            \\ & = \hrt_j\left(\frac{\xbar_i}{2},  \clbar - \tcl\right).
        \end{aligned}$$
        Thus, $\Bset(\xbar_i/2,  \tcl) = \Bset(\xbar_i/2,  \clbar - \tcl)$ by definition. Hence, without loss of generality, we assume that
        \begin{equation}
            \tlb \leq \clbar/2.
            \label{tlb-assump}
        \end{equation}

        Fix $x \in [0, \xbar_i/2)$ and let
        \begin{equation}
            \clst = \begin{cases}
                \tlb,  & \text{if } 0 \leq \tlb \leq (\clbar - \xbar_i)/2 + x;  \\
                (\clbar - \xbar_i)/2 + x,  & \text{if } (\clbar - \xbar_i)/2 + x < \tlb \leq \clbar/2.
            \end{cases}
            \label{clst-def}
        \end{equation}
        Since $\tlb \in \Zset(\xbar_i/2)$ by \eqref{tlb-def}, it is easy to check that
        \begin{equation}
            x + \left(\tlb - \frac{\xbar_i}{2}\right) \leq \clst \leq \left(\frac{\clbar - \xbar_i}{2} + x \right) \wedge \tlb.
            \label{clst-ineq}
        \end{equation}

        We next define
        \begin{equation}
            \ettl \coloneqq \tlb Q(\tlb) + (\clbar - \tlb) Q(\clbar - \tlb),
            \label{ettl-def}
        \end{equation}
        and
        \begin{equation}
            \etst \coloneqq \clst Q(\clst) + (\clbar - \clst) Q(\clbar - \clst).
            \label{etst-def}
        \end{equation}
        Since $\clst \leq \tlb \leq \clbar - \tlb \leq \clbar - \clst$ and the mapping $x \mapsto xQ(x)$ is concave by \eqref{useful-ineq-xx}, we obtain
        \begin{equation}
            \etst \leq \ettl.
            \label{csum-lb-concav}
        \end{equation}

        Define $\lid \coloneqq \argmax_{j \neq i}\xbar_j$. By \eqref{clst-ineq}, for all $j \neq i, \lid$, we have $(\tlb - \sbar_{i, j} - \xbar_i/2)^+ \leq (\clst - \sbar_{i, j} - x)^+ \leq (\xbar_j/2 - \sbar_{i, j}/2)^+$. Hence, $\hlf_j(\xbar_i/2,  \tlb) = \xbar_j Q(\clbar - \tlb)$ and $\hlf_j(x, \clst) = \xbar_j Q(\clbar - \clst)$ for all $j \neq i, \lid$. This implies that
        \begin{equation}
            \hlf_j(x, \clst) \leq \hlf_j\left(\frac{\xbar_i}{2},  \tlb \right), \quad  \text{for all } j \neq i, \lid,
            \label{cj-lb-concav}
        \end{equation}
        since $\clst \leq \tlb$ and $Q(x)$ is decreasing in $x$ by Assumption \ref{assump1}(ii). Further, we claim that
        \begin{equation}
            \hrt_j\left(\frac{\xbar_i}{2},  \tlb \right) \leq \hrt_j(x, \clst), \quad  \text{for all } j \neq i.
            \label{cj-ub-decreQ}
        \end{equation}
        For each $j \neq i$, we observe the following relationship
        \begin{equation}
        \begin{aligned}
            & \hrt_j\left(x, \clst\right) - \hrt_j\left(\frac{\xbar_i}{2},  \tlb \right)			\nonumber
            \\ &= \left(\xbar_j \wedge \left(\clst - x \right) \right) \left(Q\left(\clst \right) - Q\left(\clbar - \clst \right) \right)	\notag
             + \xbar_j \left(Q(\clbar - \clst) - Q(\clbar - \tlb) \right)		\nonumber
             \\&\qquad- \left(\xbar_j \wedge \left(\tlb - \frac{\xbar_i}{2} \right) \right) \left(Q\left(\tlb\right) - Q\left(\clbar - \tlb\right) \right)				\nonumber
            \\ &\geq \left(\xbar_j \wedge \left(\clst - x \right) \right) \left(Q\left(\clst \right) - Q\left(\tlb\right) \right)
            + \left(\xbar_j - \clst + x \right)^+ \left(Q\left(\clbar - \clst \right) - Q\left(\clbar - \tlb\right) \right) 	\notag
            \\ &\geq \left(\xbar_j \wedge \left(\clst - x \right) - \left(\xbar_j - \clst + x \right)^+ \right)	 \left(Q\left(\clst \right) - Q\left(\tlb\right) \right),			
            \label{hrt-comp-st-tl}
        \end{aligned}
        \end{equation}
        where the first inequality holds by \eqref{clst-ineq} and the fact that $Q(\tlb) \geq Q(\clbar - \tlb)$, and the second inequality is obtained since $Q(\clbar - \clst) - Q(\clbar - \tlb) \geq Q(\tlb) - Q(\clst)$ by Assumption \ref{assump1}(ii) and (iv). Moreover, the right-hand side of the second inequality is nonnegative because $\clst = \tlb $ if $\clst < x + (\clbar - \xbar_i)/2$ and $\xbar_j \wedge (\clst - x ) \geq (\xbar_j - \clst + x )^+$ otherwise. This establishes \eqref{cj-ub-decreQ}.

        Define
        \begin{equation}
            d_i \coloneqq x Q(\clst) + (\xbar_i - x) Q(\clbar - \clst),
            \label{di-def}
        \end{equation}
        and
        \begin{equation}
            d_\lid \coloneqq \hlf_\lid(x, \clst) \vee \ctl_\lid.
            \label{dl-def}
        \end{equation}
        Then by \eqref{etst-def} and \eqref{di-def}, we can obtain that
        \begin{equation}
            \begin{aligned}
                & \etst - d_i
                \\ & = \left(\clst - x \right) \left(Q(\clst) - Q(\clbar - \clst) \right) + (\clbar - \xbar_i) Q(\clbar - \clst)		
                \\ & = \left(\left(\clst - \sbar_{i, \lid} - x \right)^+ + \sbar_{i, \lid} \wedge \left(\clst - x \right) \right) \left(Q(\clst) - Q(\clbar - \clst) \right)
                + \left(\xbar_\lid + \sbar_{i, \lid} \right) Q(\clbar - \clst)			
                \\ & = \hlf_\lid(x, \clst) + \sbar_{i, \lid} Q(\clbar - \clst)	
                + \left(\sbar_{i, \lid} \wedge \left(\clst - x \right)\right) \left(Q(\clst) - Q(\clbar - \clst) \right)		
                \\ &   \leq \hlf_\lid(x, \clst) + \sbar_{i, \lid} Q(\clbar - \clst)	
                + \sum_{j \neq i, \lid}\left(\xbar_j \wedge \left(\clst - x \right) \right) \left(Q(\clst) - Q(\clbar - \clst) \right)
                \\ & = \hlf_\lid(x, \clst) + \sum_{j \neq i, \lid} \hrt_j(x, \clst).
            \end{aligned}
            \label{etst-di-ub}
        \end{equation}
        Furthermore, since
        $$
            \begin{aligned}
                \sum_{j \neq i, \lid}(\clst - \sbar_{i, j} - x)^+ & \leq \sum_{j \neq i, \lid}\left(\clht - \sbar_{i, j} - x \right)^+ + \clst - \clht
                 = \clst - x - \xbar_\lid
            \end{aligned}
        $$
        and $\clht = \clbar - \xbar_i + x \geq \clst$, we have
        \begin{equation}
            \begin{aligned}
                & \etst - d_i
                \\ & = \left(\clst - x \right) \left(Q(\clst) - Q(\clbar - \clst) \right) + (\clbar - \xbar_i) Q(\clbar - \clst)				
                \\ & = \left[\xbar_\lid \wedge (\clst - x) + \left(\clst - x - \xbar_\lid \right)^+ \right] \left(Q(\clst) - Q(\clbar - \clst) \right)
                + \left(\xbar_\lid + \sbar_{i, \lid} \right) Q(\clbar - \clst)
                \\ & = \hrt_\lid(x, \clst) + \sbar_{i, \lid} Q(\clbar - \clst)	
                + \left(\clst - x - \xbar_\lid \right)^+ \left(Q(\clst) - Q(\clbar - \clst)  \right)		
                \\ &   \geq \hrt_\lid(x, \clst) + \sbar_{i, \lid} Q(\clbar - \clst)
                + \sum_{j \neq i, \lid}\left(\clst - \sbar_{i, j} - x \right)^+ \left(Q(\clst) - Q(\clbar - \clst)  \right)
                \\ & = \hrt_\lid(x, \clst) + \sum_{j \neq i, \lid} \hlf_j(x, \clst).
            \end{aligned}
            \label{etst-di-lb}
        \end{equation}

        % Lastly, we establish the following result:
        % \begin{lemma}
        %     \label{lemma-B1}
        %     For fixed $i \in [m]$, let $\lid = \argmax_{j \neq i} \xbar_j$. Define $\bctl$, $d_i$, and $d_\lid$ as in \eqref{bctl-def}, \eqref{di-def}, and \eqref{dl-def}, respectively. If \eqref{clst-def} holds, then $\ctl_i - d_i > 0 $ and
        %         $\ctl_i - d_i \geq d_\lid - \ctl_\lid.$
        % \end{lemma}
        % The proof of the lemma is deferred to the end of this section.

        \textbf{Step 4.}  Recall that $d_i$ and $d_\lid$ are defined in \eqref{di-def} and \eqref{dl-def}, respectively. We define
        \begin{equation}
            \etdlt \coloneqq \etst - d_i - d_\lid - \sum_{j \neq i, \lid} \ctl_j,
            \label{etdlt-def}
        \end{equation}
        and for all $j \neq i, \lid$,
        \begin{equation}
            d_j \coloneqq \begin{cases}
                \ctl_j + \dfrac{\hrt_j(x, \clst) - \ctl_j}{\sum_{k \neq i, \lid}(\hrt_k(x, \clst) - \ctl_k)}\etdlt,  & \text{if } \etdlt > 0; \\[15pt]
                \ctl_j + \dfrac{\ctl_j - \hlf_j(x, \clst)}{\sum_{k \neq i, \lid}(\ctl_k - \hlf_k(x, \clst))}\etdlt,  & \text{otherwise}.
            \end{cases}
            \label{dj-def}
        \end{equation}
        We first claim that
        \begin{equation}
            \bd = (d_1, \dots, d_m)^\top \in \Bset(x, \clst).
            \label{bd-fsb}
        \end{equation}
        Observe that
        $$
            \bone^\top \bd = \sum_{j \neq i, \lid} d_j + d_i + d_\lid = \sum_{j \neq i, \lid} \ctl_j + \etdlt + d_i + d_\lid = \etst.
        $$
        Furthermore, $\clst \leq \clbar/2$ by \eqref{clst-def}. Thus, it is enough to show that $\hlf_j(x, \clst) \leq d_j \leq \hrt_j(x, \clst)$ for all $j \neq i$. Note that $d_\lid \geq \hlf_\lid(x, \clst)$ by \eqref{dl-def}. Also, since $\bctl \in \Bset(\xbar_i/2, \tlb)$, \eqref{cj-lb-concav} and \eqref{cj-ub-decreQ} imply that
        \begin{equation}
            \ctl_\lid \leq \hrt_\lid(x, \clst),
            \label{ctll-ub-clst}
        \end{equation}
        and
        \begin{equation}
            \hlf_j(x, \clst) \leq \ctl_j \leq \hrt_j(x, \clst), \quad  \text{for all } j \neq i, \lid.
            \label{ctlj-lb-ub-clst}
        \end{equation}
        Additionally, by \eqref{dl-def}, \eqref{etst-di-ub}, \eqref{etst-di-lb}, and \eqref{ctll-ub-clst}, it is straightforward to see that
        \begin{equation}
            \begin{cases}
                \etdlt \leq \sum_{k \neq i, \lid}(\hrt_k(x, \clst) - \ctl_k),  & \text{if } \etdlt > 0; \\
                \etdlt \geq -\sum_{k \neq i, \lid}(\ctl_k - \hlf_k(x, \clst)),  & \text{otherwise.}
            \end{cases}
            \label{etdlt-ineq}
        \end{equation}
        Accordingly, $\hlf_j(x, \clst) \leq d_j \leq \hrt_j(x, \clst)$ for all $j \neq i$, and hence, \eqref{bd-fsb} holds.

        Next, we want to show that
        \begin{equation}
            \ctl_i + \brLM_i^\top \bwNf(\bctl) > d_i + \brLM_i^\top \bwNf(\bd).
            \label{d-c-obj-comp}
        \end{equation}
        To that end, we use the following lemma whose proof is deferred to Appendix~\ref{app:auxiliary}.
        \begin{lemma}
            \label{lemma-B1}
            For fixed $i \in [m]$, let $\lid = \argmax_{j \neq i} \xbar_j$. Define $\bctl$, $d_i$, and $d_\lid$ as in \eqref{bctl-def}, \eqref{di-def}, and \eqref{dl-def}, respectively. If \eqref{clst-def} holds, then $\ctl_i - d_i > 0 $ and
                $\ctl_i - d_i \geq d_\lid - \ctl_\lid.$
        \end{lemma}

        If $\bwNf(\bctl) = \bwNf(\bd)$, \eqref{d-c-obj-comp} holds by Lemma \ref{lemma-B1}. Otherwise, let $\bcht \coloneqq \bd \wedge \bctl$. Then, $\cht_\lid = \ctl_\lid$ by \eqref{dl-def}, $\cht_i = d_i$ by Lemma \ref{lemma-B1}, and for all $j \neq i, \lid$,
        \begin{equation}
            \cht_j = \begin{cases}
                \ctl_j, \quad \quad & \text{if } \etdlt > 0; \\
                d_j, \quad \quad & \text{otherwise.}
            \end{cases}
            \label{chtj-def}
        \end{equation}
        by \eqref{dj-def} and \eqref{ctlj-lb-ub-clst}. Therefore, we have
        \begin{align}
            & \ctl_i + \brLM_i^\top \bwNf(\bctl) - \left(d_i + \brLM_i^\top \bwNf(\bd) \right)		\notag
            \\ & = \ctl_i - d_i + \brLM_i^\top \left(\bwNf(\bctl) - \bwNf(\bd) \right)			\nonumber
            \\ & \geq \ctl_i - d_i + \brLM_i^\top \left(\bwNf(\bcht) - \bwNf(\bd) \right)
            \label{gtld-diff-chat}
            \\ & > \ctl_i - d_i - \left\|\bwNf(\bcht) - \bwNf(\bd) \right\|_1
            \label{gtld-diff-noA}
            \\ & \geq \ctl_i - d_i - \left\|\bcht - \bd \right\|_1
            \label{gtld-diff-nonexp}
            \\ & = \ctl_i - d_i + \sum_{j \neq i} \left(\cht_j - d_j \right), 		\nonumber
            \\ & = \begin{cases}
                \ettl - \etst, \quad \quad & \text{if } \etdlt > 0; \\
                \ctl_i - d_i - (d_\lid - \ctl_\lid), \quad \quad & \text{otherwise},
            \end{cases}
            \label{gtld-diff-sum}
        \end{align}
        where \eqref{gtld-diff-chat} holds since $\brLM_i$ is nonnegative and $\bwNf(\cdot)$ is increasing~\citep[see][Lemma 5]{EisenbergNoe2001}, \eqref{gtld-diff-noA} is satisfied since $A_{ki} \leq \sum_{j = 1}^{m}A_{kj}  < 1$ for all $k \in [m]$,  \eqref{gtld-diff-nonexp} stems from the fact that $\bwNf(\cdot)$ is $\ell_1$-nonexpansive~\citep[see][Lemma 5]{EisenbergNoe2001}, and \eqref{gtld-diff-sum} follows since we have $\bone^\top \bctl = \ettl$ and $\bone^\top \bd = \etst$. Consequently, by \eqref{csum-lb-concav} and Lemma \ref{lemma-B1}, \eqref{d-c-obj-comp} follows.

        Note that $\tlb \geq \xbar_i/2$ since $\tlb \in \Zset(\xbar_i/2)$. Thus, $\clst \in \Zset(x)$ by \eqref{clst-ineq}, and we finally have
        \begin{align}
            \gtld_i\left(\frac{\xbar_i}{2} \right) & = e_i + \ctl_i + \brLM_i^\top \bwNf(\bctl)
            \label{gtldi-eq}
            \\ & > e_i + d_i + \brLM_i^\top \bwNf(\bd)			\nonumber
            \\ & \geq \min_{\tcl \in \Zset(x)} \min_{\cvec \in \Bset(x, \tcl)} \left\{e_i + c_i + \brLM_i^\top \bwNf(\cvec) \right\}
            \label{di-min-Bset-ineq}
            \\ & = \gtld_i(x),
            \label{gtld-x-eq}
        \end{align}
        where \eqref{gtldi-eq} stems from \eqref{gi-ltNi-eq}, \eqref{tlb-def} and \eqref{bctl-def}, and \eqref{di-min-Bset-ineq} holds since $\clst \in \Zset(x)$ and $\bd \in \Bset(x, \clst)$. Then, by \eqref{gi-ltNi-eq}, we have \eqref{ltN-min-comp}, and the proof is complete.
    \endproof

\section{Proofs for Auxiliary Results}\label{app:auxiliary}
    \proof{Proof of Lemma \ref{lemma:psA-pbar-comp}.}
    Assume by contradiction that $\pbar > p_2^\brLM$. For $k=0,1,2,\ldots,$ define
    \begin{equation}
        \rho_{k+1}^\dagger = Q\left(\sum_{j = 1}^{m} \frac{\left(L_j - e_j - \brLM^\top \blambda_k^\dagger\right)^+}{\rho_k^\dagger } \wedge z_j \right)
        ~~\text{and}~~
        \blambda_{k+1}^\dagger = \bL \wedge \left(\bgao + \bz \rho_k^\dagger + \brLM^\top \blambda_k^\dagger \right),
    \end{equation}
    where
    \begin{equation}
            \rho_0^\dagger \coloneqq Q\left(\sum_{j = 1}^{m} \frac{\left(L_j - e_j - \brLM^\top \brpA\right)^+}{p_2^\brLM } \wedge z_j \right)
            ~~\text{and}~~
            \blambda_0^\dagger \coloneqq \bL \wedge \left(\bgao + \bz \rho_0^\dagger + \brLM^\top \brpA \right).
    \end{equation}
    Fix $\bx\leq\bxbar$. Then, $p_1 \coloneqq Q(\bone^\top \bx) \geq Q(\bone^\top \bxbar) = \pbar > p_2^\brLM$. Moreover, by \eqref{eq:p2-l-rLM} and Assumption~\ref{assump1}(ii), we have
    \begin{align*}
        p_2^\brLM & = Q\left(\sum_{j = 1}^{m} \frac{\left(L_j - e_j - x_{j}p_1 - \brLM^\top \brpA\right)^+}{p_2^\brLM } \wedge \left(z_j - x_{j}\right) \right)
        \\ & \geq Q\left(\sum_{j = 1}^{m} \frac{\left(L_j - e_j - \brLM^\top \brpA\right)^+}{p_2^\brLM } \wedge z_j \right)
        \\ & = \rho_0^\dagger.
    \end{align*}
    Consequently, we observe that
        $\brpA = \bL \wedge \left(\bgao + \bx p_1 + (\bz - \bx) p_2^\brLM + \brLM^\top \brpA \right) \geq \bL \wedge \left(\bgao + \bz \rho_0^\dagger + \brLM^\top \brpA \right) = \blambda_0^\dagger$,

    If $\rho_k^\dagger \leq p_2^\brLM$ and $\blambda_k^\dagger \leq \brpA$ for some $k \geq 0$, then
    \begin{equation}
        \rho_{k+1}^\dagger = Q\left(\sum_{j = 1}^{m} \frac{\left(L_j - e_j - \brLM^\top \blambda_k^\dagger\right)^+}{\rho_k^\dagger } \wedge z_j \right) \leq Q\left(\sum_{j = 1}^{m} \frac{\left(L_j - e_j - x_{j}Q(\bone^\top \bx) - \brLM^\top \brpA\right)^+}{p_2^\brLM } \wedge \left(z_j - x_{j}\right) \right) = p_2^\brLM
    \end{equation}
    and thus
    \begin{equation}
        \blambda_{k+1}^\dagger = \bL \wedge \left(\bgao + \bz \rho_k^\dagger + \brLM^\top \blambda_k^\dagger \right) \leq \bL \wedge \left(\bgao + \bx p_1 + (\bz - \bx) p_2^\brLM + \brLM^\top \brpA \right) = \brpA.
    \end{equation}
    Hence, by induction, we obtain  $\rho_k^\dagger \leq p_2^\brLM$ and $\blambda_{k} \leq \brpA$ for all $k$, suggesting that $\pbar \leq p_2^\brLM$ and $\blbar \leq \brpA$ since $\rho_k^\dagger \rightarrow \pbar$ and $\blambda_k^\dagger \rightarrow \blbar$ as $k$ increases. This contradicts the assumption and indicates that $\pbar \leq p_2^\brLM$.
    \endproof
    %=================================================================================================
    \proof{Proof of Lemma \ref{lemma-B1}.}
        We first let $\delta = \tlb - \clst$. Then, by \eqref{clst-def},
        \begin{equation}
            0 \leq \delta \leq \frac{\xbar_i}{2} - x,
            \label{dlt1-ineq}
        \end{equation}
        and we observe that
        \begin{equation}
        \begin{aligned}
            \ctl_i - d_i & = \frac{\xbar_i}{2} \left(Q(\tlb) + Q(\clbar - \tlb) \right)		
             - xQ(\clst) - (\xbar_i - x) Q(\clbar - \clst)	
            \\ & = \left(x + \delta \right) Q\left(\clst + \delta \right) - x Q\left(\clst\right)
             - \left(\xbar_i - x \right) Q\left(\clbar - \clst \right)		
            \\ & \quad  + \left(\xbar_i - x - \delta \right) Q\left(\clbar - \clst - \delta \right)	
             + \left(\frac{\xbar_i}{2} - x - \delta \right) \left(Q\left(\tlb \right) - Q\left(\clbar - \tlb \right) \right)	
        \end{aligned}\label{ctli-di-diff}	
        \end{equation}
        If $\delta = 0$, then $\tlb < \clbar/2$ by \eqref{clst-def}. Thus, $Q(\tlb) > Q(\clbar - \tlb)$ by Assumption \ref{assump1}(ii), and the right-hand side of \eqref{ctli-di-diff} is strictly positive. On the other hand, if $\delta > 0$, by the mean value theorem, there exists some $\theta_1 \in (x,  x + \delta)$ such that
        $$\begin{aligned}
            &  (\textrm{The right-hand side of}~\eqref{ctli-di-diff})
            \\ &= \delta \bigg( f_x(x + \theta_1,  \clst - x)
             -  f_x(\xbar_i - x - \theta_1,  \clbar - \clst - \xbar_i + x) \bigg)
             + \left(\frac{\xbar_i}{2} - x - \delta \right) \left(Q\left(\tlb \right) - Q\left(\clbar - \tlb \right) \right)
            \\ & > \delta \bigg(f_x(x + \theta_1,  \clst - x) \bigg.
             - f_x(\xbar_i - x - \theta_1,  \clbar - \clst - \xbar_i + x) \bigg)
            \\ & \geq 0,
        \end{aligned}$$
        where the first inequality stems from \eqref{useful-ineq-xx} and \eqref{dlt1-ineq}, and the second inequality is satisfied since $\clst - x \leq \clbar - \clst - \xbar_i + x$ and \eqref{useful-ineq-xy} holds. This proves the first statement of the lemma.

        For the second statement, we observe that $\ctl_j \geq \hlf_j(\xbar_i/2,  \tlb)$ for all $j \neq i$ by \eqref{bctl-def}, implying that $d_\lid - \ctl_\lid \leq (\hlf_\lid(x, \clst) - \hlf_\lid(\xbar_i/2,  \tlb))^+$. Hence,  it suffices to show that
        % \begin{equation}
            $\ctl_i - d_i \geq \hlf_\lid(x, \clst) - \hlf_\lid(\xbar_i/2, \tlb ).$
            % \label{ctli-di-hlf-comp}
        % \end{equation}
        If $\delta=0$, a straightforward calculation yields that this inequality holds. Assume that $\delta>0$. According to \eqref{dlt1-ineq}, we have $\delta \leq \vthtl - \vthst$, where $\vthtl \coloneqq (\tlb - \sbar_{i, \lid}) \vee (\xbar_i/2)$ and $\vthst \coloneqq (\clst - \sbar_{i, \lid}) \vee x$. Thus, there exists $\theta_2 \in (0, \delta)$ such that
        \begin{equation}
        \begin{aligned}
            & \ctl_i - d_i - \left(\hlf_\lid(x, \clst) - \hlf_\lid\left(\frac{\xbar_i}{2}, \tlb \right) \right)	
            \\ & = \vthtl \left(Q\left(\tlb\right) - Q\left(\clbar - \tlb \right) \right)			
            + \left(\xbar_i + \left(\xbar_\lid - \sbar_{i, \lid} \right)^+ \right) \left(Q\left(\clbar - \tlb \right)   -Q\left(\clbar - \clst \right) \right)
            \\ & \qquad  - \vthst \left(Q\left(\clst \right) - Q\left(\clbar - \clst \right) \right)		  + \left(\xbar_\lid \wedge \sbar_{i, \lid} \right) \left( Q\left(\clbar - \tlb \right) - Q\left(\clbar - \clst \right) \right)		
            \\ &= \left(\vthst + \delta \right) \left(Q\left(\clst + \delta \right) - Q\left(\clbar - \clst - \delta \right) \right)  + \left(\xbar_i + \left(\xbar_\lid - \sbar_{i, \lid} \right)^+ \right) \left(Q\left(\clbar - \tlb \right) - Q\left(\clbar - \clst \right)\right)
            \\ & \qquad  - \vthst \left(Q\left(\clst \right) - Q\left(\clbar - \clst \right)  \right)	
             + \left(\vthtl - \vthst - \delta \right) \left(Q\left(\tlb\right) - Q\left(\clbar - \tlb \right) \right) 			
             \\&\qquad+ \left(\xbar_\lid \wedge \sbar_{i, \lid}\right)\left(Q\left(\clbar - \tlb\right) - Q\left(\clbar - \clst \right) \right)		
            \\ & \geq \delta \Big\{Q\left(\clst + \theta_2 \right) - Q\left(\clbar - \clst - \theta_2 \right)    - \left(\xbar_i + \left(\xbar_\lid - \sbar_{i, \lid} \right)^+ \right) Q'\left(\clbar - \clst - \theta_2 \right)		
            \\ & \qquad   + \left(\vthst + \theta_2 \right) \left(Q'\left(\clst + \theta_2 \right) + Q'\left(\clbar - \clst - \theta_2 \right) \right)\Big\}		
            \\ & = \delta \left\{ f_x(\vthst + \theta_2,  \clst - \vthst)  - f_x\left(\xbar_i + \left(\xbar_\lid - \sbar_{i, \lid} \right)^+ - \vthst - \theta_2,  \clbar - \clst - \xbar_i - \left(\xbar_\lid - \sbar_{i, \lid} \right)^+ + \vthst \right) \right\}	
            \\ &  > \delta \Big\{ f_x\left(\xbar_i + \left(\xbar_\lid - \sbar_{i, \lid} \right)^+ - \vthst - \theta_2,  \clst - \vthst \right)
            \\ & \qquad  - f_x\left(\xbar_i + \left(\xbar_\lid - \sbar_{i, \lid} \right)^+ - \vthst - \theta_2,  \clbar - \clst - \xbar_i - \left(\xbar_\lid - \sbar_{i, \lid} \right)^+ + \vthst \right) \Big\}	
            \\ &= 0,
            \label{ci-di-comp-xy}
        \end{aligned}
        \end{equation}
        where the first inequality holds by the mean value theorem and from the fact that $Q(\tlb) > Q(\clbar - \tlb)$ and $\tlb > \clst$. Since $\clst = (\clbar - \xbar_i)/2 + x$ by \eqref{clst-ineq} and $\theta_2 < \xbar_i/2 - x$ by \eqref{dlt1-ineq}, we have
        $$
            \begin{aligned}
                2 \left((\clst - \sbar_{i, \lid}) \vee x \right)  \leq 2 \left\{\left(\frac{\clbar - \xbar_i}{2} + x - \sbar_{i, \lid} \right) \vee x \right\}
                 \leq \left(\xbar_\lid - \sbar_{i, \lid}\right)^+ + 2x
                 < \left(\xbar_\lid - \sbar_{i, \lid}\right)^+ + \xbar_i - 2\theta_2,
            \end{aligned}
        $$
        and thus, the second inequality of \eqref{ci-di-comp-xy} follows from \eqref{useful-ineq-xx}. Similarly, we obtain
        $$
            \begin{aligned}
                2\left(\clst - \left(\clst - \sbar_{i, \lid} \right) \vee x \right)  = \left(2\sbar_{i, \lid}  \right) \wedge (2\clst - 2x)
                 = \left(2\sbar_{i, \lid}  \right) \wedge (\clbar - \xbar_i)
                 = \clbar - \xbar_i - (\xbar_\lid - \sbar_{i, \lid})^+,
            \end{aligned}
        $$
        which implies that the last equality of \eqref{ci-di-comp-xy} holds. This completes the proof.
    \endproof

\bibliographystyle{abbrvnat}
\bibliography{references}

@article{AcemogluEtAl2015,
  title = {Systemic Risk and Stability in Financial Networks},
  author = {Acemoglu, Daron and Ozdaglar, Asuman and Tahbaz-Salehi, Alireza},
  year = {2015},
  journal = {American Economic Review},
  volume = {105},
  number = {2},
  pages = {564--608},
  publisher = {American Economic Association 2014 Broadway, Suite 305, Nashville, TN 37203},
  issue = {2}
}

@article{AcharyaEtAl2011,
  title = {Rollover Risk and Market Freezes},
  author = {Acharya, Viral V and Gale, Douglas and Yorulmazer, Tanju},
  year = {2011},
  journal = {The Journal of Finance},
  volume = {66},
  number = {4},
  pages = {1177--1209},
  publisher = {Wiley Online Library},
}

@article{Ahn2019,
	title={Optimal Intervention under Stress Scenarios: A Case of the {K}orean Financial System},
	author={Ahn, Dohyun and Kim, Kyoung-Kuk},
	journal={Operations Research Letters},
	volume={47},
	number={4},
	pages={257--263},
	year={2019},
	publisher={Elsevier}
}

@article{Ahn2020,
	title={Shock Amplification in Financial Networks with Applications to the {CCP} Feasibility},
	author={Ahn, Dohyun},
	journal={Quantitative Finance},
	volume={20},
	number={7},
	pages={1045--1056},
	year={2020},
	publisher={Taylor \& Francis}
}

@article{AhnEtAl2023,
	title={Multivariate Stress Scenario Selection in Interbank Networks},
	author={Ahn, Dohyun and Kim, Kyoung-Kuk and Kwon, Eunji},
	journal={Journal of Economic Dynamics and Control},
	volume={154},
	pages={104712},
	year={2023},
	publisher={Elsevier}
}

@article{AminiEtAl2016c,
  title = {Uniqueness of Equilibrium in a Payment System with Liquidation Costs},
  author = {Amini, Hamed and Filipović, Damir and Minca, Andreea},
  year = {2016},
  journal = {Operations Research Letters},
  volume = {44},
  number = {1},
  pages = {1--5},
  publisher = {Elsevier},
  issue = {1}
}

@article{AminiEtAl2020,
  title = {Systemic Risk in Networks with a Central Node},
  author = {Amini, Hamed and Filipović, Damir and Minca, Andreea},
  year = {2020},
  journal = {SIAM Journal on Financial Mathematics},
  volume = {11},
  number = {1},
  pages = {60--98},
  publisher = {Society for Industrial and Applied Mathematics},
  issue = {1}
}

@article{Amir2005,
  title = {Supermodularity and Complementarity in Economics: An Elementary Survey},
  author = {Amir, Rabah},
  year = {2005},
  journal = {Southern Economic Journal},
  volume = {71},
  number = {3},
  pages = {636--660},
  publisher = {Wiley Online Library},
}

@article{BanerjeeFeinstein2021,
  title = {Price Mediated Contagion through Capital Ratio Requirements with {VWAP} Liquidation Prices},
  author = {Banerjee, Tathagata and Feinstein, Zachary},
  year = {2021},
  journal = {European Journal of Operational Research},
  volume = {295},
  number = {3},
  pages = {1147--1160},
  publisher = {Elsevier},
  issue = {3}
}

@article{BanerjeeFeinstein2022,
  title = {Pricing of Debt and Equity in a Financial Network with Comonotonic Endowments},
  author = {Banerjee, Tathagata and Feinstein, Zachary},
  year = {2022},
  journal = {Operations Research},
  volume = {70},
  number = {4},
  pages = {2085--2100},
  publisher = {INFORMS}
}

@article{BaruccaEtAl2020,
  title = {Network Valuation in Financial Systems},
  author = {Barucca, Paolo and Bardoscia, Marco and Caccioli, Fabio and D'Errico, Marco and Visentin, Gabriele and Caldarelli, Guido and Battiston, Stefano},
  year = {2020},
  journal = {Mathematical Finance},
  volume = {30},
  number = {4},
  pages = {1181--1204},
  publisher = {Wiley Online Library},
  issue = {4}
}

@article{BertsimasLo1998,
  title={Optimal Control of Execution Costs},
  author={Bertsimas, Dimitris and Lo, Andrew W},
  journal={Journal of Financial Markets},
  volume={1},
  number={1},
  pages={1--50},
  year={1998},
  publisher={Elsevier}
}

@article{BichuchFeinstein2019,
  title = {Optimization of Fire Sales and Borrowing in Systemic Risk},
  author = {Bichuch, Maxim and Feinstein, Zachary},
  year = {2019},
  journal = {SIAM Journal on Financial Mathematics},
  volume = {10},
  number = {1},
  pages = {68--88},
  publisher = {SIAM},
  issue = {1}
}

@article{BraouezecWagalath2019,
  title = {Strategic Fire-Sales and Price-Mediated Contagion in the Banking System},
  author = {Braouezec, Yann and Wagalath, Lakshithe},
  year = {2019},
  journal = {European Journal of Operational Research},
  volume = {274},
  number = {3},
  pages = {1180--1197},
  publisher = {Elsevier},
  issue = {3}
}

@article{CapponiBernard2022,
  author = {Bernard, Benjamin and Capponi, Agostino and Stiglitz, Joseph E.},
title = {Bail-Ins and Bailouts: Incentives, Connectivity, and Systemic Stability},
journal = {Journal of Political Economy},
volume = {130},
number = {7},
pages = {1805--1859},
year = {2022}
}

@article{CapponiEtAl2016,
  title={Liability Concentration and Systemic Losses in Financial Networks},
  author={Capponi, Agostino and Chen, Peng-Chu and Yao, David D},
  journal={Operations Research},
  volume={64},
  number={5},
  pages={1121--1134},
  year={2016},
  publisher={INFORMS}
}

@article{CapponiEtAl2020b,
  title={Swing Pricing for Mutual Funds: Breaking the Feedback Loop between Fire Sales and Fund Redemptions},
  author={Capponi, Agostino and Glasserman, Paul and Weber, Marko},
  journal={Management Science},
  volume={66},
  number={8},
  pages={3581--3602},
  year={2020},
  publisher={INFORMS}
}

@article{ChenEtAl2016,
  title = {An Optimization View of Financial Systemic Risk Modeling: Network Effect and Market Liquidity Effect},
  author = {Chen, Nan and Liu, Xin and Yao, David D},
  year = {2016},
  journal = {Operations Research},
  volume = {64},
  number = {5},
  pages = {1089--1108},
  publisher = {INFORMS},
  issue = {5}
}

@article{CifuentesEtAl2005,
  title = {Liquidity Risk and Contagion},
  author = {Cifuentes, Rodrigo and Ferrucci, Gianluigi and Shin, Hyun Song},
  year = {2005},
  journal = {Journal of the European Economic Association},
  volume = {3},
  number = {2-3},
  pages = {556--566},
  publisher = {Oxford University Press},
  issue = {2-3}
}

@article{ClercEtAl2016,
  title = {Indirect Contagion: The Policy Problem},
  author = {Clerc, Laurent and Giovannini, Alberto and Langfield, Sam and Peltonen, Tuomas and Portes, Richard and Scheicher, Martin},
  year = {2016},
  journal = {ESRB: Occasional Paper Series},
  number = {2016/09}
}

@article{ContSchaanning2019,
  title = {Monitoring Indirect Contagion},
  author = {Cont, Rama and Schaanning, Eric},
  year = {2019},
  journal = {Journal of Banking \& Finance},
  volume = {104},
  pages = {85--102},
  publisher = {Elsevier},
}

@article{DuarteEisenbach2021,
  title = {Fire‐sale Spillovers and Systemic Risk},
  author = {Duarte, Fernando and Eisenbach, Thomas M},
  year = {2021},
  journal = {The Journal of Finance},
  volume = {76},
  number = {3},
  pages = {1251--1294},
  publisher = {Wiley Online Library},
  issue = {3}
}

@article{EisenbergNoe2001,
  title = {Systemic Risk in Financial Systems},
  author = {Eisenberg, Larry and Noe, Thomas H.},
  year = {2001},
  journal = {Management Science},
  volume = {47},
  number = {2},
  pages = {236--249},
  issue = {2},
  langid = {english}
}

@article{Feinstein2017,
  title = {Financial Contagion and Asset Liquidation Strategies},
  author = {Feinstein, Zachary},
  year = {2017},
  journal = {Operations Research Letters},
  volume = {45},
  number = {2},
  pages = {109--114},
  publisher = {Elsevier},
  issue = {2}
}

@article{GatheralSchied2013,
  title={Dynamical Models of Market Impact and Algorithms for Order Execution},
  author={Gatheral, Jim and Schied, Alexander},
  journal={Handbook on Systemic Risk, Jean-Pierre Fouque, Joseph A. Langsam, eds},
  pages={579--599},
  year={2013}
}

@article{GhamamiEtAl2022,
  title = {Collateralized Networks},
  author = {Ghamami, Samim and Glasserman, Paul and Young, H Peyton},
  year = {2022},
  journal = {Management Science},
  volume = {68},
  number = {3},
  pages = {2202--2225},
  publisher = {INFORMS},
  issue = {3}
}

@article{GreenwoodEtAl2015,
  title = {Vulnerable Banks},
  author = {Greenwood, Robin and Landier, Augustin and Thesmar, David},
  year = {2015},
  journal = {Journal of Financial Economics},
  volume = {115},
  number = {3},
  pages = {471--485}
}

@article{KusnetsovMariaVeraart2019,
  title = {Interbank Clearing in Financial Networks with Multiple Maturities},
  author = {Kusnetsov, Michael and Veraart, Luitgard Anna Maria},
  year = {2019},
  journal = {SIAM Journal on Financial Mathematics},
  volume = {10},
  number = {1},
  pages = {37--67},
  publisher = {SIAM},
  issue = {1}
}

@article{PangVeraart2023,
  title = {Assessing and Mitigating Fire Sales Risk under Partial Information},
  author = {Pang, Raymond Ka-Kay and Veraart, Luitgard Anna Maria},
  year = {2023},
  journal = {Journal of Banking \& Finance},
  volume = {155},
  pages = {106989}
}

@article{PangVeraart2025,
  title={Collateralised Networks with Two Interacting Channels of Fire Sales},
  author={Pang, Raymond Ka-Kay and Veraart, Luitgard Anna Maria},
  journal={Available at SSRN 5272239},
  year={2025}
}

@book{Rockafellar2015,
    title = {Convex Analysis},
    author = {Rockafellar, Ralph Tyrell},
    year = {2015},
    address = {Princeton, NJ},
    booktitle = {Convex Analysis},
    isbn = {0-691-01586-4},
    publisher = {Princeton University Press},
    series = {Princeton Landmarks in Mathematics and Physics}
}

@article{RogersVeraart2013,
  title = {Failure and Rescue in an Interbank Network},
  author = {Rogers, Leonard CG and Veraart, Luitgard AM},
  year = {2013},
  journal = {Management Science},
  volume = {59},
  number = {4},
  pages = {882--898},
  publisher = {INFORMS},
  issue = {4}
}

@article{Veraart2020,
  title = {Distress and Default Contagion in Financial Networks},
  author = {Veraart, Luitgard Anna Maria},
  year = {2020},
  journal = {Mathematical Finance},
  volume = {30},
  number = {3},
  pages = {705--737},
  publisher = {Wiley Online Library},
  issue = {3}
}

@article{WeberWeske2017,
  title = {The Joint Impact of Bankruptcy Costs, Fire Sales and Cross-Holdings on Systemic Risk in Financial Networks},
  author = {Weber, Stefan and Weske, Kerstin},
  year = {2017},
  journal = {Probability, Uncertainty and Quantitative Risk},
  volume = {2},
  number = {9},
  pages = {1--38},
  publisher = {Springer},
  issue = {1}
}

\end{document}